\documentclass[11pt]{article}
\usepackage[reset,a4paper,vmargin=3.5truecm,hmargin=3truecm]{geometry}

\usepackage{amsmath}
\usepackage{amsfonts,amssymb,amscd,epic,bm,amsthm}
\usepackage{graphicx,}
\usepackage{paralist}
\usepackage{tikz,subcaption}
\usepackage{ctable,tabularx}
\usepackage{hyperref}
\usepackage[shortlabels]{enumitem} 

\usetikzlibrary{shapes.geometric}

\numberwithin{equation}{section}

\DeclareMathAlphabet\matheu{U}{eur}{m}{n}
\DeclareMathAlphabet\matheuscr{U}{eus}{m}{n}
\DeclareMathAlphabet\mathscr{U}{rsfs}{m}{n}
\SetMathAlphabet\matheu{bold}{U}{eur}{b}{n}
\SetMathAlphabet\matheuscr{bold}{U}{eus}{b}{n}

\newcommand{\MSC}[2]{\smallskip\noindent 
\textbf{2010 Mathematics Subject Classification:}
\textit{Primary} #1, \textit{Secondary} #2.}

\newcommand{\keywords}[1]{\smallskip\noindent\textbf{Keywords and phrases:} #1.} 

\newcommand{\Z}{\mathbb{Z}} 
\newcommand{\R}{\mathbb{R}} 
\newcommand{\C}{\mathbb{C}} 

\newcommand{\e}{\varepsilon}

\newcommand{\HRabi}[1]{H_{\text{\upshape Rabi}}^{#1}} 
\newcommand{\tHRabi}[1]{\tilde H_{\text{\upshape Rabi}}^{#1}} 
\newcommand{\cp}[2]{P^{(#1)}_{#2}} 
\newcommand{\ncp}[2]{\matheuscr{P}^{(#1)}_{#2}} 
\newcommand{\Gncp}[1]{\mathscr{P}^{(#1)}} 

\DeclareMathOperator*{\regprod}{\mathchoice%
{\ooalign{\hbox{$\displaystyle\prod$}\crcr\hbox{$\displaystyle\coprod$}}}
{\ooalign{\hbox{$\textstyle\prod$}\crcr\hbox{$\textstyle\coprod$}}}
{\ooalign{\hbox{$\scriptstyle\prod$}\crcr\hbox{$\scriptstyle\coprod$}}}
{\ooalign{\hbox{$\scriptscriptstyle\prod$}\crcr\hbox{$\scriptscriptstyle\coprod$}}}}

\DeclareMathOperator{\diag}{diag}
\DeclareMathOperator{\sgn}{sgn}

\DeclareMathOperator*{\Res}{Res} 

\DeclareMathOperator{\Diag}{Diag}
\DeclareMathOperator{\Spec}{Spec}
\DeclareMathOperator{\Ker}{Ker}

\DeclareMathOperator{\RE}{Re}
\DeclareMathOperator{\IM}{Im}

\DeclareMathOperator{\tridiag}{tridiag}
\newcommand{\Tridiag}[4]{\tridiag\mat{#1 & #2 \\ #3}_{#4}}

\newcommand{\mat}[1]{\begin{bmatrix}#1\end{bmatrix}}

\newcommand{\RpnSymbol}[1]{\bm{\mathrm{#1}}}
\newcommand{\rF}{\RpnSymbol{F}}
\newcommand{\rV}{\RpnSymbol{V}}
\newcommand{\rD}{\RpnSymbol{D}}

\newcommand{\MatSymbol}[1]{\bm{\mathrm{#1}}}
\newcommand{\mA}{\MatSymbol{A}}
\newcommand{\mB}{\MatSymbol{B}}
\newcommand{\mC}{\MatSymbol{C}}
\newcommand{\mD}{\MatSymbol{D}}
\newcommand{\mE}{\MatSymbol{E}}
\newcommand{\mM}{\MatSymbol{M}}
\newcommand{\mI}{\MatSymbol{I}}

\newcommand{\mS}{\MatSymbol{S}}
\newcommand{\mU}{\MatSymbol{U}}
\newcommand{\mV}{\MatSymbol{V}}

\newcommand{\mLV}{\MatSymbol{LV}}
\newcommand{\mRV}{\MatSymbol{RV}}

\newcommand{\mO}{\MatSymbol{O}}

\newcommand{\Bargmann}{\mathcal{B}}
\newcommand{\Ptwo}{\bm{\mathrm{P}}_2}

\newcommand{\floor}[1]{[#1]}
\newcommand{\fract}[1]{\{#1\}}

\newcommand{\QEDhere}{\pushQED{\qed}\qedhere\popQED}

\theoremstyle{plain}
\newtheorem{thm}{Theorem}[section]
\newtheorem{lem}[thm]{Lemma}
\newtheorem{prop}[thm]{Proposition}
\newtheorem{cor}[thm]{Corollary}
\newtheorem{conject}[thm]{Conjecture}

\theoremstyle{definition}
\newtheorem{dfn}{Definition}[section]
\newtheorem{ex}{Example}[section]
\newtheorem{prob}{Problem}[section]

\theoremstyle{remark}
\newtheorem{rem}{Remark}[section]

\begin{document}

\title{Determinant expressions of constraint polynomials and the spectrum of the asymmetric quantum Rabi model}

\pagestyle{myheadings}
\markboth
{K.~Kimoto, ~C.~Reyes-Bustos~ and~M.~Wakayama}
{Determinant expression of constraint polynomials and spectrum of AQRM}


\author{Kazufumi Kimoto \and Cid Reyes-Bustos \and Masato Wakayama}

\maketitle


\begin{abstract}
  The purpose of this paper is to study the exceptional eigenvalues of the asymmetric quantum Rabi models (AQRM),
  specifically, to determine the degeneracy of their eigenstates. Here, the Hamiltonian $\HRabi{\e}$ of the AQRM is defined by adding
  the fluctuation term $\e\sigma_x$, with $\sigma_x$ being the Pauli matrix, to the Hamiltonian of the quantum Rabi model, breaking its $\Z_{2}$-symmetry.
  The spectrum of $\HRabi{\e}$ contains a set of exceptional eigenvalues, considered to be remains of the eigenvalues of the uncoupled
  bosonic mode, which are further classified in two types: Juddian, associated with polynomial eigensolutions, and non-Juddian exceptional.
  We explicitly describe the constraint relations for allowing the model to have exceptional eigenvalues. By studying these relations
  we obtain the proof of the conjecture on constraint polynomials previously proposed by the third author.
  In fact, we prove that the spectrum of the AQRM possesses degeneracies if and only if the parameter $\e$ is a half-integer. Moreover,
  we show that non-Juddian exceptional eigenvalues do not contribute any degeneracy and we characterize
  exceptional eigenvalues by representations of $\mathfrak{sl}_2$. Upon these results, we draw the whole picture of the spectrum of the AQRM.
  Furthermore, generating functions of constraint polynomials from the viewpoint of confluent Heun equations are also discussed.

  \MSC{34L40}{81Q10, 34M05, 81S05}
  
\end{abstract}

\setcounter{tocdepth}{2}
\tableofcontents

\section{Introduction and overview}

In quantum optics, the quantum Rabi model (QRM) \cite{JC1963} describes the simplest interaction between matter and light,
i.e. the one between a two-level atom and photon, a single bosonic mode (see e.g. \cite{bcbs2016,Le2016}).
Actually, it appears ubiquitously in various quantum systems including cavity and circuit quantum electrodynamics,
quantum dots and artificial atoms \cite{YS2018}, with potential applications in quantum information technologies
including quantum cryptography, quantum computing, etc. (see e.g \cite{BMSSRWU2017,HR2008}).
In addition, the fact \cite{W2015IMRN} that the confluent Heun ODE picture of QRM is derived by coalescing two singularities in the Heun picture of the
non-commutative harmonic oscillator (NCHO: \cite{P2010S,PW2001}) strongly suggests the existence of a rich number theoretical structure
behind the QRM, including modular forms, elliptic curves \cite{KW2007}, Ap\'ery-like numbers \cite{KW2006KJM,LOS2016PAMS}
and Eichler cohomology groups \cite{KW2012RIMS,KW2019} through the study of the spectral zeta function \cite{IW2005a,IW2005KJM,O2008RJ}
(see also \cite{Sugi2016} for the spectral zeta function for the QRM). For the reasons above and according to recent development of experimental technology (cf. e.g. \cite{Ni2010}),
lately there has been considerable progress in the investigation of the QRM
not only in theoretical physics and mathematical analysis (cf.  e.g. \cite{HH2012,H2009IUMJ,P2014Milan,Sugi2016})
but also in experimental physics. For instance, there is a proposal to reproduce/realize the quantum Rabi models experimentally \cite{Pe2015SR}
(see also \cite{Lv2017}). In practice, in the weak parameter coupling regime  the Jaynes-Cummings model, the rotating-wave approximation (RWA)
of the QRM \cite{JC1963}, experimentally meets the QRM.
However, this is not the case in the ultra-strong and deep strong coupling regimes, where the RWA, or similar approximations, is no longer suitable
and the full Hamiltonian of the QRM has to be considered (for a review of recent developments, see \cite{Le2016}).
In contrast with the Jaynes-Cummings model, which has a continuous $U(1)$-symmetry, the QRM only has a
$\Z_2$-symmetry (parity). In 2011, paying attention to this $\Z_2$-symmetry, Braak found the analytical solutions of eigenstates (for the non-exceptional type)
and derived the conditions for determining the energy spectrum of the QRM \cite{B2011PRL} (see also \cite{B2013MfI}).
These conditions are described by the so-called $G$-functions in \cite{B2011PRL,B2011PRL-OnlineSupplement}.
Since then, various aspect of the QRM and its generalizations have been discussed widely and intensively,
and developed from the theoretical viewpoint (see \cite{Le2016} and references therein). For instance, for large eigenvalues of the quantum Rabi model a three-term asymptotic formula with an oscillatory term was recently obtained in \cite{BdMZ2019}.

In the present paper, we study the spectrum of the asymmetric quantum Rabi model (AQRM) \cite{Le2016}.
This asymmetric model actually provides a more realistic description of the circuit QED experiments
employing flux qubits than the QRM itself \cite{Ni2010,YS2018}. The Hamiltonian $\HRabi{\e}$ of the AQRM ($\hbar=1$) is given by
\begin{equation}\label{eq:aH}
  \HRabi{\e} = \omega a^\dag a+\Delta \sigma_z +g\sigma_x(a^\dag+a) + \e \sigma_x,
\end{equation}
where $a^\dag$ and $a$ are the creation and annihilation operators of the bosonic mode,
i.e. $[a,\,a^\dag]=1$ and
\(\sigma_x = \begin{bmatrix}
 0 & 1  \\
 1 & 0
\end{bmatrix} \),
\( \sigma_z= \begin{bmatrix}
 1 & 0  \\
 0 & -1
\end{bmatrix}\)
are the Pauli matrices, $2\Delta$ is the energy difference between the two levels,
$g$ denotes the coupling strength between the two-level system and the bosonic mode with frequency $\omega$
(subsequently, we set $\omega=1$ without loss of generality) and \(\e\) is a real parameter.
The Hamiltonian of the ``symmetric'' quantum Rabi model (QRM) is then given by $\HRabi{0}$.
In this respect, the AQRM has been also referred to as the generalized-, biased- or driven QRM
(see, e.g. \cite{B2011PRL,LB2015JPA,Le2016}).

The initial purpose of the present paper is to study the ``exceptional" eigenvalues of the AQRM
and to determine the degeneracy of its eigenstates. Let us first recall the situation for the QRM. In this case, an eigenvalue \(\lambda\) is
called exceptional if \( \lambda = N - g^2\) for a non-negative integer \(N \in \Z_{\geq 0}\). It was shown by Ku\'s \cite{K1985JMP} that the
degeneracy of an eigenstate (i.e. the energy level crossing at the spectral graph) happens in the QRM if and only if the eigenvalue is exceptional
and the corresponding state is essentially described by a polynomial, i.e. a Juddian eigensolution \cite{J1079}. Non-degenerate
exceptional eigenvalues are also present in the spectrum of QRM (cf. \cite{MPS2013,B2013MfI}), and we call these eigenvalues (and the associated
eigensolutions) non-Juddian exceptional. Exceptional eigenvalues, especially Juddian eigenvalues, are considered to be remains
of the eigenvalues of the uncoupled bosonic mode (i.e. the quantum harmonic oscillator).

Similarly, in the AQRM case, an eigenvalue \(\lambda\) is called exceptional if there is an integer \(N \in \Z_{\geq 0}\) such that
\( \lambda\) is of the form  \( \lambda= N \pm \e - g^2\) \cite{LB2015JPA}. An eigenvalue which is not exceptional
is called regular and is always non-degenerate \cite{B2011PRL}. The presence of degeneracy, in other words, a level crossing in the spectral graph,
for the asymmetric model is highly non-trivial. This is because the additional term $\e \sigma_x$ breaks the $\Z_2$-symmetry
which couples the bosonic mode and the two-level system by allowing spontaneous tunneling between the two atomic states.
The AQRM has been studied, for instance, numerically in the context of the process of physical bodies reaching thermal
equilibrium through mutual interaction \cite{La2013}. In recent works \cite{EJ2017,SK2017} on the AQRM and the quantum
Rabi-Stark model (another generalization of the QRM) the respective authors have studied the degeneracies of the spectrum from different
points of view than the present work.

Without the $\Z_2$-symmetry, there seems to be no invariant subspaces whose respective spectral graphs intersect to create
``accidental'' degeneracies in the spectrum for specific values of the coupling. However, the presence of degeneracies (crossings
at the spectral graph) was claimed for the case $\e=\frac12$ and supporting numerical evidence was presented for the half-integral
parameter $\e$ in \cite{LB2015JPA}, investigating an earlier empirical observation in \cite{B2011PRL}.
Moreover, this numerical verification was proved for $\e=\frac12$ and formulated mathematically as a conjecture
(see \S \ref{sec:constPoly}, Conjecture \ref{W2016JPA}) for the general half-integer $\e$ case in \cite{W2016JPA},
hinting at a hidden symmetry present in this case. In this paper we prove the conjecture affirmatively in general for \(\e = \ell/2 \, (\ell \in \Z) \)
(cf. Theorem \ref{thm:Main}).

Let us now briefly draw the whole picture of the spectrum of the AQRM by restricting ourselves to mention the technical issues
that we prove in this paper. The eigenvalues of the AQRM can be visualized in the spectral graph,
that is, the graph of the curves \(\{\lambda_i(g)\}_{i=1}^{\infty}\) in the \((g,E)\)-plane (\(E=\text{Energy}\)) for fixed \(\e \in \R\) and \(\Delta>0\).
In this picture, the exceptional eigenvalues are those that lie in the energy curves \(E  = N \pm \e - g^2 \), as shown conceptually
in Figure \ref{fig:ExceptEigen}(a).

\begin{figure}[htb]
  \begin{subfigure}[b]{0.45\textwidth}
    \centering
    \includegraphics[width=5.5cm]{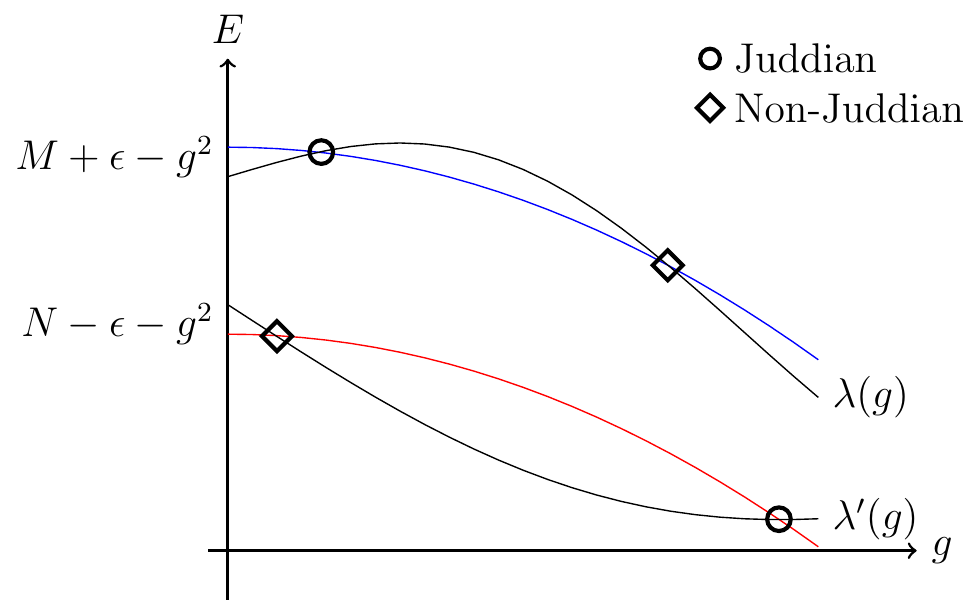}
   \caption{Eigenvalue curves $\lambda(g),\lambda'(g)$ and exceptional eigenvalues
      of $\HRabi{\e}$ for two integers $N,M \in \Z_{\geq 0}$. (e.g. see Figure~\ref{fig:specgraph_nhi}(a)) }
  \end{subfigure}
  ~
  \begin{subfigure}[b]{0.45\textwidth}
    \centering
    \includegraphics[height=4cm]{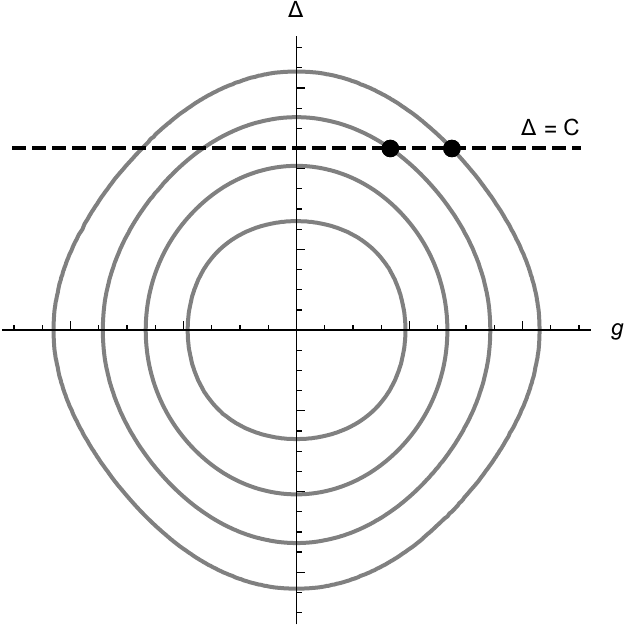}
    \caption{Curve \(\cp{M,\e}{M}((2g)^2,\Delta^2)=0\) in the \((g,\Delta)\)-plane. Juddian solutions correspond
      to points in the quadrant \(g,\Delta>0 \)}
  \end{subfigure}
  \caption{Exceptional eigenvalues of AQRM.}
  \label{fig:ExceptEigen}
\end{figure}

An eigenfunction \(\psi\) corresponding to an exceptional eigenvalue \(\lambda\) is called a
{\em Juddian solution} if its representation in the Bargmann space \(\Bargmann\) (cf. \S\ref{sec:confpict}) consists of polynomial components.
The associated eigenvalue \(\lambda\) is also called Juddian. Juddian solutions are also called {\em quasi-exact} and
have been investigated by Turbiner \cite{T1988CMP} with a viewpoint of $\mathfrak{sl}_2$-action and Heun operators.
Notice that Juddian solutions are not present for arbitrary
parameters \(g\) and \(\Delta\). In fact, it is known (\cite{LB2016JPA,WY2014JPA}) that an exceptional eigenvalue \(\lambda = N + \e - g^2 \) is present in the spectrum of $\HRabi{\e}$ and corresponds to a Juddian solution if and only if the parameters \(g\) and \(\Delta\) satisfy the polynomial equation
\begin{equation}
  \label{eq:constEq}
  \cp{N,\e}{N}((2g)^2,\Delta^2) = 0.
\end{equation}
The polynomial \(\cp{N,\e}{N}(x,y)\) (cf. \S \ref{sec:detexp}) is called {\em constraint polynomial} and
\eqref{eq:constEq} is called {\em constraint relation}. 
  The constraint polynomial \(\cp{N,\e}{N}(x,y)\) is actually the $N$-th member on a family of polynomials
  \(\{\cp{N,\e}{k}(x,y)\}_{k \geq 0}\) defined by a three-term recurrence relation.
\begin{dfn} \label{def:cp}
  Let $ N \in \Z_{\geq 0}$. The polynomials \(\cp{N,\e}{k}(x,y)\) of degree $k$ are defined
  recursively by
\begin{align*}
  \cp{N,\e}{0}(x,y) &= 1, \\
  \cp{N,\e}{1}(x,y) &= x + y - 1 - 2 \e, \\
  \cp{N,\e}{k}(x,y) &= (k x + y - k(k + 2 \e) ) P_{k-1}^{(N,\e)}(x,y) \\
                    &\qquad  - k(k-1)(N-k+1) x P_{k-2}^{(N,\e)}(x,y),
\end{align*}
for \(k \geq 2 \).
\end{dfn}
We note here that the family of polynomials \(\{\cp{N,\e}{k}(x,y)\}_{k \geq 0}\) falls outside the class of orthogonal
polynomials and therefore require special considerations. In \S \ref{sec:detexp} we describe some of the properties of the polynomials \(\cp{N,\e}{k}(x,y)\) and their roots.

In practice, however, not all exceptional eigenvalues correspond to Juddian (i.e. quasi-exact) solutions and, as in the case of the QRM, we call these eigenvalues and the corresponding eigensolutions {\em non-Juddian exceptional}.
This situation is illustrated conceptually in Figure \ref{fig:ExceptEigen}(a) (see Figure~\ref{fig:specgraph_nhi}(a) for a numerical example).
Further, the constraint relation for non-Juddian exceptional eigenvalues (cf. \S \ref{sec:NonJuddian}), which are shown to be non-degenerate
when $\e=0$ in \cite{B2013MfI}, cannot be obtained in terms of polynomials.

The constraint relation \eqref{eq:constEq} determines a curve in the \((g,\Delta)\)-plane consisting of
a number of concentric closed curves, shown conceptually in Figure \ref{fig:ExceptEigen}(b). In this picture, for fixed
\(\Delta = C>0\) the Juddian eigenvalues \(\lambda = N + \e - g^2 \) of \(\HRabi{\e} \) correspond to points in the intersection of the
curve \(\cp{N,\e}{N}((2g)^2, \Delta^2) = 0\) with the horizontal line \(\Delta = C\) in the \(g,\Delta >0 \) quadrant.

Next, in Figure \ref{fig:excepteigen2} we illustrate conceptually the way degeneracies appear in the exceptional spectrum
for the case $N,\ell\in \Z_{\geq0}$. When \(\e \in \R\) satisfies \(0 <|\e - \ell/2|<\delta \) for small \(\delta >0 \),
there are \(g,g'>0\) such that \( \lambda = N + \ell - \e - g^2 \) and \( \lambda' = N + \e - g'^2\) are
non-degenerate eigenvalues corresponding to Juddian solutions (shown with circle marks in Figure \ref{fig:excepteigen2}(a)).
In addition, exceptional eigenvalues \( \lambda = N + \ell - \e - g''^2 \) with non-Juddian solutions
may be present for \(g'' \not= g,g'\) (shown with diamond marks in Figure \ref{fig:excepteigen2}(a)).
On the other hand, the case \(\e = \ell/2 \,(\ell \in \Z)\) is illustrated in Figure \ref{fig:excepteigen2}(b). In this case, the energy
curves \(E = N + \ell - \e - g^2\) and \(E=N + \e - g^2\) coincide into the curve \( E = N + \ell/2 -g^2 \). As
\(\e \to \ell/2 \), the non-degenerate Juddian eigenvalues lying in the disjoint energy curves of Figure \ref{fig:excepteigen2}(a)
join into a single degenerate Juddian eigenvalue with multiplicity \(2\) lying on the resulting energy curve \( E = N + \ell/2 -g^2 \).
However, we remark that for \(g'>0\), with \(g\not=g'\) there may be additional non-Juddian solutions with
exceptional eigenvalue \(\lambda' = N + \ell/2 -g^2 \), as demonstrated in \cite{MPS2013} for the QRM (case \(\e = 0\)).
In \S \ref{sec:NonJuddian} we present numerical examples of these graphs, and we direct the reader to
\cite{LB2015JPA} for further examples.

\begin{figure}[htb]
  \centering
  \begin{subfigure}[b]{0.45\textwidth}
    \centering
    \includegraphics[width=6cm]{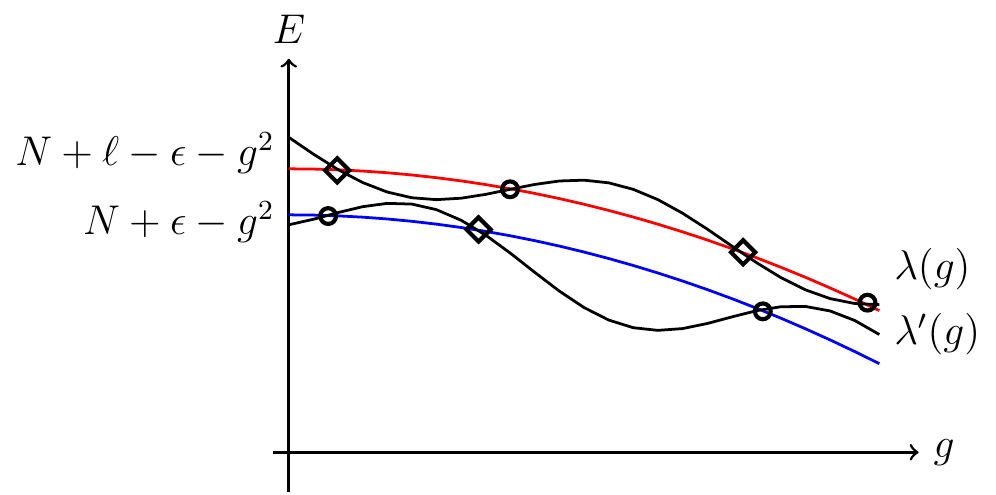}
    \caption{Case \( 0<|\e - \ell/2|< \delta\).}
    \label{fig:2(a)}
  \end{subfigure}
  ~
  \begin{subfigure}[b]{0.45\textwidth}
    \centering
    \includegraphics[width=6cm]{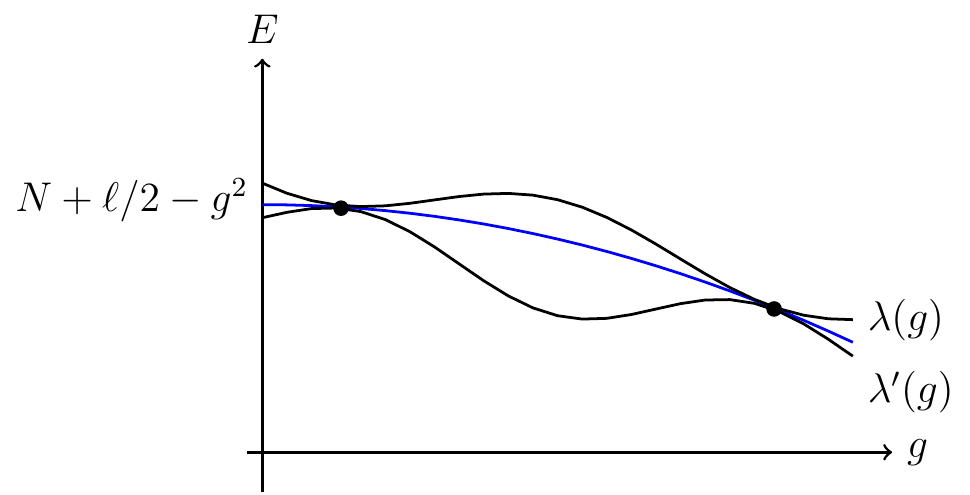}
    \caption{Case \(\e = \ell/2\).}
    \label{fig:2b}
  \end{subfigure}
  \caption{Exceptional eigenvalues for \(N,\ell \in \Z_{\geq 0}\). Circle marks denote Juddian solutions and diamond marks
    denote non-Juddian exceptional solutions. The degeneracies in the case \(\e=\ell/2\) consist only of Juddian solutions.}
  \label{fig:excepteigen2}
\end{figure}

In \S \ref{sec:constPoly} we prove that degenerate exceptional eigenvalues \(\lambda\) with Juddian solutions
exist in general for any half-integer $\e$ (Theorem \ref{thm:Main}) by studying certain determinant expressions for
the constraint polynomials $\cp{N,\e}{N}((2g)^2,\Delta^2)$. In particular, if \(\lambda = N + \ell/2 -g^2 \, (\ell \in \Z) \)
is a Juddian eigenvalue (corresponding to a root of the constraint polynomial) then its multiplicity is \(2\) and
the two linearly independent solutions are Juddian. In \S\ref{sec:degen-eigenv-aqrm} we show that all the Juddian solutions corresponding to
exceptional eigenvalues are degenerate. Moreover, in \S \ref{sec:estimation}, Theorem \ref{LB} we count the exact number of Juddian
solutions relative to the pair \((g, \Delta)\), giving a (complete) generalization of the results given in \cite{LB2015JPA} for the AQRM and in
\cite{K1985JMP} for the QRM.

The situation for the degeneracy of Juddian solutions in the case \(N = 5\) and \(\e = \ell/2 = 3/2\) is illustrated in
Figure \ref{fig:constraintgraph} with the graphs of the curves \(\cp{N,\e}{N}((2g)^2,\Delta^2)=0\) and \(\cp{N+\ell,-\e}{N+\ell}((2g)^2,\Delta^2)=0\)
in the \((g,\Delta)\)-plane for different choices of \(\e > 0 \). As we can see in Figures \ref{fig:constraintgraph}(a)
and Figure \ref{fig:constraintgraph}(b), as \(\e\) tends to \(\ell/2\) the two curves become coincident until finally, at \(\e = \ell/2\)
(Figure \ref{fig:constraintgraph}(c)) the two curves coincide completely. In the case \(\e = \ell/2 \, (\ell \in \Z)\), any point \((g,\Delta)\) with \(g,\Delta>0\)
in the resulting curve corresponds to a degenerate Juddian solution for the eigenvalue \(\lambda = N + \ell/2 - g^2\).
The aforementioned Theorem \ref{thm:Main} (cf. Conjecture \ref{W2016JPA}) gives a complete explanation of the coincidence of
the two curves. In particular, by Theorem \ref{thm:Main} we have the divisibility of the constraint polynomial
$\cp{N+\ell,-\ell/2}{N+\ell}((2g)^2,\Delta^2)$ by $\cp{N,\ell/2}{N}((2g)^2,\Delta^2)$ and positivity of the resulting divisor (a polynomial of degree $\ell$).

We  notice, however, that the crossings between the curves of the constraint relations appearing in Figures \ref{fig:constraintgraph}(a)
and \ref{fig:constraintgraph}(b) do not constitute degeneracies as the associated Juddian solutions have different eigenvalues for \(\e \neq \ell/2 \)
(\(\lambda_1 = N + \e - g^2 \) and \(\lambda_2 = N + \ell - \e -g^2\) respectively).

\begin{figure}[htb]
  ~
  \begin{subfigure}[b]{0.30\textwidth}
    \centering
    \includegraphics[height=4cm]{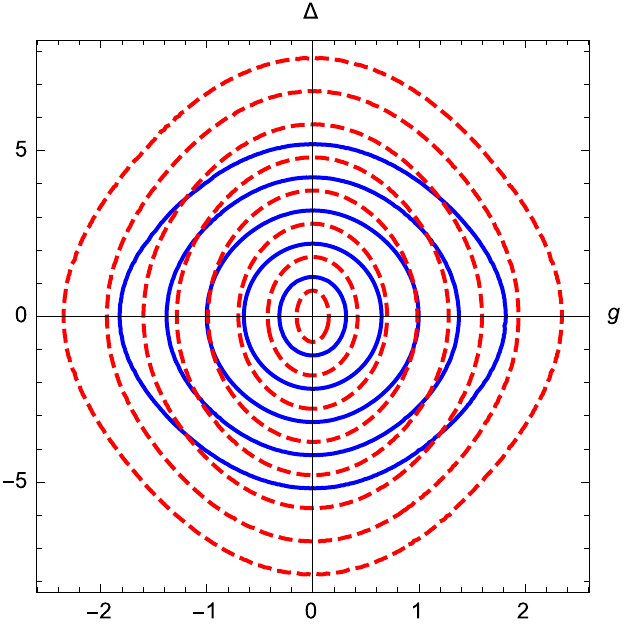}
    \caption{\(\e = 0.2.\)}
  \end{subfigure}
  ~
  \begin{subfigure}[b]{0.30\textwidth}
    \centering
    \includegraphics[height=4cm]{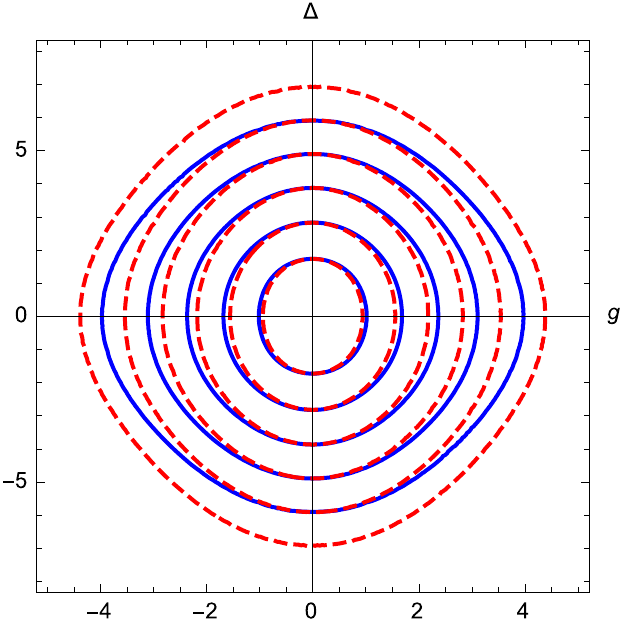}
    \caption{\(\e = 1. \)}
  \end{subfigure}
  ~
  \begin{subfigure}[b]{0.30\textwidth}
    \centering
    \includegraphics[height=4cm]{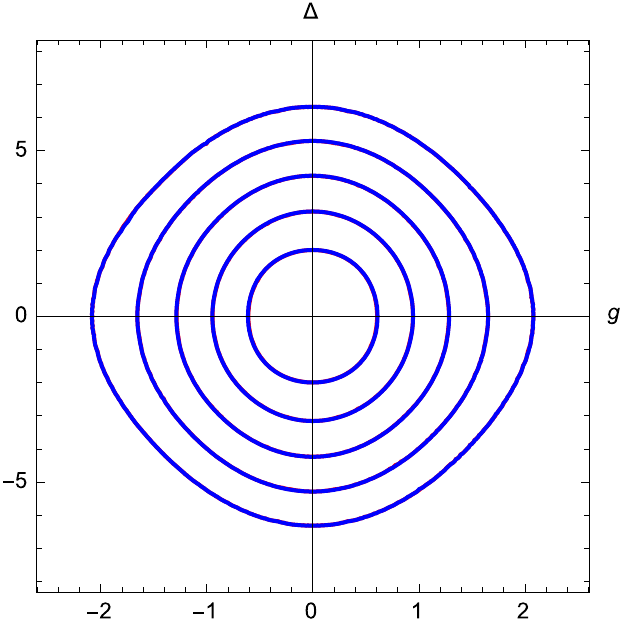}
    \caption{\(\e = 3/2.\)}
  \end{subfigure}
  \caption{Curves \(\cp{5,\e}{5}((2g)^2,\Delta^2)=0\) (continuous line)  and \(\cp{8,-\e}{8}((2g)^2,\Delta^2)=0\) (dashed line). The two curves overlap in
    the case (c) $\e=3/2$ (Theorem \ref{thm:Main}).}
  \label{fig:constraintgraph}
\end{figure}

The second purpose of this paper is to complete the whole picture of the spectrum based on the study of the exceptional eigenvalues,
particularly the aforementioned Juddian eigenvalues.

As in the case of the QRM, eigenvalues other than the exceptional ones are called regular. Equivalently, a regular eigenvalue $\lambda$
is one of the form $\lambda= x \pm\e -g^2 (x \not\in \Z_{\geq 0})$. It is known that regular eigenvalues are non-degenerate. Notice also that
regular eigenvalues are always obtained from zeros of the $G$-function \(G_\e(x;g,\Delta)\) (cf. \S \ref{sec:Gfunct}). In \S \ref{sec:NonJuddian}
we define the constraint $T$-function \(T_\e^{(N)}(g,\Delta)\) whose zeros correspond to exceptional eigenvalues \( \lambda = N \pm  \e - g^2\) with
non-Juddian solution. Thus, the transcendental equation \(T_\e^{(N)}(g,\Delta) = 0\) gives the constraint relation for non-Juddian exceptional
eigenvalues. In \S \ref{sec:NonJuddian}  we present numerical examples of the curves determined by the constraint relation for non-Juddian exceptional
eigenvalues when \(\e = \frac12 \). Further numerical examples can be found in  \cite{MPS2013} for the case \(\e = 0 \) and in
\cite{LB2015JPA} for the non half-integral general case.

We summarize the spectrum of the AQRM using the constraint relations given by the $G$-functions,
the constraint polynomials \(\cp{N,\e}{N}(x,y) \) and constraint $T$-functions. Setting \(N,\ell \in\Z_{\geq0}\), we have

\begin{compactitem}
\item If \(G_\e(x;g,\Delta)= 0\), then \(\lambda = x - g^2\) $(x \pm \e \not \in \Z_{\geq0})$ is a regular eigenvalue.
\item If \( \cp{N,\pm \e}{N}((2g)^2,\Delta^2) = 0 \) ($N\in \Z_{\geq0}$), then \(\lambda = N \pm \e - g^2\) is an exceptional eigenvalue with Juddian solution.
  Furthermore, if \(\e = \ell/2 \) \( (\ell \in \Z_{\geq 0})\) and \(\lambda = N + \ell/2 - g^2\) is a Juddian eigenvalue for some $N$, then  \(\lambda\)
  is degenerate with multiplicity two and the two eigensolutions are Juddian. Moreover, we show that \(\lambda = N - \ell/2 - g^2\) with
  \(N  \leq  \ell \) does not occur as a Juddian eigenvalue (cf. \S \ref{sec:pos}).
\item  If \( T_{\pm\e}^{(N)}(g,\Delta)  = 0 \)  ($N\in \Z_{\geq 0}$), then \(\lambda = N \pm \e - g^2\) is an exceptional eigenvalue with non-Juddian exceptional
  solution (cf. \S \ref{Non-Juddian}).
\item The spectrum of the AQRM possesses a degenerate eigenvalue if and only if the parameter $\e$ is a half integer. Furthermore, all degenerate
  eigenstates consist of Juddian solutions (cf. Theorem \ref{thm:AQRMSpec} in \S \ref{sec:LargestExponent})
\end{compactitem}

We also make an extensive study of the $G$-function $G_\e(x;g,\Delta)$ and its relation with $T_\e^{(N)}(g,\Delta)$ and $\cp{N, \e}{N}((2g)^2,\Delta^2) $ in
\S\ref{sec:Gfunct}, \S\ref{sec:NonJuddian} and \S\ref{sec:SpectralDet}. Especially, we observe that the meromorphic function $G_\e(x;g,\Delta)$
actually possesses almost complete information about Juddian and non-Juddian exceptional eigenvalues. In particular, in \S\ref{sec:SpectralDet} the pole
structure of $G_\e(x;g,\Delta)$ reveals a finer structure for exceptional eigenvalues for $\e=\ell/2\,(\ell\in \Z_{\geq 0})$. For instance, we see that
$G_{\ell/2}(x;g,\Delta)$ can have, in general, simple poles at \( x = N - \ell/2 \, (N  < \ell )\) and double poles at \( x = N + \ell/2 \, ( N \in \Z_{\geq 0})\).
When \(\lambda = N - \ell/2-g^2\) with \(N  < \ell \) is a non-Juddian exceptional eigenvalue the simple pole at \( x = N - \ell/2 \) disappears (cf. Proposition
\ref{prop:polehi1simple}). In contrast, when \(\lambda = N + \ell/2 - g^2 \) is a Juddian eigenvalue the double pole of $G_{\ell/2}(x;g,\Delta)$ at $x=N\pm\ell/2$
disappears and if there is a non-Juddian exceptional eigenvalue \(\lambda = N + \ell/2 - g^2 \), then the double pole of $G_{\ell/2}(x;g,\Delta)$ at $x=N\pm\ell/2$
either vanishes or is simple (cf. Proposition \ref{prop:polehi}).
Moreover, we prove that the meromorphic function $G_\e(x;g,\Delta)$ is essentially, i.e. up to a multiple of two gamma functions, identified with the
spectral determinant of the Hamiltonian $\HRabi{\e}$. In other words, $G_\e(x;g,\Delta)$ is expressed by the zeta regularized product
(cf. \cite{QHS1993TAMS}) defined by the Hurwitz-type spectral zeta function of the AQRM, and equivalently this fact confirms the (physically
intuitive) experimental numerical observation done in \cite{LB2016JPA}.

In \S\ref{sec:repr-theor-pict}  we make a representation theoretic description of the non-Juddian exceptional eigenvalues. Recall that the eigenstates of the quantum harmonic oscillator are described by certain weight subspaces of the oscillator representation of \(\mathfrak{sl}_2\) (cf. \cite{HT1992}). Similarly, the Juddian solutions are known to be captured (i.e. determined) by a pair of irreducible finite dimensional representations of \(\mathfrak{sl}_2\) \cite{W2016JPA}. In the same manner, in Theorem \ref{thm:eigenproblem} we show that the non-Juddian exceptional eigenvalues are captured by a pair of lowest weight irreducible representations of \(\mathfrak{sl}_2\). 

Finally, in \S \ref{sec:further} we study generating functions of constraint polynomials with their defining sequence
$P_k^{(N,\e)}((2g)^2,\Delta^2)$. In fact, we observe that the generating function of $P_k^{(N,\e)}((2g)^2,\Delta^2)$ satisfies a
confluent Heun equation (see \S\ref{sec:div}), which can be seen as natural by virtue of certain properties of
the aforementioned $G$-function. As a byproduct of the discussion we have also an alternative proof of the divisibility part of the conjecture.

It is important to notice that, although having analytic solutions, the asymmetric quantum Rabi models,
even in the symmetric (QRM) case, are in general known not to be integrable models in the Yang-Baxter sense \cite{BZ2015}. However, it is interesting to note that recently the existence of monodromies associated with the singular
points of the eigenvalue problem for the quantum Rabi model has been discussed in \cite{CCR2016}. We also remark that there has been recent interest in the integrability of the  Jaynes-Cummings model and generalizations (see e.g. \cite{LPN2016}). Moreover, we note that although there are various coupling regimes of the AQRM given in terms of $\Delta, \omega\, (=1)$ and $g$ physically, the discussion in this paper is independent of the choice of regimes (cf. \cite{BMSSRWU2017}).

\section{Confluent Heun picture of AQRM} \label{sec:confpict}

In this section, using the confluent Heun picture, we give a description of the exceptional (Juddian and non-Juddian exceptional) and regular eigenvalues of the AQRM. For that purpose, we employ the Bargmann space \(\Bargmann\) representation of boson operators \cite{Bar1961}, in the standard way (cf. \cite{K1985JMP,B2011PRL-OnlineSupplement}, etc), to reformulate the eigenvalue problem of \(\HRabi{\e}\) as a system of linear differential equations.

Recall that the Hilbert space $\Bargmann$ is the space of entire functions equipped with the inner product
\[
  (f|g)= \frac1\pi \int_{\C} \overline{f(z)}g(z)e^{-|z|^2}d(\RE(z))d(\IM(z)).
\]
In this representation, the operators $a^\dag$ and $a$ are realized as the
multiplication and differentiation operators over the complex variable: $a^\dag = (x-\partial_x)/\sqrt2 \to z$ and $a= (x+\partial_x)/\sqrt2 \to \partial_z:=\frac{d}{dz}$,
so that the Hamiltonian \(\HRabi{\e}\) is mapped to the operator
\[
\tHRabi{\e} :=  \begin{bmatrix}
    z \partial_z + \Delta & g(z + \partial_z) + \e  \\
    g(z+\partial_z) + \e  & z \partial_z - \Delta
  \end{bmatrix}.
\]
By the standard procedure (cf. \cite{B2013AP,LB2015JPA}), we observe that the Schr\"odinger equation
$\HRabi{\e}\varphi=\lambda \varphi \, (\lambda \in \R)$ is equivalent to the system of first order differential equations
\begin{equation*}
\tHRabi{\e}\psi=\lambda \psi, \quad
\psi=  \begin{bmatrix}
 \psi_{1}(z) \\
 \psi_{2}(z)
\end{bmatrix}.
\end{equation*}
Hence, in order to have an eigenstate of $\HRabi{\e}$, it is sufficient to obtain an eigenstate
$\psi \in \Bargmann$, that is, $\mathbf{BI}$: $(\psi_i|\psi_i) <\infty $, and $\mathbf{BII}$: $\psi_i$ are holomorphic everywhere in the whole complex
plane $\C$ for $i=1,2$. Actually, it can be shown that any such function satisfying condition $\mathbf{BII}$ also satisfies
the condition $\mathbf{BI}$ (cf. \cite{B2013AP}).

Therefore, the eigenvalue problem of the AQRM amounts to finding entire functions \(\psi_1,\psi_2 \in \Bargmann\) and real number
\( \lambda \) satisfying
\begin{equation*}
  \left\{
  \begin{aligned}
    (z \partial_z + \Delta) \psi_1 + (g (z + \partial_z) + \e) \psi_2 &= \lambda \psi_1, \\
    (g (z + \partial_z) + \e) \psi_1 + (z \partial_z - \Delta) \psi_2 &= \lambda \psi_2.
  \end{aligned}
  \right.
\end{equation*}

Now, by setting \(f_\pm = \psi_1 \pm \psi_2 \), we get
\begin{equation}  \label{eq:system-1}
  \left\{
  \begin{aligned}
  (z + g) \frac{d}{d z} f_+ + (g z + \e - \lambda) f_+ + \Delta f_- &= 0, \\
  (z - g) \frac{d}{d z} f_- - (g z + \e + \lambda) f_- + \Delta f_+ &= 0,
  \end{aligned}
  \right.
\end{equation}
where, by using the substitution \( \phi_{1,\pm}(z) := e^{g z} f_\pm (z) \)
and the change of variable \( y = \frac{g+z}{2 g}\), we obtain
\begin{equation} \label{eq:system0}
  \left\{
  \begin{aligned}
  y \frac{d}{d y} \phi_{1,+}(y) &= ( \lambda + g^2  - \e ) \phi_{1,+}(y) - \Delta \phi_{1,-}(y), \\
  (y-1) \frac{d}{d y} \phi_{1,-}(y) &= ( \lambda + g^2 - \e  -4 g^2 + 4 g^2 y + 2\e) \phi_{1,-}(y) - \Delta \phi_{1,+}(y).
  \end{aligned} \right.
\end{equation}
Defining \( a := -(\lambda + g^2 - \e) \), we get
\begin{equation} \label{eq:system1}
  \left\{
  \begin{aligned}
   y \frac{d}{d y} \phi_{1,+}(y) &= - a \phi_{1,+}(y) - \Delta \phi_{1,-}(y),  \\
   (y-1) \frac{d}{d y} \phi_{1,-}(y) &= -( 4 g^2 - 4 g^2 y  + a - 2\e) \phi_{1,-} (y) - \Delta \phi_{1,+} (y).
  \end{aligned} \right.
\end{equation}

Similarly, by applying the substitutions \( \phi_{2,\pm}(z) := e^{-g z} f_\pm (z) \) and \(\bar{y} = \frac{g-z}{2g}\) to the system
\eqref{eq:system-1}, we get
\begin{equation} \label{eq:system2p}
  \left\{
  \begin{aligned}
    (\bar{y}-1) \frac{d}{d \bar{y}} \phi_{2,+}(\bar{y}) &= -(4g^2 - 4 g^2 \bar{y} + a)\phi_{2,+}(\bar{y}) - \Delta \phi_{2,-}(\bar{y}),  \\
    \bar{y} \frac{d}{d \bar{y}} \phi_{2,-}(\bar{y}) &= -( a  -  2\e) \phi_{2,-} (\bar{y}) - \Delta \phi_{2,+} (\bar{y}).
  \end{aligned} \right.
\end{equation}
This system gives another (possible) solution \((\phi_{2,+}(\bar{y}),\phi_{2,-}(\bar{y}))\) to the eigenvalue problem.
Note that $a  -  2\e= -(\lambda+g^2+\e)$ and that \(\bar{y} = 1-y\), where \(y\) is the variable used in \eqref{eq:system1}.

The singularities of system \eqref{eq:system1} and  \eqref{eq:system2p} at \(y =0\) and \( y=1 \) are regular.
The exponents of the equation system can be obtained by standard computation, and are shown in Table \ref{tab:exp} for reference.

\ctable[
caption = Exponents of systems \eqref{eq:system1} and \eqref{eq:system2p}.,
label   = tab:exp,
pos     = htb
]{rcllll}{
}{                                                          \FL
&            & \(\phi_{1,-}(y) \)    & \(\phi_{1,+}(y) \) & \(\phi_{2,-}(1-y)\) &  \(\phi_{2,+}(1-y)\)\ML
&\(y=0 \) & \(0, - a+1\)  & \(0, -a \) &  \(0,-a + 1 \)   &     \(0, - a \)    \NN
&\(y=1 \) & \(0, -a + 2\e \) &  \(0, -a + 2\e+ 1 \)& \(0, -a + 2\e \)  &   \(0,-a+2\e + 1 \)  \LL
}

We remark that Table \ref{tab:exp} in particular shows that each regular eigenvalue is not degenerate because one of
the exponents is necessarily not an integer.

Each differential equation system determines a second order differential operator of confluent Heun type \cite{Heun2008,SL2000}. For instance, by eliminating \(\phi_{1,-}(y) \) from the system \eqref{eq:system1} we obtain the operator
  \begin{equation}
    \label{eq:H1eps}
        \mathcal{H}_1^{\e}(\lambda) = \frac{d^2}{d y^2} + \left( -4 g^2  + \frac{a+1}{y} + \frac{a-2 \epsilon}{y-1} \right) \frac{d}{d y} + \frac{- 4g^2 a y +  \mu + 4\epsilon g^2 - \epsilon^2 }{y(y-1)}.
  \end{equation}
  Similarly, by eliminating \(\phi_{2,-}(\bar{y})\) from the system \eqref{eq:system2p} we obtain
  \begin{equation}
    \label{eq:H2eps}
    \mathcal{H}_2^{\e}(\lambda) = \frac{d^2}{d \bar{y}^2} + \left( -4 g^2  + \frac{a - 2\epsilon}{\bar{y}} + \frac{a+1}{\bar{y}-1} \right) \frac{d}{d \bar{y}} + \frac{- 4g^2( a - 2\epsilon + 1) y +  \mu - 4\epsilon g^2 - \epsilon^2 }{\bar{y}(\bar{y}-1)}.
  \end{equation}
  Here, the accesory parameter \(\mu \) is given by
  \[
    \mu = (\lambda + g^2)^2 - 4g^2 (\lambda+g^2) - \Delta^2.
  \]
  
  In \S \ref{sec:repr-theor-pict} we describe how these operators can be captured by a particular element of \(\mathcal{U}(\mathfrak{sl}_2)\), the universal enveloping algebra of \(\mathfrak{sl}_2\), and how different type of eigenvalues (regular, Juddian, non-Juddian exceptional) correspond to distinct type of irreducible representations of \(\mathfrak{sl}_2\).

\subsection{Exceptional solutions corresponding to the smallest exponent} \label{sec:smallexp}

We proceed to study exceptional eigenvalues from the point of view of the confluent picture of
the AQRM. A part of the discussion here follows \cite{B2013AP} and \cite{B2013MfI} for $\e=0$.
Recall that an eigenvalue \(\lambda\) of \(\HRabi{\e}\) is exceptional if there is an integer
\(N \in \Z_{\geq 0}\) such that $\lambda=N\pm \e -g^2$.

Let us take $a$ as \( -a = (\lambda + g^2 - \e) = N \in \Z_{\geq 0}\).
The corresponding system \eqref{eq:system1} of differential equations is then given by
\begin{equation} \label{eq:systemN}
  \left\{
    \begin{aligned}
    y \frac{d}{d y} \phi_{1,+}(y) &= N \phi_{1,+}(y) - \Delta \phi_{1,-}(y)  \\
    (y-1) \frac{d}{d y} \phi_{1,-}(y) &= ( N  -4 g^2 + 4 g^2 y + 2\e) \phi_{1,-}(y) - \Delta \phi_{1,+}(y).
    \end{aligned} \right.
\end{equation}
The exponents of \(\phi_{1,-}\) at \(y=0\) are \( \rho^-_1 =0, \rho^-_2 = N+1\). Likewise, the exponents of \(\phi_{1,+}\) at
\(y=0\) are \( \rho^+_1 =0, \rho^+_2 = N\). Since the difference between the exponents is a positive integer,
the local analytic solutions will develop a logarithmic branch-cut at \(y=0\).

In this subsection, we revisit Juddian solutions. The local Frobenius solution corresponding to the smallest
exponent \(\rho_1^- = 0\) has the form
\begin{equation}\label{smallerFrobenius}
  \phi_{1,-}(y) (= \phi_{1,-}(y;\e)) =  \sum_{n=0}^\infty K^{(N,\e)}_n y^n,
\end{equation}
where \(K^{(N,\e)}_0 \not= 0 \) and \(K^{(N,\e)}_n = K^{(N,\e)}_n(g,\Delta)\). Integration of the first equation of \eqref{eq:systemN}
gives

\begin{equation}\label{LogTermFrobenius}
  \phi_{1,+}(y) (=\phi_{1,+}(y;\e))
  = c y^N - \Delta \sum_{n \neq N}^\infty \frac{K^{(N,\e)}_n}{n-N} y^n - \Delta K^{(N,\e)}_N y^N \log y,
\end{equation}
with constant \(c \in \C \). A necessary condition for \(\phi_{1,+}(y)\) to be an element of the Bargmann space \(\Bargmann\) is that
\(\phi_{1,+}(y)\) is an entire function, forcing \(K_N^{(N,\e)} = 0\) to make the logarithmic term vanish. Suppose \(\phi_{1,+}(y)\ \in \Bargmann\),
then by using the second equation of \eqref{eq:systemN} we obtain the recurrence relation for the coefficients
\begin{equation}\label{eq:recurrKn1}
  (n+1) K^{(N,\e)}_{n+1} + \left(N-n - (2g)^2 + \frac{\Delta^2}{n-N} + 2\e \right) K^{(N,\e)}_n + (2g)^2 K^{(N,\e)}_{n-1} = 0,
\end{equation}
valid for \(n \not= N\). This recurrence relation clearly shows the dependence of the coefficients \(K^{(N,\e)}_n = K^{(N,\e)}_n(g,\Delta)\)
on the parameters of the system. Additionally, for \( n = N \), by the second equation of \eqref{eq:systemN}, we have
\begin{equation}\label{eq:constant}
 \Delta c = (2g)^2 K^{(N,\e)}_{N-1} + (N+1) K^{(N,\e)}_{N+1}.
\end{equation}
Setting \( c = (2g)^2 K^{(N,\e)}_{N-1} / \Delta \) makes \( K^{(N,\e)}_{N+1}\) vanish, and then, by repeated use of the recurrence
\eqref{eq:recurrKn1}, we see that for all positive integers \(k\) the coefficients \(K^{(N,\e)}_{N+k}\) also vanish.
Thus, the solutions of \eqref{eq:systemN} given by
\begin{align}\label{eq:polynomial_sol}
    \phi_{1,-}(y) &= \sum_{n=0}^{N-1} K^{(N,\e)}_n y^n, \\
    \phi_{1,+}(y) &= \frac{4 g^2 K^{(N,\e)}_{N-1}}{\Delta} y^N - \Delta \sum_{n=0}^{N-1} \frac{K^{(N,\e)}_n}{n-N} y^n
\end{align}
are polynomial solutions.

From the discussion above, we can regard the equation
\begin{equation}
  \label{eq:constK}
  K_N^{(N,\e)}(g,\Delta) = 0,
\end{equation}
as a {\em constraint relation} for the Juddian eigenvalue \(\lambda = N \pm \e - g^2\). In this context, a constraint
relation is an additional condition imposed to the parameters of the differential equation system in order to obtain
certain type of solutions. In fact, the constraint equation \eqref{eq:constK} is equivalent to the constraint relation \eqref{eq:constEq} given in the introduction.

In order to see this, let us introduce some notation used throughout the paper. For a tridiagonal matrix, we put
\begin{equation*}
  \Tridiag{a_i}{b_i}{c_i}{1\le i\le n}
  :=\begin{bmatrix}
    a_1 & b_1 & 0 &  0 & \cdots & 0    \\
    c_1 & a_2 & b_2 & 0 & \cdots  & 0\\
    0 & c_2 & a_3 & b_3 &  \cdots & 0 \\
    \vdots & \ddots   & \ddots & \ddots & \ddots & \vdots   \\
    0 &  \cdots &  0 & c_{n-2} & a_{n-1} & b_{n-1} \\
    0 & \cdots  & 0 & 0 & c_{n-1} & a_n
  \end{bmatrix}.
\end{equation*}

\begin{prop}\label{Prop:reccurEquiv}
  Let \(N \in \Z_{\geq 0}\) and fix \(\Delta> 0\).   Then, the zeros \(g\) of \(K^{(N,\e)}_N = K_N^{(N,\e)}(g,\Delta)\) defined by
  \eqref{eq:recurrKn1} and \( P_N^{(N,\e)}((2g)^2,\Delta^2)\) coincide.
  In particular, if \(g\) is a zero of \(\cp{N,\e}{N}((2g)^2,\Delta^2)\), then \( \lambda = N  + \e - g^2\) is an exceptional eigenvalue with
  corresponding Juddian solution given by \(\phi_{1,+}(y)\) and \(\phi_{1,-}(y)\) above.
\end{prop}

\begin{proof}
  By multiplying \(K_n^{(N,\e)}\) by \((K_0^{(N,\e)})^{-1} \) for all \(n \in \Z_{\geq 0}\), we can assume that \(K_0^{(N,\e)} = 1 \).
  Then, we can rewrite the recurrence relation for the coefficients \(K_n^{(N,\e)}\) as
  \begin{equation*}
    K_{n}^{(N,\e)} = \frac{1}{n} \left( (2g)^2 + \frac{\Delta^2}{N-n+1} + n-1 - N - 2\e \right)K_{n-1}^{(N,\e)}
    - \frac{1}{n} (2g)^2 K_{n-2}^{(N,\e)},
  \end{equation*}
  for \(n \leq N\). As in \S \ref{sec:detexp}, we easily see that \(K_N^{(N,\e)}\) has the determinant expression
  \begin{equation*}
    K^{(N,\e)}_N=\det\Tridiag{\frac1{N-i+1}((2g)^2+\frac{\Delta^2}i-i-2\e)}{\frac{2g}{N-i+1}}{2g}{1\le i\le N}.
  \end{equation*}
  Next, for \(i=1,2,\ldots,N \), factor \(\frac{1}{i(N+1-i)}\) from the \(i\)-th row in the
  determinant to get the expression of $K^{(N,\e)}_N$ as
  \begin{equation*}
    \frac1{(N!)^2}\det\Tridiag{i(2g)^2+\Delta^2-i^2-2i\e}{2ig}{(2N-1-i)g}{1\le i\le N}.
  \end{equation*}
  The recurrence relation corresponding to this continuant is the same as the recurrence relation of \(\cp{N,\e}{k}((2g)^2,\Delta^2)\) (cf. Definition \ref{def:cp}), including the initial conditions.
  Thus
  \[
    K_N^{(N,\e)}(N+\e;g,\Delta,\e) = \frac{1}{(N!)^2}\cp{N,\e}{N}((2g)^2,\Delta^2),
  \]
  completing the proof. 
\end{proof}

\subsection{Exceptional solutions corresponding to the largest exponent}\label{sec:LargestExponent}

From general theory, the differential equation system \eqref{eq:systemN}, corresponding to an exceptional eigenvalue \(\lambda = N \pm \e - g^2 \), may have a Frobenius solution associated to the largest exponent. Any exceptional eigenvalue arising from such a solution is called non-Juddian exceptional.

The largest exponent of \(\phi_{1,-}\) at \(y=0\) is \(\rho_2^-=N+1\), it follows that there is a local Frobenius solution analytic at \(y=0\) of the form
\begin{equation}\label{eq:largexpsol-}
  \phi_{1,-}(y) (= \phi_{1,-}(y;\e)) =\sum_{n=N+1}^\infty \bar{K}^{(N,\e)}_n y^n,
\end{equation}
where \(\bar{K}^{(N,\e)}_{N+1} \ne 0\) and \(\bar{K}^{(N,\e)}_n = \bar{K}^{(N,\e)}_n(g,\Delta)\). Integration of the first equation of
\eqref{eq:systemN} gives
\begin{equation}\label{eq:largexpsol+}
  \phi_{1,+}(y) (= \phi_{1,+}(y;\e)) = c y^N - \Delta \sum_{n=N+1}^\infty \frac{\bar{K}^{(N,\e)}_n}{n-N} y^n,
\end{equation}
with constant \(c \in \C\). The second equation of \eqref{eq:systemN} gives the recurrence relation
\begin{equation}\label{eq:recurrKn}
(n+1) \bar{K}^{(N,\e)}_{n+1} + \Bigl(N-n - (2 g)^2 + \frac{\Delta^2}{n-N} + 2\e \Bigr) \bar{K}^{(N,\e)}_n
+ (2g)^2 \bar{K}^{(N,\e)}_{n-1} = 0,
\end{equation}
for \(n \geq N+1\) with initial conditions \(\bar{K}^{(N,\e)}_{N+1}=1 \) and \( \bar{K}^{(N,\e)}_{N}=0\). Furthermore, we also have
the condition
\[
  (N+1) \bar{K}^{(N,\e)}_{N+1} = (N+1) =  c \Delta,
\]
which determines value of the constant \(c = (N+1)/\Delta\). Notice that the radius of convergence of each series above equals $1$ from the defining recurrence relation \eqref{eq:recurrKn}.

\begin{rem}
  Notice that $\Delta c = (2g)^2 K^{(N,\e)}_{N-1} + (N+1) K^{(N,\e)}_{N+1}$ in \eqref{eq:constant}. This shows that there are two possibilities
  for the choice of solutions of \eqref{eq:systemN}. In other words, the choice of $c = (2g)^2 K_{N-1} / \Delta$ (for the smallest exponent)
  provides a polynomial solution, i.e. the Juddian solution while the choice $c = (N+1)\bar{K}_{N+1}/\Delta$ (for the largest exponent) provides
  a non-degenerate exceptional solution when the \(g\) satisfies \(T_\e^{(N)}(g,\Delta) = 0\) (cf. \S \ref{sec:Gfunct}). However, there is no chance
  to have contributions from both Juddian and non-Juddian exceptional eigenvalues (cf. Remark \ref{no-non-Juddian}).
\end{rem} 

In \S \ref{sec:NonJuddian} we describe the constraint function and constraint relation for the non-Juddian exceptional solution
corresponding to the eigenvalue \(\lambda = N \pm \e - g^2\).

\subsection{Remarks on regular eigenvalues and the \texorpdfstring{$G$}{G}-function} \label{sec:Gfunct}

To complete the description of the spectrum of the AQRM, we now turn our attention to the regular spectrum.
As we remarked in the introduction, any eigenvalues \(\lambda\) of the AQRM that is not an exceptional is called regular.
The regular eigenvalues arise as as zeros of the $G$-function introduced by Braak \cite{B2011PRL} in the study of
the integrability of the (symmetric) quantum Rabi model. Moreover, the converse also holds, thus the regular
spectrum is completely determined by the zeros of the $G$-function.

The \(G\)-function for the Hamiltonian \(\HRabi{\e}\) is defined as
\begin{equation*}
  G_\e(x;g,\Delta) := \Delta^2\bar{R}^+(x;g,\Delta,\e) \bar{R}^-(x;g,\Delta,\e) - R^+(x;g,\Delta,\e)R^-(x;g,\Delta,\e)
\end{equation*}
where
\begin{gather} \label{eq:gfuncR}
  R^\pm(x;g,\Delta,\e) = \sum_{n=0}^\infty K^{\pm}_n(x) g^n,
  \qquad
  \bar{R}^{\pm}(x;g,\Delta,\e) = \sum_{n=0}^\infty \frac{K^{\pm}_n(x)}{x-n\pm \e} g^n,
\end{gather}
whenever $x \mp \e  \not\in\Z_{\geq0} $, respectively. For \(n \in \Z_{\geq 0}\), define the functions \(f^{\pm}_n = f^{\pm}_n(x,g,\Delta,\e)\) by
\begin{equation}\label{eq:fn}
  f^{\pm}_n(x,g,\Delta,\e) = 2g + \frac{1}{2g} \Bigl( n - x \pm \e + \frac{\Delta^2}{x -n \pm \e}  \Bigr),
\end{equation}
then, the coefficients \(K^{\pm}_n(x)=K^{\pm}_n(x,g, \Delta,\e)\) are given by the recurrence relation
\begin{equation}\label{eq:RecursionKn}
 n K^{\pm}_{n}(x) = f^{\pm}_{n-1}(x,g,\Delta,\e) K^{\pm}_{n-1}(x) - K^{\pm}_{n-2}(x) \quad (n\geq1)
\end{equation}
with initial condition \(K_{-1}^{\pm} = 0\) and \(K_0^{\pm} = 1\), whence \(K_1^{\pm}=f_0^{\pm}(x,g,\Delta,\e)\).
It is also clear from the definitions that the equality
\begin{equation}\label{eq:signKpm}
  K^{\pm}_n(x,g, \Delta,-\e) =K^{\mp}_n(x,g, \Delta,\e),
\end{equation}
holds.

It is well-known (e.g. \cite{B2011PRL,B2013AP,Le2016}) that for fixed parameters \(\{g,\Delta,\e\}\) the zeros \(x_n\) of \( G_\e(x;g,\Delta)\) correspond to the regular eigenvalues \(\lambda_n = x_n-g^2\) of \(\HRabi{\e}\). These eigenvalues are called regular and always non-degenerate as in the case of regular eigenvalues of the quantum Rabi model.

\begin{rem}\label{rem:Parity}
  We remark that in the case of the QRM (i.e. $\e=0$), we have $G_0(x;g,\Delta)= G_+(x)\cdot G_-(x)$,  $G_\pm(x)$ being the $G$-functions
  corresponding to the parity defined as
  \[
    G_\pm(x)=\sum_{n=0}^\infty K_n(x)\Big(1\mp \frac{\Delta}{x-n}\Big)g^n,
  \]
  where $ K_n(x)= K_n^\pm(x,g,\Delta, 0)$ \cite{B2011PRL}.
  It is also known that these functions can be written in terms of confluent Heun functions (cf. \cite{B2013MfI}).
  Note also that there are no degeneracies within each parity subspace.
\end{rem}

The following result is obvious from the definition and the property \eqref{eq:signKpm} above.

\begin{lem}\label{G-function}
  The $G$-functions of \(\HRabi{\e}\) coincides with that of \(\HRabi{-\e}\):
  \begin{equation}
    \label{eq:gfunctrel}
    G_\e(x;g,\Delta)  = G_{-\e}(x;g,\Delta).
  \end{equation}
  In other words, the regular spectrum of \(\HRabi{\e}\) depends only on \(|\e|\). \QEDhere
\end{lem}

Now we make a brief remark on the poles of the function $G_\e(x;g,\Delta)$.
Although  $G_\e(x;g,\Delta)$ is not defined at \(x\in\Z_{\geq0}\pm\e\),
by the formula \eqref{eq:gfuncR} we can consider \(G_\e(x;g,\Delta)\) as a function with
singularities at the points \(x= N \pm \e\), \(N \in \Z_{\geq 0} \).

In order to describe the behavior of the $G$-function at the poles,
we notice the relation between the coefficients of the $G$-function
and the constraint polynomials in the following lemma.
The proof can be done in the same manner as Proposition \ref{Prop:reccurEquiv}.

\begin{lem} \label{lem:GcoeffcPoly}
  Let \(N \in \Z_{\geq 0}\). Then the following relation hold for $g>0$.
  \begin{equation}
    \label{eq:coeffGfunct}
         (N!)^2 (2g)^{N} K^{-}_N(N+\e;g,\Delta,\e) = \cp{N,\e}{N}((2g)^2,\Delta^2),
  \end{equation}
  In addition, if \(\e = \ell/2 \, (\ell \in \Z)\), it also holds that
\begin{equation*}
((N+\ell)!)^2 (2g)^{N+\ell} K^{+}_{N+\ell}(N+\ell/2;g,\Delta,\ell/2) = \cp{N+\ell,-\ell/2}{N+\ell}((2g)^2,\Delta^2). \QEDhere
\end{equation*}
\end{lem}

Now, let us consider the case $x=N+\e$. Observe that $f_N^-(x,g,\Delta,\e)$, as a function in $x$, has a simple
pole at $x=N+\e$, and thus, each of the rational functions \(K^{-}_n(x)=K^{-}_n(x,g, \Delta,\e)\), for \(n \geq N+1\),
also has a simple pole at $x=N+\e$. If $\cp{N,\e}{N}((2g)^2,\Delta^2)=0$, then \(K^{-}_N(N+\e)=0\) by the lemma above and
the coefficients \(K^{-}_n(N+\e)\) are finite for \(n \geq N+1\). Therefore, the functions \(R^-(x;g,\Delta,\e)\) and
\(R^+(x;g,\Delta,\e)\) converge to a finite value at \(x= N+\e\).

In the case $\cp{N,\e}{N}((2g)^2,\Delta^2)\ne0$ the $G$-function may or may not have a pole depending on the value of the residue
(as a function of \(g\) and \(\Delta\)) at \(x = N + \e \). We leave the detailed discussion to the subsection \S \ref{sec:SpectralDet}
after we have introduced the constraint $T$-function for non-Juddian exceptional eigenvalues.

To illustrate the discussion we show in Figure \ref{fig:gfunct_graph} the plot of the $G$-function \(G_\e(x;g,\Delta) \)
for fixed \(g,\Delta>0\) corresponding to roots of the constraint polynomials \(\cp{N,\e}{N}((2g)^2,\Delta^2)\).
In Figure \ref{fig:gfunct_graph}(a) we show the case of \(\e = 0.3, g\approx 0.5809, \Delta=1\) and \(N=1\), observe the finite value
of the \(G\)-function \(G_\e(x;g,\Delta)\) at \(x = 1.3 \) and the poles at \( x = N \pm 0.3 \, (N \in \Z_{\geq 0}, N \not=1)\).
In the Figures \ref{fig:gfunct_graph}(b)-(c) we show the half-integer case. Concretely, in Figure \ref{fig:gfunct_graph}(b)
we show the case \(\e = 1/2, g=1/2, \Delta =1 \) and \(N = 1\) and in Figure \ref{fig:gfunct_graph}(c) the case
\(\e = 1, g \approx 1.01229, \Delta=1.5\) and \(N = 2\). As expected from the discussion above, the function \(G_\e(x;g,\Delta)\)
has a finite value at \(x =  1.5 \) (for Figure \ref{fig:gfunct_graph}(b)) and \(x = 3 \) (for Figure \ref{fig:gfunct_graph}(c)),
while other values of \(x = N \pm \e \) are poles.

\begin{figure}[htb]
  ~
  \begin{subfigure}[b]{0.45\textwidth}
    \centering
    \includegraphics[height=3.25cm]{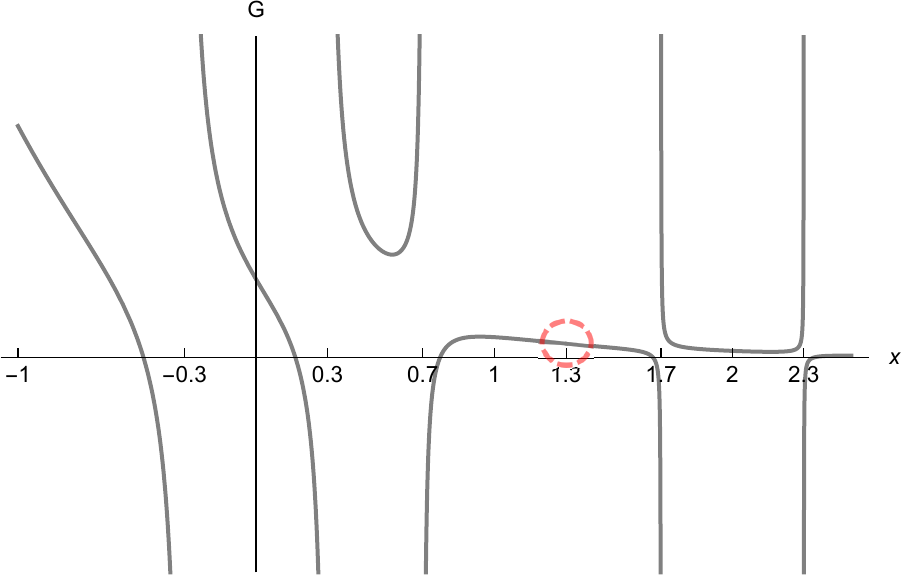}
    \caption{\(\e = 0.3, g \approx 0.5809 , \Delta=1/2\)}
  \end{subfigure}
  ~
  \begin{subfigure}[b]{0.45\textwidth}
    \centering
    \includegraphics[height=3.25cm]{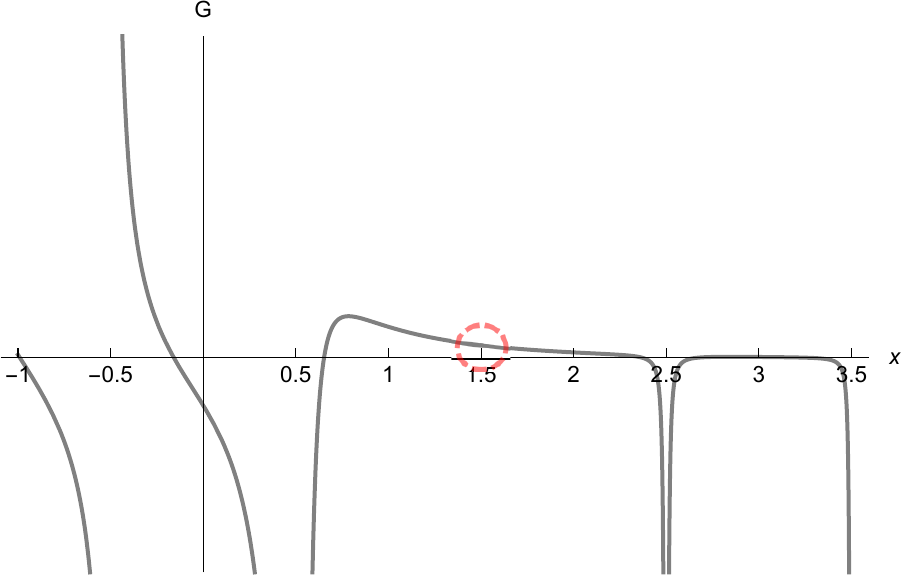}
    \caption{\(\e = 0.5, g = 0.5, \Delta=1 \)}
  \end{subfigure}
  \\
  \centering
  \begin{subfigure}[b]{0.45\textwidth}
    \centering
    \includegraphics[height=3.25cm]{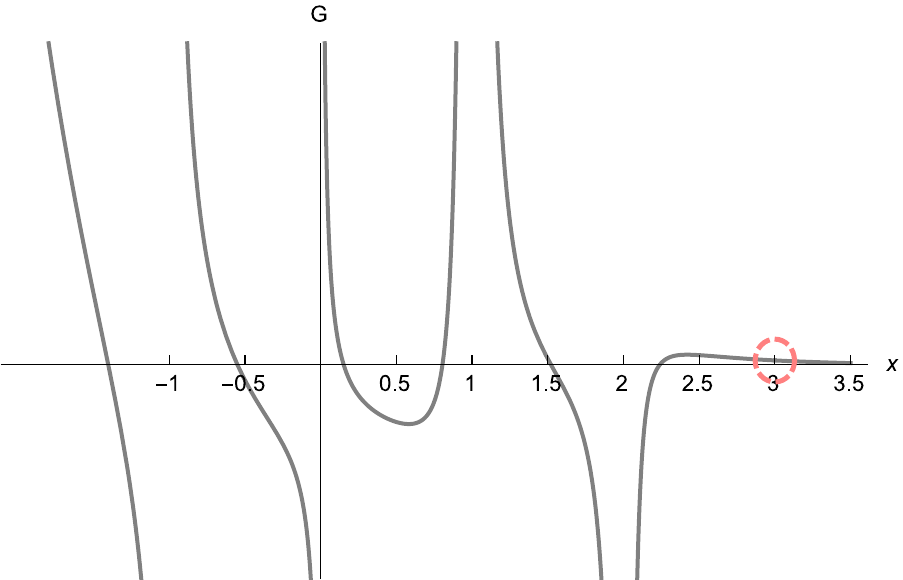}
    \caption{\(\e = 1, g \approx 1.01229, \Delta=3/2\)}
  \end{subfigure}
  \caption{Plot of \(G_\e(x;g,\Delta)\) for fixed \(g \) and \(\Delta\), corresponding to roots of constraint
    polynomials \(\cp{N,\e}{N}((2g)^2,\Delta^2)\). Notice the vanishing of the poles (indicated with dashed circles) at \(x = N + \e \) for
    $N = 1$  in (a) and (b), and $N=2$ in (c).}
  \label{fig:gfunct_graph}
\end{figure}

To conclude this subsection, we remark that there is a non-trivial relation between the parameters \(g,\Delta\)
and the pole structure of \(G_\e(x;g,\Delta) \), that is, the lateral limits at the poles for \(x \in \R\).

For instance, in Figure \ref{fig:gfunct_graph2} we show the plots of \(G_\e(x;g,\Delta)\) for fixed \(\Delta = 3/2, \e = 2 \) and
\(g = 1 \) (Figure \ref{fig:gfunct_graph2}(a)), a root \(g \approx 1.283 \) of \(\cp{2,2}{2}((2g)^2,(3/2)^2) \)
(Figure \ref{fig:gfunct_graph2}(b)) and \(g= 2 \) (Figure \ref{fig:gfunct_graph2}(c)). Note that in all cases the lateral
limits of the $G$-function \(G_\e(x;g,\Delta)\) at the pole \(x = 1\) are the same, while at the poles \(x = 2,3,4 \) the limits
have different signs in Figures \ref{fig:gfunct_graph2}(a) and (c). In addition, in Figure \ref{fig:gfunct_graph2}(b) the pole at
\(x = 4 \) actually vanishes. A deep understanding of this relation is crucial for the study of the distribution of eigenvalues of the AQRM,
for instance, the conjecture of Braak for the QRM \cite{B2011PRL} (see also Remark \ref{rem:gammafactor} below) and its possible generalizations to the AQRM.

\begin{figure}[htb]
  ~
  \begin{subfigure}[b]{0.45\textwidth}
    \centering
    \includegraphics[height=3.25cm]{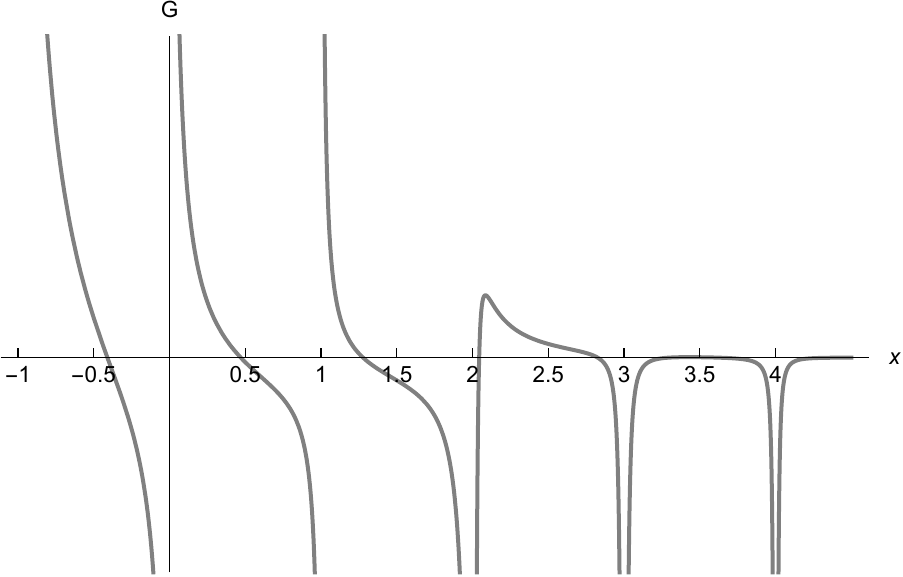}
    \caption{\(g=1\)}
  \end{subfigure}
  ~
  \begin{subfigure}[b]{0.45\textwidth}
    \centering
    \includegraphics[height=3.25cm]{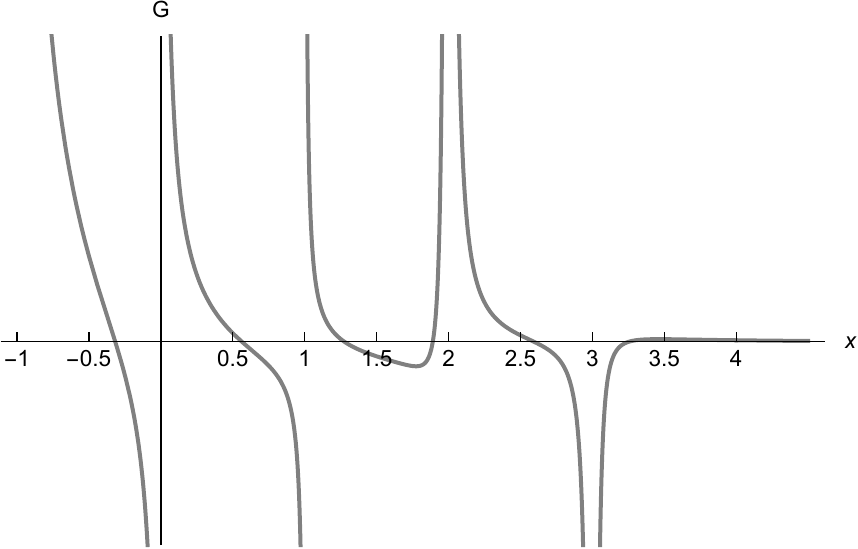}
    \caption{\( g \approx 1.283\)}
  \end{subfigure}
  \\
  \centering
  \begin{subfigure}[b]{0.45\textwidth}
    \centering
    \includegraphics[height=3.25cm]{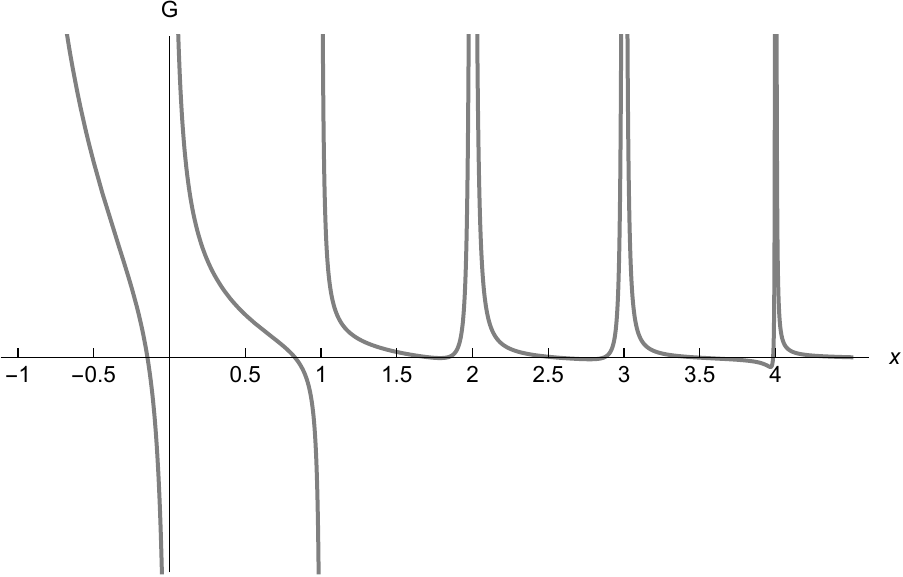}
    \caption{\(g = 2\)}
  \end{subfigure}
  \caption{Plot of \(G_\e(x;g,\Delta)\) for fixed \(\Delta=1.5, \e = 2\) and different values of \(g\).}
  \label{fig:gfunct_graph2}
\end{figure}

\begin{rem}
  The constraint \(T\)-function \(T_\e^{(N)}(g,\Delta)\) to be introduced in
  \S \ref{sec:NonJuddian} below, is defined in a similar manner to the $G$-function \(G_\e(x;g,\Delta)\).
  Thus, in addition to the references already given,
  we direct the reader to \S \ref{sec:NonJuddian} for the derivation and properties of the $G$-function.
\end{rem}

\section{Degeneracies of the spectrum and constraint polynomials}\label{sec:constPoly}

The constraint polynomials for the AQRM were originally defined by Li and Batchelor \cite{LB2015JPA}, following the work of Ku\'s \cite{K1985JMP} on the (symmetric) quantum Rabi model (QRM). In \cite{WY2014JPA,W2016JPA}, these polynomials were derived in the framework of finite-dimensional irreducible representations of $\mathfrak{sl}_2$ in the confluent Heun picture of the AQRM (see also \S \ref{sec:judd-solut-finite}). As we have seen in Proposition \ref{Prop:reccurEquiv}, the zeros of the constraint polynomials give the Juddian, or quasi-exact, solutions of the model. 

For reference, we recall the definition of the constraint polynomials (Definition \ref{def:cp}) and the associated three-term recurrence relation. For $ N \in \Z_{\geq 0}$, the polynomials \(\cp{N,\e}{k}(x,y)\) of degree $k$ are given by
\begin{align*}
  \cp{N,\e}{0}(x,y) &= 1, \qquad \cp{N,\e}{1}(x,y) = x + y - 1 - 2 \e, \\
  \cp{N,\e}{k}(x,y) &= (k x + y - k(k + 2 \e) ) P_{k-1}^{(N,\e)}(x,y) \\
                    &\qquad  - k(k-1)(N-k+1) x P_{k-2}^{(N,\e)}(x,y),
\end{align*}
for \(k \geq 2 \).

\begin{ex}
  For \(k=2,3\), we have
  \begin{align*}
    \cp{N,\e}{2}(x,y) &= 2 x^2 + 3 x y + y^2 - 2( N +  2(1+2\e)) x - (5+6\e) y + 4(1+3\e + 2\e^2), \\
    \cp{N,\e}{3}(x,y) &= 6 x^3 + 11 x^2 y + 6 x y^2 + y^3 - 6(2 N + 3(1+2\e))x^2 - 2(7+6\e) y^2 \\
                      &\quad{} - 2 (4 N +  17 + 22\e) x y  + 6 (2 N + 3(1+2\e))(2+2\e) x \\
                      &\quad{} + (49 + 4\e (24 + 11\e)) y - 6(1+2\e)(2+2\e)(3+2\e).
  \end{align*}
\end{ex}

When \(k = N\), the polynomial \(\cp{N,\e}{N}(x,y)\) is called {\em constraint polynomial}.
Actually, for a fixed value \(y = \Delta^2 \), if \(x = (2g)^2\) is a root of \( \cp{N,\e}{N}(x,y)\), then
\( \lambda =   N + \e - g^2 \) is an exceptional eigenvalue corresponding to a Juddian solution for \(\HRabi{\e}\).
Likewise, if \(x = (2g)^2\) is a root of \(\tilde{P}_N^{(N,\e)}(x,y):=\cp{N,-\e}{N}(x,y)\) (\cite{W2016JPA,RW2017}), then
\( \lambda =  N - \e - g^2  \) is an exceptional eigenvalue corresponding to a Juddian solution of \(\HRabi{\e}\). Mathematically,
the constraint polynomial \( \cp{N,\e}{N}(x,y)\) possesses certain particular properties not shared with \( \cp{N,\e}{k}(x,y)\)
with \(k \not= N\), these are studied in \S \ref{sec:detexp}.

The main objective is to prove the following conjecture.

\begin{conject}[\cite{W2016JPA}]\label{W2016JPA}
  \label{conj1A}
  For  \(\ell, N \in \Z_{\geq 0}\), there exists a polynomial \( A^\ell_N(x,y) \in \Z[x,y]\) such that
  \begin{align} \label{eq:conj}
   \cp{N+\ell,-\ell/2}{N+\ell}(x,y) =  A^\ell_N(x,y) \cp{N,\ell/2}{N}(x,y).
  \end{align}
  Moreover, the polynomial \( A^\ell_N(x,y)\) is positive for any $x, y > 0$. \QEDhere
\end{conject}

If the conjecture holds and \(g,\Delta >0 \) satisfy \(\cp{N,\ell/2}{N}((2g)^2,\Delta^2) = 0\), the exceptional eigenvalue
\( \lambda =  N  +  \ell/2 - g^2 ( = (N+\ell)  - \ell/2 - g^2) \) of \(\HRabi{\ell/2}\) is degenerate.

Actually, in order to complete the argument, it is necessary to show that the associated Juddian solutions are linearly independent.
The outline of the proof is as follows. The main point is that each root \(x,y>0\) of a constraint polynomial
\(\cp{N,\ell/2}{N}(x,y)\) determines an eigenvector in a finite dimensional representation space of  \(\mathfrak{sl}_2 \) (\(\rF_{2m} \) or \(\rF_{2m+1} \)  depending on the parity of \(N\) , cf. \S\ref{sec:repr-theor-pict}) associated with the exceptional eigenvalue \(\lambda = N  + \e - x \). For instance, suppose \(N = 2m \in \Z_{\geq 0}\) and \(\e = \ell \in \Z_{\geq 0}\) and that \(x,y>0\) make \(\cp{2 m,\ell}{2 m}(x,y)\) vanish. By Section 5.1 of \cite{W2016JPA} (see also \S\ref{sec:judd-solut-finite}), the eigenvalue \(\lambda = 2m + \ell - x \) has a corresponding eigenvector \(\nu \in \rF_{2m + 1} \). Moreover, under the
assumption of the conjecture, \(\cp{2m + 2\ell,-\ell}{2 m+2\ell}(x,y)\) vanishes as well, therefore the eigenvalue \(\lambda = 2m + \ell - x\)
also has the eigenvector \(\tilde{\nu} \in \rF_{2m + 2\ell} (= \rF_{2(m+\ell)}) \). For any \(\ell \in \Z\), it is clear that \( \rF_{2m + 1} \not\simeq \rF_{2(m+\ell)}\)
so the eigenvectors \(\nu \) and \(\tilde{\nu}\) (and hence the associated Juddian solutions) are linearly independent. The linear
independence in the remaining cases is shown in an completely analogous way, with the exception of the case \(\e = 1/2\), where we direct
the reader to Proposition 6.6 of \cite{W2016JPA} for the proof.

The condition \( A^\ell_N(x,y)>0\) of Conjecture \ref{conj1A} ensures that for \(N \in \Z_{\geq 0} \) there are no non-degenerate exceptional
eigenvalues \(\lambda = N  +  \ell/2 - g^2\) corresponding to Juddian solutions (see Corollary \ref{NoNonJudd} below).

We prove Conjecture \ref{conj1A} in two parts. We show the existence of the polynomial \(A_N^{\ell}(x,y)\) by showing that
\(\cp{N,\ell/2}{N}(x,y)\) divides \(\cp{N+\ell,-\ell/2}{N+\ell}(x,y)\) (as polynomials in \(\Z[x,y]\)) in \S \ref{sec:negeps}.
Additionally, this method gives an explicit determinant expression for the polynomial \(A_N^{\ell}(x,y)\). The proof is completed
in \S \ref{sec:pos} by studying the eigenvalues of the matrices involved in the determinant expressions for \(A_N^{\ell}(x,y)\).

\subsection{Determinant expressions of constraint polynomials} \label{sec:detexp}

It is well-known that orthogonal polynomials can be expressed as determinants
of tridiagonal matrices. Those determinant expressions are derived from
the fact that orthogonal polynomials satisfy three-term recurrence relations.
It is not difficult to verify that the polynomials \(\{\cp{N,\e}{k}(x,y)\}_{k\geq 0}\) do not
constitute families of orthogonal polynomials with respect to either of their variables (in a standard sense).
Nevertheless, since they are defined by three-term recurrence relations we can derive determinant expressions
using the same methods. We direct the reader to \cite{C1978} or \cite{K2008} for
the case of orthogonal polynomials.

Let \(N \in \Z_{\geq 0}\) and \(\e \in \R\) be fixed, by setting \(c^{(\e)}_k = k(k+2\e) \) and \(\lambda_k = k(k-1)(N-k+1)\) \((k \in \Z)\), the family of polynomials \(\{\cp{N,\e}{k}(x,y)\}_{k \geq 0}\) is given by the three-term recurrence relation
\[
  \cp{N,\e}{k}(x,y) = ( k x + y -  c^{(\e)}_k) \cp{N,\e}{k-1}(x,y)  - \lambda_k \cp{N,\e}{k-2}(x,y),
\]
for \(k \geq 2\), with initial conditions \(\cp{N,\e}{1}(x,y)= x + y - c^{(\e)}_1\) and
\(\cp{N,\e}{0}(x,y) = 1\). Hence, the polynomial \(\cp{N,\e}{k}(x,y)\) is the determinant of a \(k\times k\)
tridiagonal matrix
\begin{equation}\label{eq:detForm1}
\cp{N,\e}{k}(x,y)=\det(\mI_k y+\mA^{(N)}_kx+\mU^{(\e)}_k)
\end{equation}
where
\(\mI_k\) is the identity matrix of size \(k\) and
\begin{equation*}
\mA^{(N)}_k=\Tridiag{i}{0}{\lambda_{i+1}}{1\le i\le k},\quad
\mU^{(\e)}_k=\Tridiag{-c^{(\e)}_i}{1}{0}{1\le i\le k}.
\end{equation*}

In this section, bold font is reserved for matrices and vectors, subscript denotes the dimensions
of the square matrices and the superscript denotes dependence on parameters.

An important property of the constraint polynomial \(\cp{N,\e}{N}(x,y)\) is that it satisfies a determinant
expression with tridiagonal matrices that is different from \eqref{eq:detForm1}.

\begin{prop}
  \label{prop:detexp}
  Let \( N \in \Z_{\geq 0}\). We have
  \[
    \cp{N,\e}{N}(x,y) = \det\left( \mI_N y + \mD_N x + \mC_N^{(N,\e)} \right),
  \]
  where $\mD_N=\diag(1,2,\dots,N)$
and
\begin{equation*}
\mC^{(N,\e)}_N=\Tridiag{-i(2(N-i)+1+2\e)}{1}{i(i+1)c^{(\e)}_{N-i}}{1\le i\le N}.
\end{equation*}
\end{prop}

In order to prove Proposition \ref{prop:detexp}, we need the following lemma on the diagonalization
of the matrix \(\mA^{(N)}_k\) for \(k \in \Z_{> 0}\).

\begin{lem}
  \label{lem:eigenA}
  For \(1 \leq k \leq N\), the eigenvalues of \(\mA^{(N)}_k\) are \(\{1,2,\ldots,k\}\) and
  the eigenvectors are given by the columns of the lower triangular matrix \( \mE^{(N)}_k \)
  given by
  \[
    (\mE^{(N)}_k)_{i,j} = (-1)^{i-j}\binom ij\frac{(i-1)!(N-j)!}{(j-1)!(N-i)!},
  \]
  for \( 1 \leq i,j \leq k \).
\end{lem}

\begin{proof}
  We have to check that $(\mA^{(N)}_k\mE^{(N)}_k)_{i,j}=j(\mE^{(N)}_k)_{i,j}$ for every $i,j$.
  By definition, we see that
  \begin{align*}
  (\mA^{(N)}_k\mE^{(N)}_k)_{i,j}=j(\mE^{(N)}_k)_{i,j}
  &\iff (j-i)(\mE^{(N)}_k)_{i,j}=\lambda_i(\mE^{(N)}_k)_{i-1,j} \\
  &\iff (j-i)\binom ij=-i\binom{i-1}j,
  \end{align*}
  and the last equality is easily verified. 
\end{proof}

\begin{proof}[Proof of Proposition \ref{prop:detexp}]
  Since \((\mE^{(N)}_N)^{-1} \mA_{N}^{(N)} \mE^{(N)}_N = \mD_N \) by Lemma \ref{lem:eigenA},
  it suffices to prove that
  \begin{equation}
    \label{eq:UE=EC}
    \mU^{(\e)}_N \mE^{(N)}_N = \mE^{(N)}_N \mC^{(N,\e)}_N.
  \end{equation}
  Write $e_{ij}=(\mE^{(N)}_N)_{i,j}$ for brevity.
  The $(i,j)$-entry of $\mU^{(\e)}_N \mE^{(N)}_N-\mE^{(N)}_N \mC^{(N,\e)}_N$ is
  \begin{equation}\label{eq:UE=EC:2}
    -c^{(\e)}_ie_{ij}+e_{i+1,j}
    +j(2(N-j)+1+2\e)e_{ij}-e_{i,j-1}-j(j+1)c^{(\e)}_{N-j}e_{i,j+1}.
  \end{equation}
  Using the elementary relations
  \begin{equation*}
    j(j+1)c^{(\e)}_{N-j}e_{i,j+1}=-(i-j)(N-j+2\e)e_{ij},
  \end{equation*}
  \begin{equation*}
  e_{i+1,j}-e_{i,j-1}=(i^2+j^2+ij-j-iN-jN)e_{ij},
  \end{equation*}
  we immediately see that \eqref{eq:UE=EC:2} is equal to zero.
\end{proof}

Recall that the determinant \(J_n\) of a tridiagonal matrix
\begin{equation*}
  J_n=\det\Tridiag{a_i}{b_i}{c_i}{1\le i\le n}
\end{equation*}
is called {\em continuant} (see \cite{Muir1960}). It satisfies the
three-term recurrence relation
\begin{equation} \label{eq:continuant}
    J_n = a_{n} J_{n-1} - b_{n-1} c_{n-1} J_{n-2},
\end{equation}
with initial condition \(J_{-1} = 0 , J_0 = 1 \).
As a consequence of this, notice that the continuant equivalence
\begin{equation}
  \label{eq:propCont1}
  \det\Tridiag{a_i}{b_i}{c_i}{1\le i\le n} =\det\Tridiag{a_i}{b'_i}{c'_i}{1\le i\le n}
\end{equation}
holds whenever \(b_i c_i = b_i' c_i'\) for all \(i=1,2,\cdots,n-1 \), since the continuants on both sides of the equation define the same
recurrence relations with the same initial conditions.

\begin{cor}
  \label{cor:detsym}
  Let \( N \in \Z_{\geq 0}\). We have
  \[
    \cp{N,\e}{N}(x,y) = \det\left( \mI_N y + \mD_N x + \mS_N^{(N,\e)} \right),
  \]
  where \(\mD_N\) is the diagonal matrix of Proposition \ref{prop:detexp} and \(\mS_N^{(N,\e)}\)
  is the symmetric matrix given by
  \[
    \mS_N^{(N,\e)} = \Tridiag{-i(2(N-i)+1+2\e)}{\sqrt{i(i+1)c^{(\e)}_{N-i}}}{\sqrt{i(i+1)c^{(\e)}_{N-i}}}{1\le i\le N}.
  \]
\end{cor}

\begin{proof}
  Notice that the matrices \(\mI_N y + \mD_N x + \mC^{(N,\e)}_N\) and
  \(\mI_N y + \mD_N x + \mS_N^{(N,\e)} \) are tridiagonal. Then,
  it is clear by the continuant equivalence \eqref{eq:propCont1} that the determinants of the matrices are equal, establishing
  the result. 
\end{proof}

As a corollary to the discussion on the determinant expression \eqref{eq:detForm1}
we have the following result used in \S \ref{sec:pos} to prove the positivity of the
polynomial \(A_N^{\ell}(x,y)\).

\begin{cor}
  \label{cor:realrootsposx}
  For \(x \geq 0\), \(\e \in \R \) and \(N,k \in \Z_{\geq 0}\), all the roots of \(\cp{N,\e}{k}(x,y) \) with respect to \(y\)
  are real.
\end{cor}

\begin{proof}
  When \(x \geq 0\), using the continuant equivalence \eqref{eq:propCont1} on
  the determinant expression \eqref{eq:detForm1} of \(\cp{N,\e}{k}(x,y) \) we can find an equivalent expression
  \( \det( \mI_k y - \mV_k(x)) \) for a real symmetric matrix \(\mV_k(x)\). Since the roots of
  \(\cp{N,\e}{k}(x,y) \) with respect to \(y\) are the eigenvalues of the real symmetric matrix \(\mV_k(x) \),
  the result follows immediately. 
\end{proof}

In the case of the constraint polynomials \(\cp{N,\e}{N}(x,y)\), the determinant
expression of Corollary \ref{cor:detsym} gives the following result of similar
type, used for the estimation of positive roots of constraint
polynomials in \S \ref{sec:posRoots}.

\begin{thm}
  \label{thm:allrealroots}
  Let \( N \in \Z_{\geq 0}\) and \(\e > -1/2 \). Then, for fixed \(x \in \R\) (resp. \(y \in \R\)),
  all the roots of \(\cp{N,\e}{N}(x,y) \) with respect to \(y\) (resp. \(x\))
  are real.
\end{thm}

\begin{proof}
  Upon setting \(x = \alpha \in \R\), the zeros of \(\cp{N,\e}{N}(\alpha,y)\) are the
  eigenvalues of the matrix \(-(\mD_N \alpha + \mS^{(N,\e)}_N )\).
  For \(\e > -1/2 \), the matrix is real symmetric, so the eigenvalues,
  therefore the zeros, are real.
  The case of \(y = \beta \in \R \) is completely analogous
  since $\cp{N,\e}{N}(x,\beta)=\det\mD_N\det(\mI_N x +\mD_N^{-1}\beta +\mD_N^{-1/2}\mS_N^{(N,\e)}\mD_N^{-1/2})$. 
\end{proof}

The next example shows that we should not expect a determinant expression
of the type of Corollary \ref{cor:detsym} for general \(\cp{N,\e}{k}(x,y)\) with \(k \neq N\).

\begin{ex} \label{ex:nonrealroots}
  For a fixed \(y\), the roots of the polynomial
  \[
    \cp{6,0}{2}(x,y) = 2 x^2 +y^2 -16 x + 3 x y -5 y + 4,
  \]
  are given by
  \[
    \frac{1}{4} \left(16 -3 y  \pm \sqrt{y^2-56 y+224}\right).
  \]
  Clearly, the roots are not real for every value \(y \in \R\).
\end{ex}

\subsection{Divisibility of constraint polynomials} \label{sec:negeps}

In this subsection, we study the case where \(\e\) is a
negative half-integer. In this case, the determinant expressions
of \S \ref{sec:detexp} give the proof of the divisibility in Conjecture \ref{conj1A}.

First, from the determinant expression for \(\cp{N,\e}{N}(x,y)\) given in Corollary \ref{cor:detsym},
by means of the continuant equivalence \eqref{eq:propCont1} and elementary determinant operations it is not difficult to see that
\begin{align}\label{eq:detForm2}
  \cp{N,\e}{N}(x,y) = N! \det\left( \mI_N x + \mD_N^{-1} y + \mV_N^{(N,\e)} \right),
\end{align}
where
\[
  \mV_N^{(N,\e)} = \Tridiag{-2(N-i)-1-2\e }{\sqrt{c^{(\e)}_{N-i}}}{\sqrt{c^{(\e)}_{N-i}}}{1\le i\le N}.
\]
For \( N, \ell \in \Z_{\geq 0} \), the expression above reads
\begin{align}\label{eq:detform3}
  \cp{N+\ell,-\ell/2}{N+\ell}(x,y)
  = (N+\ell)! \det\left( \mI_{N+\ell} x + \mD_{N+\ell}^{-1} y + \mV_{N+\ell}^{(N+\ell,-\ell/2)} \right).
\end{align}
Noting that \(c_{\ell}^{(-\ell/2)} = \ell(\ell - \ell)= 0\), the matrix
\(\mV_{N+\ell}^{(N+\ell,-\ell/2)} \) has the block-diagonal form
\[
  \mV_{N+\ell}^{(N+\ell,-\ell/2)} =
  \begin{bmatrix}
    \mLV_{N}^{(N+\ell,-\ell/2)} & \mO_{N,\ell} \\
    \mO_{\ell,N} & \mRV_{\ell}^{(N+\ell,-\ell/2)}
  \end{bmatrix},
\]
where \(\mO_{n,m} \) is the \(n \times m \) zero matrix. Next, by setting
\[
  \mD^{(N)}_\ell = \diag\Bigl(\frac1{N+1},\frac1{N+2},\dots,\frac1{N+\ell}\Bigr),
\]
immediately it follows that
\begin{align*}
  \cp{N+\ell,-\ell/2}{N+\ell}(x,y)  =(N+\ell&)! \det\left( \mI_{N} x + \mD_{N}^{-1} y + \mLV_{N}^{(N+\ell,-\ell/2)} \right) \\
  &\times \det\left( \mI_{\ell} x + \mD^{(N)}_{\ell} y + \mRV_{\ell}^{(N+\ell,-\ell/2)} \right).
\end{align*}
For \(i =1,2,\ldots,N\), the \(i\)-th diagonal element of
\(\mLV_N^{(N+\ell,-\ell/2)} \) is \[-(2(N+1+\ell-i)-1 - \ell) = -(2(N+1-i) -1 +\ell) \]
and the off-diagonal elements are \(c_{N+\ell-i}^{(-\ell/2)} = c_{N-i}^{(\ell/2)}\). Therefore,
\begin{align*}
  \mLV_N^{(N+\ell,-\ell/2)} &= \mV_N^{(N,\ell/2)},
\end{align*}
and then, from \eqref{eq:detForm2} we have
\begin{align*}
  \cp{N+\ell,-\ell/2}{N+\ell}(x,y) = \cp{N,\ell/2}{N}(x,y) \frac{(N+\ell)!}{N!}
  \det\left( \mI_{\ell} x + \mD^{(N)}_{\ell} y + \mRV_{\ell}^{(N+\ell,-\ell/2)} \right).
\end{align*}
Let \(A_{N}^\ell(x,y) = \frac{(N+\ell)!}{N!} \det\left( \mI_{\ell} x + \mD^{(N)}_{\ell} y + \mRV_{\ell}^{(N+\ell,-\ell/2)} \right)\).
By expanding the determinant as a recurrence relation (cf. \eqref{eq:continuant}) or by appealing to Gauss' lemma, it is
easy to see that \( A_{N}^\ell(x,y)\) is a polynomial with integer coefficients. Therefore, the discussion above proves the
following theorem.

\begin{thm} \label{thm:div1}
  For \(N,\ell \in \Z_{\geq 0}\), there is a polynomial \(A_N^\ell(x,y) \in \Z[x,y]\) such that
  \[
    \cp{N+\ell,-\ell/2}{N+\ell}(x,y) = A_N^\ell(x,y) \cp{N,\ell/2}{N}(x,y).
  \]
  Furthermore, \(A_N^\ell(x,y) \) is given by
  \[
    \frac{(N+\ell)!}{N!}\det\Tridiag{x+\frac y{N+i}-\ell+2i-1}{\sqrt{c^{(\ell/2)}_{-i}}}{\sqrt{c^{(\ell/2)}_{-i}}}{1\le i\le \ell}.
  \]
\end{thm}

To complete the proof of Conjecture \ref{conj1A}, it remains to prove that
\(A_N^\ell(x,y) >0 \) for \(x,y>0\). This is done in \S \ref{sec:pos} below.

\begin{ex}[\cite{RW2017}] \label{ex:A1}
For small values of \(\ell\), the explicit form of \(A_N^\ell(x,y) \) is given by

\begin{align*}
  &A^1_N(x,y) = (N+1)x +y,   \\
  &A^2_N(x,y) = (N+1)_2 x^2 + \biggl(\sum_{i=1}^2 (N+i) \biggr)xy + y (1+ y), \\
  &A^3_N(x,y) = (N+1)_3 x^3  + \biggl(\sum_{i<j}^3 (N+i)(N+j) \biggr) x^2 y \\
  & \qquad \qquad \quad      + (N+2)x(3 y+4)y + y (2 + y)^2, \\
  &A^4_N(x,y)  = (N+1)_4 x^4 + \biggl( \sum_{i<j<k}^4 (N+i)(N+j)(N+j) \biggr)x^3 y\\
  &\, + \biggl(\sum_{i<j}^4 (N+i)(N+j) \biggr) x^2 y^2 + 2\biggl( \sum_{i<j}^4 (N+i)(N+j) - (N+2)(N+3) \biggr) x^2 y  \\
  &\,  + \biggl(\sum_{i=1}^4 (N+i) \biggr) xy(y+2)(y+3) + y(3 + y)^2(4 + y),
\end{align*}
where the symbol $(a)_n$ denotes the Pochhammer symbol, or raising factorial, that is
\( (a)_n:= a(a+1)\cdots (a+n-1)=\frac{\Gamma(a+n)}{\Gamma(a)} \) for $a\in\C$ and a non-negative integer $n$.

For a fixed degree \(\ell \in \Z_{\geq 0}\), the polynomial equation \(A_N^\ell(x,y)= 0\) defines certain algebraic curve depending
on the parameter \(N\): the case \(\ell = 2 \) is parabolic, \(\ell = 3 \) gives an elliptic curve and \(\ell = 4\) is super elliptic,
and so on.

For instance, let us consider the case \(\ell = 3\). Here, by using the change of variable $X=-x/y$ and $Y=1/y$,
the equation $A^3_N(x,y)=0$ turns out to be
\[
  4 Y^2 + 4 Y - 4(N+2) X Y = (N+1)_3 X^3 - (11 + 3N( N+4)) X^2 + 3(N+2) X - 1,
\]
which is easily seen to be (birationally) equivalent to the
elliptic curve in Legendre form (cf. \cite{Kob1984}).
\[
  Y_1^2 = X_1(X_1-1)\Big(X_1- \frac{(N+2)^2}{(N+1)(N+3)}\Big).
\]
with variables \(X_1 = \frac{X}{N+2}\) and \(Y_1 = \frac{(N+2)}{\sqrt{(N+1)(N+3)}} (2 Y  - (N+2) X + 1) \).
\end{ex}

\subsection{Proof of the positivity of \texorpdfstring{$A_N^{\ell}(x,y) $}{A(x,y)}}\label{sec:pos}

In this subsection we complete the proof of Conjecture \ref{conj1A} by proving the positivity of
the polynomial \(A_N^\ell(x,y)\) for \(x,y>0\). Let \( N \in \Z_{\geq 0}\) and \( \ell \in \Z_{>0}\) be fixed.
From Theorem \ref{thm:div1} and the continuant equivalence \eqref{eq:propCont1}, we see that the polynomial \(A_N^\ell(x,y)\) has
the determinant expression
\[
  \frac{(N+\ell)!}{N!} \det( \mD^{(N)}_{\ell} y + \mB_{\ell}(x) )
\]
where \(\mB_{\ell}(x)  \) is an matrix-valued function given by
\begin{equation}
  \label{eq:matbx}
  \mB_{\ell}(x) = \Tridiag{x-\ell+2i-1}{1}{c^{(\ell/2)}_{-i}}{1\le i\le \ell}.
\end{equation}

Next, multiplying the \(\frac{(N+\ell)!}{N!}\) factor into the determinant in such a way that the
\(i\)-th row  is multiplied by \(N+i\), we obtain the expression
\begin{equation}
  \label{eq:Adetexp}
  A_N^\ell(x,y) = \det(\mI_\ell y + \mM^{(N)}_\ell(x)) {}= \prod_{\lambda\in\Spec(\mM^{(N)}_\ell(x))}(y+\lambda)
\end{equation}
with
\begin{equation*}
  \mM^{(N)}_\ell(x) = \Tridiag{(N+i)(x-\ell+2i-1)}{N+i}{(N+i+1)c^{(\ell/2)}_{-i}}{1\le i\le \ell}.
\end{equation*}

Thus, it suffices to show that all the eigenvalues of \(\mM^{(N)}_\ell(x)\) are positive for \(x > 0 \) to prove that
$A_N^\ell(x,y)>0$ when $x,y>0$.

First, we compute the determinant of the matrix \(\mM^{(N)}_\ell(x)\), or equivalently, the value of \(A_N^{\ell}(x,0)\).

\begin{lem}
  \label{lem:detM}
  We have
  \[ \det(\mM^{(N)}_\ell(x)) = A_N^{\ell}(x,0)  =\frac{(N+\ell)!}{N!} x^{\ell}.\]
\end{lem}

\begin{proof}
  Consider the recurrence relation
  \[
    J_i(x) =  (x+ \ell + 1 - 2 i) J_{i-1}(x) + (i-1)(\ell+1-i) J_{i-2}(x),
  \]
  with initial conditions $J_0(x) = 1$ and \(J_{-1}(x) = 0 \). Notice that this
  recurrence relation corresponds to the continuant \(\det \mB_{\ell}(x)\) (compare with \eqref{eq:matbx} above)
  and therefore, \(\frac{(N+\ell)!}{N!} J_\ell (x) = \frac{(N+\ell)!}{N!} \det \mB_{\ell}(x) = \det(\mM^{(N)}_\ell(x)) \).
  We claim that \(J_i(x) = \sum_{j=0}^i (\ell-i)_j \binom{i}{j} x^{i-j} \). Clearly, the claim holds
  for \(J_0(x) = 1 \) and \(J_1(x) = x + \ell -1 \). Assuming it holds for integers up to a fixed \(i\), then
  \(J_{i+1}(x)\) is given by
  \begin{align*}
    &(x + \ell - 1- 2 i) \sum_{j=0}^{i} (\ell-i)_j \binom{i}{j} x^{i-j} + i(\ell - i) \sum_{j=0}^{i-1} (\ell -i + 1)_j \binom{i-1}{j} x^{i-1-j} \\
    &\qquad=  \sum_{j=0}^{i} (\ell-i)_j \binom{i}{j} x^{i+1-j} + (\ell - 1- 2 i) \sum_{j=0}^{i} (\ell-i)_j \binom{i}{j} x^{i-j} \\
    &\qquad \qquad{}+ i(\ell - i) \sum_{j=0}^{i-1} (\ell -i + 1)_j \binom{i-1}{j} x^{i-1-j},
  \end{align*}
  by grouping the terms in the sums we obtain
  \begin{multline*}
    x^{i+1} + (\ell -i -1) x^i + (\ell-i -1)_{i+1} \\
    {} + \sum_{j=1}^{i-1} (\ell-i)_j \left( (\ell-i+j) \binom{i}{j+1} + (\ell-1-2 i)\binom{i}{k} + j \binom{i}{j} \right) x^{i-j}.
  \end{multline*}
  The sum on the right is
  \begin{multline*}
    \sum_{j=1}^{i-1} (\ell-i)_j \binom{i}{j} \left( \frac{(\ell-i+j)(i-j)}{j+1} + \ell -1 - 2 i + j \right)  x^{i-j} \\
    =  \sum_{j=1}^{i-1} (\ell-i)_j \binom{i}{j} \left(\frac{(i+1)(\ell -i -1)}{j+1}\right) x^{i-j}
    =  \sum_{j=2}^{i} (\ell -i - 1)_j \binom{i+1}{j} x^{i+1-j},
  \end{multline*}
  and the claim follows by joining the remaining terms into the sum. Finally, notice that
  \(J_\ell (x) = \sum_{j=0}^i (0)_j \binom{\ell}{j} x^{\ell-j} = x^\ell \), as desired. 
\end{proof}

\begin{rem}
  The lemma above is a generalization of the case \(N = 0 \) studied in \cite{RW2017} (Prop. 4.1) using
  continued fractions. It would be interesting to study the combinatorial properties of the coefficients of the polynomials $A_N^{\ell}(x,y)$
  using the determinant expressions given above.
\end{rem}

From the lemma above, we immediately obtain the following result.

\begin{cor}
  \label{cor:zeroeigen}
  For \(N \in \Z_{\geq 0}\), the eigenvalue \(\lambda =0\) is in  \(\Spec(\mM^{(N)}_\ell(x))\) if and only if \(x = 0\).
\end{cor}

The next result collects some basic properties of the eigenvalues of the matrix \(\mM^{(N)}_\ell(x)\) that
are used in the proof of the positivity of \(A_N^{\ell}(x,y)\).

\begin{lem}
  \label{lem:eigenprop}
  Denote the spectrum of the matrix \(\mM^{(N)}_\ell(x)\) by \(\Spec(\mM^{(N)}_\ell(x))\).
  \begin{enumerate}[$(1)$]
  \item For \(x \geq 0 \), the eigenvalues \( \lambda \in \Spec(\mM^{(N)}_\ell(x))\) are real.
  \item We have \( \Spec(\mM^{(N)}_\ell(0)) = \{ i(\ell-i) : i=1,2,\cdots,\ell \}\).
    In particular, \( 0 \in \Spec(\mM^{(N)}_\ell(0)) \) is a simple eigenvalue and
    any eigenvalue \(\lambda \in \Spec(\mM^{(N)}_\ell(0)) \) satisfies \( \lambda \geq 0\).
  \item If \( x' > \ell-1\), all eigenvalues  \( \lambda \in \Spec(\mM^{(N)}_\ell(x'))  \)
    satisfy \( \lambda > 0\).
  \end{enumerate}
\end{lem}

\begin{proof}
  Note that by Corollary \ref{cor:realrootsposx} and the divisibility of Theorem \ref{thm:div1}, if \(x \geq 0 \) all the roots of \(A_N^\ell(x,y)\) with
  respect to  \(y\) are real. By definition, the same holds for the elements of \(\Spec(\mM^{(N)}_\ell(x))\), proving the first claim.
  From the defining recurrence relation, we see that \(\cp{N,\e}{N}(0,y) = \prod_{i=1}^{N}(y - i(i+2\e))\), and by divisibility we have
  \(A_N^{\ell}(0,y) = \prod_{i=1}^\ell(y - i(i - \ell)) \) proving the second claim. For the third claim, notice that when \(x' > \ell-1 \) all
  the diagonal elements of \(\mM^{(N)}_\ell(x')\) are positive. Therefore, the continuant \eqref{eq:Adetexp} defines a recurrence relation with
  positive coefficients, so that \(A_N^{\ell}(x',y)\) is a polynomial in \(y\) with positive coefficients and real roots. Since \(y=0 \) is not
  a root of \(A_N^{\ell}(x',y)\) by Corollary \ref{cor:zeroeigen}, all of the roots of \(A_N^{\ell}(x',y)\) must be negative and the third claim follows. 
\end{proof}

With these preparations, we come to the proof of the positivity of the polynomial \(A_N^{\ell}(x,y)\)

\begin{thm} \label{thm:pos}
  With the notation of Theorem \ref{thm:div1}, \(A_N^{\ell}(x,y)> 0 \) for \(x,y>0\).
\end{thm}

\begin{proof}
  By virtue of \eqref{eq:Adetexp}, it is enough to show that all the eigenvalues of $\mM^{(N)}_{\ell}(x)$ are positive if $x>0$.
  Notice that each eigenvalue of $\mM^{(N)}_{\ell}(x)$ is a real-valued continuous function in $x$. Assume that there is a positive $x'$
  such that $\mM^{(N)}_{\ell}(x')$ has a negative eigenvalue. Then, there also exists $x''$ such that $x'<x''<\ell$ and $0\in\Spec(\mM^{(N)}_{\ell}(x''))$
  since all eigenvalues of $\mM^{(N)}_{\ell}(\ell)$ are positive by Lemma \ref{lem:eigenprop} (3). This contradicts to
  Corollary \ref{cor:zeroeigen}. 
\end{proof}

The proof of Conjecture \ref{conj1A}, which we reformulate as a theorem below, is completed by Theorems \ref{thm:div1} and \ref{thm:pos}.

\begin{thm} \label{thm:Main}
  For  \(\ell, N \in \Z_{\geq 0}\), there exists a polynomial \( A^\ell_N(x,y) \in \Z[x,y]\) such that
 \begin{align} \label{eq:conj:proved}
   \cp{N+\ell,-\ell/2}{N+\ell}((2g)^2,\Delta^2) =  A^\ell_N((2g)^2,\Delta^2) \cp{N,\ell/2}{N}((2g)^2,\Delta^2).
 \end{align}
 for \(g,\Delta >0 \). Moreover, the polynomial \( A^\ell_N(x,y)\) is positive for any $x, y > 0$. \QEDhere
\end{thm}

A consequence of the positivity of \( A^\ell_N(x,y)\) in Theorem \ref{thm:Main} is that all the positive roots of the constraint
polynomials \(\cp{N,\ell/2}{N}(x,y)\) and \(\cp{N+\ell,-\ell/2}{N+\ell}(x,y) \) (\(N,\ell \in \Z_{\geq 0}\)) must coincide.
This explains the fact that the two curves defined by the constraint polynomials
in Figure \ref{fig:constraintgraph} appear to coincide when \(\e = \ell/2\, (\ell \in \Z)\).

Note that since \(A_0^{\ell}(x,y) = \cp{\ell,-\ell/2}{\ell}(x,y)\) and \(\cp{0,\ell/2}{0}(x,y) = 1 \not= 0\), the positivity of
\(A_0^{\ell}(x,y)\) also implies the nonexistence of Juddian eigenvalues \(\lambda =  \ell/2 - g^2 \) for \(\ell > 0\). In fact, the
positivity can be extended to a larger set of constraint polynomials \( \cp{k,-\ell/2}{k}(x,y) \).

\begin{prop} \label{prop:pos2}
  Let \(\ell \in \Z_{>0} \) and \( 1 \leq  k \leq \ell\). Then the constraint polynomial \( \cp{k,-\ell/2}{k}(x,y) \) is positive
  for \(x,y>0 \).
\end{prop}

\begin{proof}
  For \( 1 \leq k \leq \ell \), define the \(k \times k\) matrix
  \begin{equation*}
    \mathbf{M}_k(x) = \Tridiag{x+\ell-1-2(k-i)}{i}{(i+1)c^{(-\ell/2)}_{k-i}}{1 \le i \le k}
  \end{equation*}
  then \( \cp{k,-\ell/2}{k}(x,y) =  \det(\mathbf{I}_k y + \mathbf{M}_k(x)) \) and the roots of \(\cp{k,-\ell/2}{k}(x,y) \)
  with respect to \(y\) are the eigenvalues of the matrix \(-\mathbf{M}_k(x)\). Thus, as in the case of \(A_N^{\ell}(x,y)\), it
  suffices to prove that all the eigenvalues of \(\mathbf{M}_k(x)\) are positive for \(x >0\).

  First, from \eqref{eq:detForm2}, we see that \( \det(\mathbf{M}_k(x)) = \cp{k,-\ell/2}{k}(x,0) = k! \sum_{j=0}^k (\ell-k)_j \binom{k}{j} x^{k-j}\).
  Indeed, we verify that \( \det(\mathbf{M}_k(x)) = k! J_k(x)\), where \(\{ J_i(x)\}_{i \geq 0} \) is the recurrence relation defined
  in Lemma \ref{lem:detM}. In particular, \(\det(\mathbf{M}_k(x)) \) is a polynomial with positive coefficients and thus it
  never vanishes for \(x >0 \).

  Next, we verify that the matrix \(\mathbf{M}_k(x)\) has the properties of the matrices \(\mathbf{M}^{(N)}_\ell(x)\) given in Lemma
  \ref{lem:eigenprop}. From Corollary \ref{cor:realrootsposx}, it is clear that for \(x \geq 0\) the eigenvalues of
  \(\mathbf{M}_k(x)\) are real. By the definition of the constraint polynomials, it is obvious that
  \( \Spec(\mathbf{M}_k(0)) = \{ i(\ell-i) : i=1,2,\cdots,k \}\), hence any eigenvalue \(\lambda \in \Spec(\mathbf{M}_k(0)) \) is non-negative.
  Finally, as in the proof of Lemma \ref{lem:eigenprop}, we see that for \(x' > \max(0,2 k - \ell - 1)\) all eigenvalues
  \(\lambda \in \Spec(\mathbf{M}_k(x')) \) satisfy \(\lambda > 0 \).

  The proof of positivity then follows exactly as in the proof of Theorem \ref{thm:pos}. 
\end{proof}

  \subsection{Degeneracy in the spectrum of AQRM}
\label{sec:degen-eigenv-aqrm}

The results on divisibility of constraint polynomials and the confluent Heun picture of the AQRM allow us to 
to fully characterize the degeneracies in the spectrum of the AQRM.

We begin by restating Proposition \ref{prop:pos2} in terms of Juddian eigenvalues of AQRM. This result eliminates
the possibility of Juddian eigenvalues of multiplicity 1 for the case \(\e = \ell/2 \) with \(\ell \in \Z \).

\begin{cor} \label{cor:noNegJudd}
  For \(\ell \in \Z_{>0} \) and \( 0 \leq k \leq \ell  \) there are no Juddian eigenvalues \(\lambda = k - \ell/2 - g^2 \) in \(\HRabi{\ell/2}\).
\end{cor}

\begin{proof}
  The case \(k=\ell \) was already proved in Theorem \ref{thm:pos} and the case \(k=0\) is trivial since \(\cp{0,-\ell/2}{0}((2g)^2,\Delta^2) = 1 \not= 0 \).
  For \( 1 \leq  k < \ell \), if \( \lambda = k - \ell/2 - g^2 \) is a Juddian eigenvalue then \(\cp{k,-\ell/2}{k}((2g)^2,\Delta^2)=0 \) for some parameters
  \(g,\Delta >0\). This is a contradiction to Proposition \ref{prop:pos2}. Note that in this case there is no possibility of a contribution of Juddian
  eigenvalues by roots of constraint polynomials \(\cp{N,\ell/2}{N}((2g)^2,\Delta^2) \) as this would necessarily require \(N = k - \ell < 0\). 
\end{proof}

\begin{rem}
  In Proposition 5.8  of \cite{W2016JPA}, it is shown that the roots of the constraint polynomials \(\cp{N,\e}{N}(x,y)\) are simple.
  In particular, this implies that for \(\e \not\in \frac12 \Z \), there are no degenerate exceptional eigenvalues
  consisting of two Juddian solutions.
\end{rem}

Since the multiplicity of the eigenvalues is at most two, as a corollary of the divisibility
in Theorem \ref{thm:Main} and Corollary \ref{cor:noNegJudd}, we have the following result.

\begin{cor}\label{NoNonJudd}
  If $x=(2g)^2$ is a root of the equation $P_N^{(N, \ell/2)}(x, \Delta^2)=0$, then the (Juddian) eigenvalue $\lambda= N+\ell/2-g^2$
  must be a degenerate exceptional eigenvalue. In fact, the  multiplicity of the exceptional
  eigenvalue $\lambda$ is exactly $2$ and the two linearly independent solutions are Juddian (see Figure \ref{fig:excepteigen2}(b)). \QEDhere
\end{cor}

\begin{rem}\label{no-non-Juddian}
  What the corollary means is, although a non-Juddian exceptional eigenvalue may exist on the energy curve
  $E=N+\ell -\e-g^2$ (resp. $E=N+\e-g^2$) (see Figure \ref{fig:excepteigen2}(a)) for $0<|\e-\ell/2|<\delta$ for sufficiently
  small $\delta$, as the numerical result in \cite{LB2016JPA} suggests, the non-Juddian exceptional eigenvalues
  disappear when $\e=\ell/2 \in \Z_{\geq0}$ and the exceptional eigenvalue $\lambda:=E$ is Juddian.
\end{rem}

We are now in a position to describe the general structure of the degeneracy of the spectrum of the AQRM.

\begin{cor}\label{DegenerateStructure}
  The degeneracy of the spectrum of $\HRabi{\e}$ occurs
  only when \(\e = \ell/2 \) for $\ell \in \Z_{\geq0}$ and $P_N^{(N,\ell/2)}((2g)^2,\Delta^2)=0$. In particular,
  any non-Juddian exceptional solution is non-degenerate.
\end{cor}

\begin{proof}
  We first consider the case \(N \ne 0 \).
  When $P_N^{(N,\e)}((2g)^2,\Delta^2)\ne0$ if we look at the local Frobenius solutions at $y=0$,
  then there is always a local solution containing a $\log$-term as seen in \S \ref{sec:smallexp}
  (see Proposition \ref{Prop:reccurEquiv}),
  so the solutions corresponding to the smaller exponent cannot be components of the eigenfunction.
  Then, the solution corresponds to the largest exponent (i.e. non-Juddian exceptional)
  and this implies that the dimension of the corresponding eigenspace is at most one
  (cf. \cite{B2011PRL,ZXGBGL2014} and also \S \ref{sec:NonJuddian}).
  We note that in the case \(\e = \ell/2 \, (\ell \in \Z)\) there is no chance of a contribution of Juddian solution
  (i.e. $P_{N+\ell}^{(N+\ell,-\ell/2)}((2g)^2,\Delta^2) =0$) by Theorem \ref{thm:Main}.
  Suppose next that $P_N^{(N,\e)}((2g)^2,\Delta^2)=0$ for $\e\notin \frac12\Z_{\geq0}$.
  Looking at the local Frobenius solutions at $y=1$,
  since the exponent different from $0$ is not a non-negative integer (see Table \ref{tab:exp}),
  we observe that only the solution corresponding to the exponent $0$ can give a eigensolution of $\HRabi{\e}$
  so that the dimension of the eigenspace is also at most one.
  By Corollary \ref{NoNonJudd}, there is no non-Juddian exceptional eigensolution
  when $P_N^{(N,\ell/2)}((2g)^2,\Delta^2)=0$ for $\ell \in \Z_{\geq0}$.

  On the other hand, if $N=0$, the exponents of the system \eqref{eq:systemN} are $\rho_1^- = 0$ and $ \rho_2^- = 1$,
  therefore there is one holomorphic Frobenius solution and a local solution with a \(\log\)-term.
  This implies that the corresponding eigenstate cannot be degenerate.
  In addition, note that if $K^{(N,\e)}_0(g, \Delta)=0$,
  the $\log$-term in the Frobenius solution with smaller exponent \eqref{smallerFrobenius} vanishes
  making it identical to the solution \eqref{eq:largexpsol-} (corresponding to the larger exponent).
  Hence, the exceptional eigenvalue $\lambda = \pm\e-g^2$ must be non-Juddian exceptional, and thus, non-degenerate.
  Since $\cp{0,\pm \e}{0}((2g)^2,\Delta^2) = 1 \not=0$   and \(\cp{\ell,-\ell/2}{\ell}((2g)^2,\Delta^2) \not= 0 \) for \(g,\Delta>0 \)
    and \(\ell > 0 \) (cf. Proposition \ref{prop:pos2}), the desired claim follows. 
\end{proof}

\begin{rem}
The non-degeneracy of the ground state for the QRM is shown in \cite{HH2012}.
\end{rem}

Thus, summarizing the results so far obtained in Theorem \ref{thm:div1}
with Theorem \ref{thm:pos} and Corollary \ref{DegenerateStructure},
we have the following result.

\begin{thm}\label{thm:AQRMSpec}
  The spectrum of the AQRM possesses a degenerate eigenvalue if and only if the parameter $\e$ is a half integer.
  Furthermore, all degenerate eigenvalues of the AQRM are Juddian. \QEDhere
\end{thm}

We conclude by illustrating numerically the degeneracy structure of the spectrum of the AQRM
described in Theorem \ref{thm:AQRMSpec} (for the numerical computation of spectral curves see Theorem \ref{thm:LBcomp}).
For half-integer \(\e\), Figure \ref{fig:specgraph} shows the spectral graphs for fixed \(\Delta=1\) and \(\e =0,1/2,3/2\).
In the graphs, the dashed lines represent the exceptional energy curves \(y = i + \ell/2 - g^2\) for \(i \in \Z_{\geq 0}\),
any crossings of these curves with the spectral curves correspond to exceptional eigenvalues. The crossings of the eigenvalue
curves in the exceptional points correspond to Juddian degenerate solutions, given by Theorem \ref{thm:pos}.
Notice also the non-degenerate exceptional points in the curves, these points correspond to the non-Juddian exceptional eigenvalues.

\begin{figure}[htb]
  ~
  \begin{subfigure}[b]{0.30\textwidth}
    \centering
    \includegraphics[height=4.5cm]{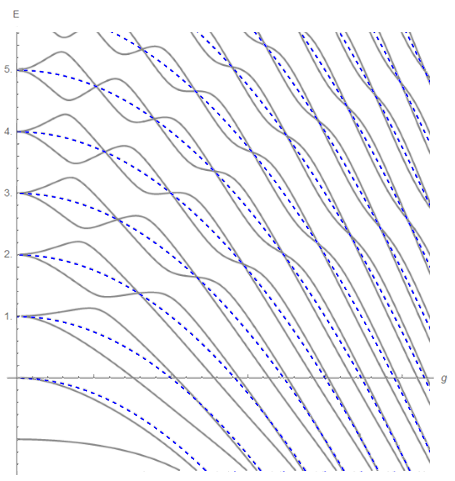}
    \caption{\(\e = 0\)}
  \end{subfigure}
  ~
  \begin{subfigure}[b]{0.30\textwidth}
    \centering
    \includegraphics[height=4.5cm]{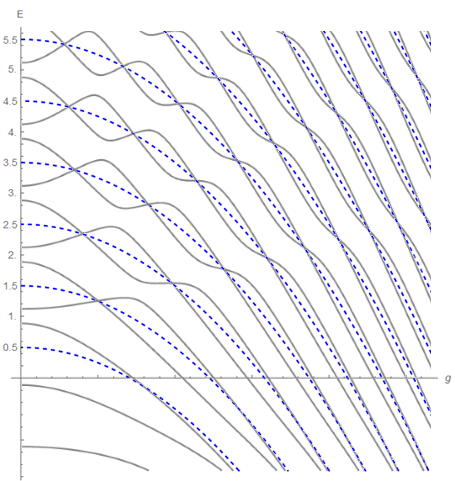}
    \caption{\(\e = 0.5 \)}
  \end{subfigure}
  ~
  \centering
  \begin{subfigure}[b]{0.30\textwidth}
    \centering
    \includegraphics[height=4.5cm]{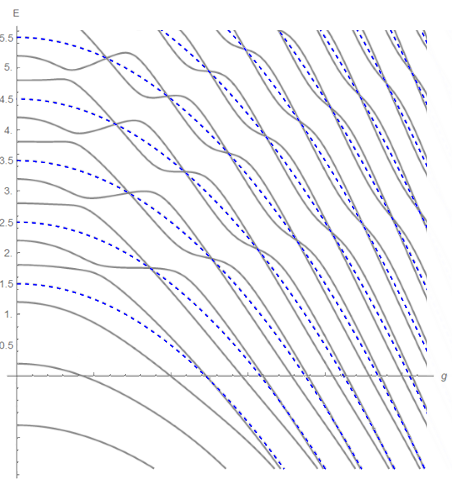}
    \caption{\(\e = 1.5\)}
  \end{subfigure}
  \caption{Spectral curves for the case of \(\Delta=1\) for the cases \(\e \in \{0,0.5,1.5\}\) for \(0 \leq g \leq 2.7\) and energy (\(E\)) \(-1.5 \leq E \leq 5.5 \) .}
  \label{fig:specgraph}
\end{figure}

The case of \(\e \not\in \frac12 \Z\) is shown in Figure \ref{fig:specgraph_nhi}.
In these graphs, for \(i \in \Z_{\geq 0}\) the dashed lines represent the exceptional energy curves \(y = i \pm \e - g^2\).
Notice that we have the situation of the conceptual graphs of Figure \ref{fig:excepteigen2} in the introduction.
In particular, note that due to the bounds on positive solutions of constraint polynomials of \S \ref{sec:interlacing},
not all exceptional eigenvalues \(\lambda = N \pm \e - g^2\) with the same \(N \in \Z_{\geq 0}\) can be Juddian
(see also the discussion on Figure \ref{fig:constraintgraph3} below).

\begin{figure}[htb]
  ~
  \begin{subfigure}[b]{0.45\textwidth}
    \centering
    \includegraphics[height=5cm]{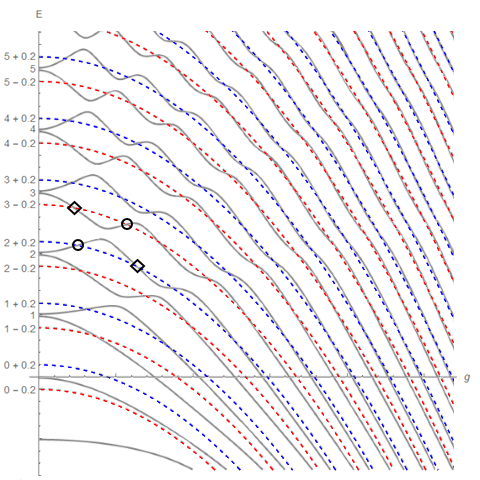}
    \caption{\(\e = 0.2\)}
  \end{subfigure}
  ~
  \begin{subfigure}[b]{0.45\textwidth}
    \centering
    \includegraphics[height=5cm]{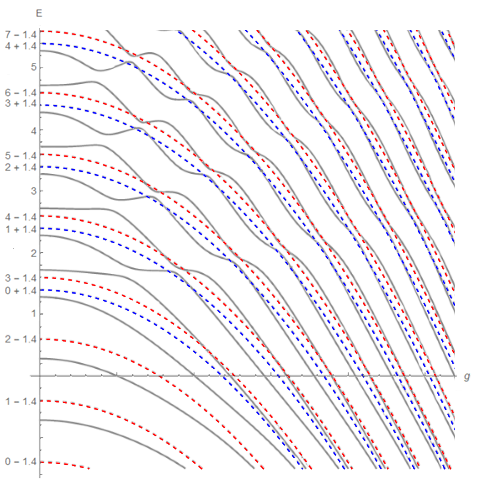}
    \caption{\(\e = 1.4 \)}
  \end{subfigure}
  \caption{Spectral curves for the case of \(\Delta=1\) for the cases \(\e \in \{0.2,1.4\}\) for \(0 \leq g \leq 2.7\) and energy (\(E\)) \(-1.5 \leq E \leq 5.5 \).
     In (a), circle marks denote points corresponding to Juddian solutions and diamond marks denote non-Juddian
      exceptional solutions (cf. Figure~\ref{fig:ExceptEigen}(a)). }
  \label{fig:specgraph_nhi}
\end{figure}

\begin{rem}
  The mathematical model known as the non-commutative harmonic oscillator (NcHO) \cite{PW2001} (see \cite{P2010S} for a detailed study
  and information of the NcHO with references therein, and \cite{P2014Milan} for a recent development) is given by
  \[
    Q = Q_{\alpha,\beta} =
    \begin{pmatrix}
      \alpha & 0\\
      0 &   \beta
    \end{pmatrix}
    \left( - \frac{1}{2} \frac{d^2}{d x^2} + \frac{1}{2} x^2 \right) +
    \begin{pmatrix}
      0 & -1\\
      1 &   0
    \end{pmatrix}
    \left( x \frac{d}{d x} + \frac{1}{2}  \right).
  \]
  The NcHO is a self-adjoint ordinary differential operator with a \(\Z_2\)-symmetry that generalizes the quantum harmonic oscillator by
  introducing an interaction term. When the parameters \(\alpha,\beta >0 \) satisfy \(\alpha \beta > 1\), the Hamiltonian $Q$ is positive definite, whence
  it has only positive (discrete) eigenvalues. It is known \cite{W2013PJA} that the multiplicity of the eigenvalues  is at most $2$. Moreover,
  the possibilities that an eigenstate of $Q$ is degenerate ($2$ dimensional) are the following two cases
  \cite{W2015IMRN}:
  \begin{itemize}
  \item a quasi-exact (Juddian) solution and a non-Juddian solution
    with the same parity (in this case the eigenvalue $\lambda$ is of the form $\lambda=2\frac{\sqrt{\alpha\beta(\alpha\beta-1)}}{\alpha+\beta}(m+\frac12)$ for some $m\in \Z_{\geq0}$),
  \item two non-quasi-exact solutions with different parity.
  \end{itemize}
  There is a close connection between the NcHO and the quantum Rabi model \cite{W2015IMRN}, arising from their representation theoretical pictures
  via a confluent process for the Heun ODE (see \cite{SL2000}). It is desirable to clarify the reason concerning the
  difference of the structure of the degeneracies between the NcHO and the QRM (also AQRM for $\e\in \frac12 \Z_{\geq0}$: see Theorem \ref{thm:AQRMSpec} in
  \S \ref{sec:LargestExponent}). Actually, the degeneracies occur only for quasi-exact solutions in both models and  those are considered to be
  remains of the eigenvalues of the quantum harmonic oscillator. Therefore, it is quite interesting to develop a similar discussion for constraint polynomials in the
  former ``exceptional" case for the NcHO in \cite{W2015IMRN}.
\end{rem}

\subsection{The degenerate atomic limit} \label{sec:orthopoly}

In this subsection we make a brief remark on the case \(\Delta = 0 \) from the orthogonal polynomials viewpoint. Recall that the polynomials
\(\cp{N,\e}{k}(x,y)\) are defined by a three-term recurrence relation. However, it is not possible to set the parameters to define a family of orthogonal polynomials in \(x\) or \(y\).

Consider the determinant expression \eqref{eq:detForm2} and set \(y=0\). The expansion of the continuant from the lower-right corner gives
the three-term recurrence relation
\begin{align*}
  \frac{1}{k!}\cp{k,\e}{k}(x,0)
  &= (x +1 - 2 k - 2\e ) \frac{1}{(k-1)!}\cp{k-1,\e}{k-1}(x,0) \\
  &\quad{}- (k-1) (k-1 + 2\e) \frac{1}{(k-2)!} \cp{k-2,\e}{k-2}(x,0).
\end{align*}
By Favard's theorem (see, e.g.  \cite{C1978}), when \(\e > -\frac12 \) the family of normalized constraint polynomials
\(\{\frac{1}{k!}\cp{k,\e}{k}(x,0) \}_{k \geq 0} \) defines an orthogonal polynomial system. Recall that the generalized Laguerre polynomials
\cite{AAR1999} are given by
\[
  L^{(\alpha)}_k (x) = \frac{x^{-\alpha}e^x}{k!} \frac{d^k}{d x^k} (e^{-x}x^{k+\alpha}).
\]
for \(k \geq 1\) and \(\alpha > -1 \), and the monic generalized Laguerre polynomials are given by \((-1)^k k! L^{(\alpha)}_k (x)\). Comparing the recurrence
relations and the initial conditions we immediately obtain the following result.

\begin{thm}\label{thm:Laguerre}
  For \(k \geq 0 \), we have
  \[
    \frac{1}{k!}\cp{k,\e}{k}(x,0) = (-1)^k k! L^{(2\e)}_k (x). \QEDhere
  \]
\end{thm}

The case \(y =0\) corresponds to the model \(\HRabi{\e}\) with \(\Delta= 0\), namely
\begin{equation} \label{eq:degatomlimit}
  H^\e := \omega a^\dag a +g\sigma_x(a^\dag+a) + \e \sigma_x,
\end{equation}
called the degenerate atomic limit in \cite{LB2015JPA}. The Hamiltonian \(H^\e\) is a generalization of the displaced harmonic oscillator
(corresponding to \(H^0\)) studied in \cite{Sc1967AP}. For the Hamiltonian \(H^\e\), the constraint equation for the exceptional eigenvalue
parameterized by integer \(N\) is given by
\[
  L^{(2\e)}_N (x) = 0.
\]
The presence of the Laguerre polynomials in the constraint equation is explained in the study of the solutions of the model \eqref{eq:degatomlimit}.
For instance, in \cite{Sc1967AP}, the solutions of the displaced harmonic oscillator is given in terms of power series, its coefficients are
multiples of associated Laguerre polynomials. The explicit form of the exceptional solutions of \eqref{eq:degatomlimit} is obtained in
\cite{LB2015JPA} by a related method.

\begin{rem}
  For the case \(y \not= 0 \), let \(\mC_N^{(N,\e)}\) be the matrix in the determinant expression of \(\cp{N,\e}{N}(x,y)\) of Proposition
  \ref{prop:detexp}, then we have
  \[
    (\mC_N^{(N,\e)})_{k,k} = -k(2(N+1-k) -1 + 2\e), \quad \, (\mC_N^{(N,\e)})_{i,i-1} = i (N+1-i)(N+1-i+2\e)
  \]
  for \(k =1,2,\cdots,N\) and \(i =2,3,\cdots,N\). By expanding the continuant as a recurrence relation we obtain a family
  \(\{Q^{(N,\e)}_k(x,y)\}_{k\geq 0}\) of polynomials in two variables  given by
  \begin{align*}
    Q^{(N,\e)}_0(x,y) =& 1, \quad Q^{(N,\e)}_1(x,y) = x+y -  (2 N -1  + 2\e)  \\
    Q^{(N,\e)}_k(x,y) =& (k x + y  - k (2(N+1-k) -1 + 2\e)) Q^{(N,\e)}_{k-1}(x,y) \\
                      & - k (k-1)(N+1-k)(N+1-k+2\e) Q^{(N,\e)}_{k-2}(x,y).
  \end{align*}
  By definition  $Q^{(N,\e)}_N(x,y) = \cp{N,\e}{N}(x,y) $. In general, for \(k \neq N\), it does not hold that
  $Q^{(N,\e)}_k(x,y) = \cp{N,\e}{k}(x,y)$. Moreover, in contrast with the case \(y=0 \),
  it is not clear how to relate the $k$-th polynomial $Q^{(N,\e)}_k(x,y)$ with the constraint polynomial $\cp{k,\e}{k}(x,y)$.
\end{rem}
\,

\section{Estimation of positive roots of constraint polynomials} \label{sec:estimation}

In this section, we study existence of degenerate exceptional eigenvalues corresponding to Juddian solutions
by giving an estimate on the number of positive roots \(x=(2g)^2\) for the constraint polynomial \(\cp{N,\e}{N}(x,y)\) according
on the value of  \(y = \Delta^2 \). For our current interest concerning the degeneracy of Juddian solutions, it is sufficient to obtain
the estimate when $\e\geq0$. However, we give also a conjecture which counts precisely a number of positive roots of \(\cp{N,\e}{N}(x,y)\)
for negative $\e$ when $N$ is sufficiently large, i.e. $N\geq -[2\e]$, $[x]$ being the integer part of $x\in \R$.

\subsection{Interlacing of roots for constraint polynomials} \label{sec:interlacing}

When considered as polynomials in \(\R[y][x]\), there is non-trivial interlacing among the roots of the coefficients of the
constraint polynomials \(\cp{N,\e}{N}(x,y)\). This interlacing is essential for the proof of the upper bound on the number of
positive roots of the constraint polynomials in the next sections.

For \(N \in \Z_{\geq 0}\), let
\[
  \cp{N,\e}{N}(x,y) = \sum_{i=0}^{N} a^{(N)}_i(y) x^i.
\]
Noticing that \(\deg(a^{(N)}_i(y)) = N-i \), the interlacing property is given in the following lemma.

\begin{lem}
  \label{lem:rootcoeff}
  Let \( N \in \Z_{\geq 0}\) and \(\e > -1/2 \). Then the roots of \(a^{(N)}_j(y) \, ( 0 \leq j \leq N-1) \) are real. Denote the roots
  of \(a^{(N)}_j(y)\) by \(\xi_1^{(j)} \leq \xi_2^{(j)} \leq \cdots \leq \xi_{N-j}^{(j)}  \). Then, for \(j = 0,1,\ldots,N-2\) we have
  \[
    \xi_{i}^{(j)} < \xi_{i}^{(j+1)} < \xi_{i+1}^{(j)}
  \]
  for \( i = 1,2,\ldots, N-j-1 \).
\end{lem}

The constraint polynomials \(\cp{N,\e}{N}(x,y) \), with \(\e > -\frac{1}{2}\), belong to a special class of polynomials in two
variables, the class \(\Ptwo\) (see \cite{F2008}). The class \(\Ptwo\) is a generalization of polynomials of one variable with
all real roots. A polynomial \(p(x,y)\) of degree \(n\) belongs to the class \(\Ptwo\) if it satisfies the following conditions:

\begin{itemize}
\item For any \( \alpha \in \R \), the polynomials  \(p(\alpha,y)\) and \(p(x,\alpha)\) have
  all real roots.
\item Monomials of degree \(n\) in \(p(x,y)\) all have positive coefficients.
\end{itemize}
Equivalently, a polynomial \(p(x,y)\) is in the class  \( \Ptwo \) if it has a determinant expression
\[
  p(x,y) = \det\left( \mI_n y + \mD_n x + \mS_n \right),
\]
with \(\mD_n\) a diagonal matrix with positive entries and \(\mS_n\) a real symmetric matrix.

Recall the following property of polynomials of the class \(\Ptwo\).
\begin{lem}[Lemma 9.63 of \cite{F2008}]
  Let \(f(x,y) \in \Ptwo \) and set
\begin{equation*}
    f(x,y) = f_0 (x) + f_1(x) y + \cdots + f_n(x) y^n.
\end{equation*}
If \(f(x,0) \) has all distinct roots, then all \(f_i\) have distinct roots, and the roots of \(f_i \) and \(f_{i+1}\) interlace.
\end{lem}
Note that the lemma above tacitly implies that the roots of the polynomials \(f_i\) are real. With these preparations, we
prove Lemma \ref{lem:rootcoeff}.

\begin{proof}[Proof of Lemma \ref{lem:rootcoeff}]
  By Corollary \ref{cor:detsym}, \(\cp{N,\e}{N}(x,y) \in \Ptwo\).
  Since \[\cp{N,\e}{N}(0,y) = \prod_{i=i}^{N}(y-i(i+2\e)), \] for \(\e > -1/2 \), the roots
  are different and the lemma applies, establishing the result. 
\end{proof}

\subsection{Number of positive roots of constraint polynomials} \label{sec:posRoots}

In this section we give an estimation on the number of positive roots of constraint polynomials. In particular, this result proves
the existence of exceptional eigenvalues corresponding to Juddian solutions in the spectrum of the AQRM.
We note that although there is a description of the statement of Theorem \ref{LB} for open intervals in
\cite{LB2016JPA}, the proof provided by the authors only gives a lower bound on the number of positive roots.

\begin{thm}\label{LB}
  Let \( \e > -\frac12 \). For each \(k \, (0\leq k <N)\), there are exactly $N-k$ positive roots (in the variable $x$) of the
  constraint polynomial \(\cp{N,\e}{N}(x,y)\) for $y$ in the range
  \begin{equation*}
     k(k+2\e) \leq y < (k+1)(k+1 +2\e).
  \end{equation*}
  Furthermore, when \(y \geq N(N + 2\e)\), the polynomial \(\cp{N,\e}{N}(x,y)\) has no positive roots with respect to \(x\).
\end{thm}

We illustrate numerically the proposition for the case \(N=6 \) and \(\e=0.4\) in Figure \ref{fig:lbgraph}. For fixed \(\Delta>0\)
satisfying \(k(k+2\e) \leq \Delta^2 < (k+1)(k+1 +2\e)\) (\(k \in \{1,2,\ldots,N\}\)), the number of points \((g,\Delta) \) with \(g>0\) in the curve
\(\cp{N,\e}{N}((2g)^2,\Delta^2)=0 \) is exactly \(N - k\). Likewise, as it is clear in the figure, there are no points \((g,\Delta) \) in
the curve with \(g>0\) and \(\Delta^2 \geq N(N+2\e)\).

\begin{figure}[htb]
  \centering
  \includegraphics[height=5cm]{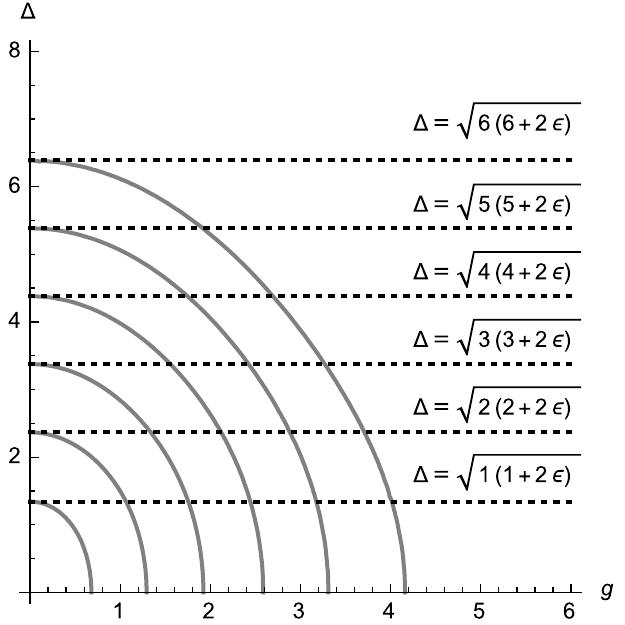}
  \caption{Curve \( \cp{6,\e}{6}((2g)^2,\Delta^2)=0 \) with \(\e = 0.4 \) (for \(g,\Delta >0\))}
  \label{fig:lbgraph}
\end{figure}

First, we establish a lower bound on the number of positive roots for the constraint polynomials. The following Lemma extends Li and
Batchelor's result (\cite{LB2016JPA}, Theorem), to the case of semi-closed intervals.

\begin{lem}
  \label{lem:LBJPA}
  Let \( \e > -\frac12 \). For each \(k \, (0\leq k <N)\), there are at least $N-k$ positive roots
  (in the variable $x$) of the constraint polynomial \(\cp{N,\e}{N}(x,y)\) for $y$ in the range
  \begin{equation*}
     k(k+2\e) \leq y < (k+1)(k+1 +2\e).
  \end{equation*}
\end{lem}

\begin{rem}
  The proof is a modification to the argument given in \cite{LB2016JPA} (Appendix B), which is based on the proof of  Ku\'s for the case of
  the (symmetric) quantum Rabi model (\cite{K1985JMP}, Section IV, Thm. 3).
\end{rem}

\begin{proof}
  Define the normalized polynomials \(S_k^{(N,\e)}(x,y)\) by
  \[
    S_k^{(N,\e)}(x,y) = \frac{\cp{N,\e}{k}(x,y)}{k!}.
  \]
  Fix \(y\) and consider the polynomials \(S_k^{(N,\e)}(x,y)\) as polynomials in the variable \(x\) and write \(S_k^{(N,\e)}(x)\) for
  simplicity. Set \(\alpha_i = (i(i + 2\e) - y)/i\) and \(\beta_i = N-i+1\), then the recurrence relation becomes
  \begin{align}
    \label{eq:recnormLBK}
    S_0^{(N,\e)}(x) &= 1, \quad S_1^{(N,\e)}(x) = x - \alpha_1 \nonumber \\
    S_k^{(N,\e)}(x) &= (x - \alpha_k ) S_{k-1}^{(N,\e)}(x) - \beta_k x S_{k-2}^{(N,\e)}(x).
  \end{align}
  Let \(k \, (0\leq k < N)\) be fixed. If \(k(k+2\e) < y < (k+1)(k+1 +2\e)\), then it is clear that \(\alpha_i < 0\) for \(i < k \), \(\alpha_i > 0 \) for
  \(i >k \) and \( \beta_i >0  \) for \(0 \leq i < N\). Moreover, when \( y =  k(k+2\e)\) we have \(\alpha_k = 0\) and, from \eqref{eq:recnormLBK},
  we see that \(x = 0\) is a root of all polynomials \(S_{k+i}(x) \) for \(i=1,\ldots,N-k\).

  For \(i=0,1,\ldots,N-k\), set
  \[
    \tilde{S}_{k+i}(x) =
    \begin{cases}
      S_{k+i}(x) &\mbox{if } y \neq k(k+2\e)  \\
      (1/x)S_{k+i}(x) & \mbox{if } y = k(k+2\e)
    \end{cases}.
  \]
  With this modification, the proof follows as in \cite{K1985JMP}. First, notice that
  \begin{align*}
    \sgn(S_l^{(N,\e)}(0)) &= \sgn((-1)^l \alpha_1 \alpha_2 \ldots, \alpha_l) = (-1)^{2l} = 1
  \end{align*}
  for \(l < k\). Similarly, \(\sgn(S_k^{(N,\e)}(0)) = 1\) if \(y \not= k(k+2\e) \) and \(\sgn(S_k^{(N,\e)}(0)) = 0 \) if \(y = k(k + 2\e )\).
  On the other hand, for \(i=1,\ldots,N-k\), we have
  \begin{align*}
    \sgn(\tilde{S}_{k+i}^{(N,\e)}(0)) =
    \begin{cases}
      \sgn((-1)^{k+i} \alpha_1 \ldots \alpha_{k-1} \alpha_k \alpha_{k+1} \ldots \alpha_{k+i} ) &  \text{ if } y \not= k(k+2\e) \\
      \sgn((-1)^{k+i-1} \alpha_1 \ldots \alpha_{k-1} \alpha_{k+1} \ldots \alpha_{k+i} ) & \text{ if } y = k(k+2\e)
    \end{cases},
  \end{align*}
  and we directly verify that in both cases the expression is equal to \( (-1)^{i} \).
  
  In addition, from the recurrence relation \eqref{eq:recnormLBK} we easily see the following
  \begin{itemize}
  \item if \(S_{i}^{(N,\e)}(a) = 0\) for \(a > 0 \), then \(S_{i+1}^{(N,\e)}(a)\)
    and \(S_{i-1}^{(N,\e)}(a)\) have opposite signs,
  \item \(S_{i}^{(N,\e)}(x)\) and \(S_{i-1}^{(N,\e)}(x)\) cannot
    have the same positive root.
  \end{itemize}
  These remarks are easily seen to hold for the auxiliary polynomials \(\tilde{S}_{k+i}(x)\) as well. Next, denote by \(V(x)\) the number
  of change of signs  of the sequence
  \[
    \tilde{S}_N^{(N,\e)}(x), \tilde{S}_{N-1}^{(N,\e)}(x), \ldots, \tilde{S}_{k+1}^{(N,\e)}(x), S_{k}^{(N,\e)}(x),
    S_{k-1}^{(N,\e)}(x), \ldots, S_0^{(N,\e)}(x).
  \]
  By the remarks above, variations of \(V(x)\) by \(\pm 1\) occur only at zeros of \(\tilde{S}_N^{(N,\e)}(x)\) or \(\tilde{S}_0^{(N,\e)}(x) = 1\).
  At \(x=0\), the first terms of the sequence are \( (-1)^{N-k-i}\) for \(i=0,\ldots,N-k-1\), then \(0\) if \( y = k(k+2\e)\) and all the remaining
  terms are \( 1\), hence \(V(0)= N-k \). On the other hand, it is clear that as \(x\) tends to infinity \(\sgn(S^{(N,\e)}_i(x))= 1\) and
  \(\sgn(\tilde{S}^{(N,\e)}_{k+i}(x))= 1\). This proves that there are at least \(N-k\) positive roots of the polynomial \(\tilde{S}_N^{(N,\e)}(x)\) and
  the same holds for \(\cp{N,\e}{N}(x)\). 
\end{proof}

To complete the proof we give an upper bound to the number of positive roots using Descartes' rule of
signs (see e.g. \cite{K2008}, Theorem 7.5). This result states that the number of positive roots
of a polynomial does not exceed the number of the sign changes in its coefficients.

\begin{lem} \label{lem:upperBN}
  Let \( \e > -\frac12 \). When \(y \geq N(N + 2\e)\), the polynomial \(\cp{N,\e}{N}(x,y)\) has no positive
  roots with respect to \(x\).
\end{lem}

\begin{proof}
  First, using the notation of \S \ref{sec:interlacing}, we note that \(y = N(N + 2\e)\) is the largest
  root of \(a_0^{N}(y) = \cp{N,\e}{N}(0,y)\). Then, by the interlacing of the roots of \(a_i^{N}(y)\) (\(i=0,1,\ldots,N-1\)) of
  Lemma \ref{lem:rootcoeff}, all \(a_i^{N}(y)\) must be non-negative. Thus, there are no changes of signs
  in the coefficients of \(\cp{N,\e}{N}(x,y)\) (as a polynomial in \(x\)) and the result follows by
  Descartes' rule of signs. 
\end{proof}

\begin{lem}
  \label{lem:Upperbound}
  Let \( \e > -\frac12 \). For each $k$ $(0\leq k <N)$, there are at most $N-k$ positive roots
  (in the variable $x$) of the constraint polynomial \(\cp{N,\e}{N}(x,y)\)
  for $y$ in the range
  \begin{equation*}
     k(k+2\e) \leq y < (k+1)(k+1 +2\e).
  \end{equation*}
\end{lem}

\begin{proof}
  First, using the notation of \S \ref{sec:interlacing}, note as in Lemma \ref{lem:upperBN} that
  when \(y \geq N(N + 2\e)\), all the coefficients \(a_i^{(N)}(y)\) of the polynomial \(\cp{N,\e}{N}(x,y)\) are
  non-negative.

  By Lemma \ref{lem:rootcoeff}, for
  \[(N-1)(N-1+2\e) < y < N(N + 2\e),\]
  the sign sequence \((\sgn a_N^{(N)}(y), \sgn a_{N-1}^{(N)}(y), \ldots, \sgn a_0^{(N)}(y))\) is given by
  \[
    +,+,\cdots,+,-,-,\cdots,-,-,
  \]
  that is, it consists of a subsequence \(+,+,\ldots,+\) of positive signs followed by a subsequence \(-,-,\ldots,- \) of negative
  signs. Thus,  by Descartes' rule of signs we have at most \(1 = N - (N-1)\) positive roots for \(\cp{N,\e}{N}(x,y)\).
  When \(y = (N-1)(N-1+2\e)\), we have \(a_0^{(N)}(y)=0 \) and the sequence is the same except for a \(0\) at the end,
  so the result holds without change. Continuing this process, we see that for
  \((N-2)(N-2+2\e) < y < (N-1)(N-1 + 2\e)\), the sign sequence given by
  \[
    +,+,\cdots,+,-,-,\cdots,-,+,+,\ldots,+
  \]
  from where it holds that the polynomial hast at most \(2 = N-(N-2)\) roots (with respect to \(x\)).
  We continue this process until we reach \(0 < y < 1(1 + 2\e)\), where we have
  \[
    +,-,+,-,\ldots,(-1)^{N-1},(-1)^{N}
  \]
  giving \(N = N- 0\) roots (with respect to \(x\)) by Descartes' rule of signs. Therefore, to complete the proof
  it is enough to show that the number of sign changes in the sequence
  \((\sgn a_N^{(N)}(y),\sgn a_{N-1}^{(N)}(y),\ldots,\sgn a_0^{(N)}(y))\) does not vary for \(y\) satisfying
  \((k-1)(k-1+2\e) < y < k(k + 2\e)\), and that there is exactly an additional sign change when \(y\) crosses
  \((k-1)(k-1+2\e)\). To see this, note that due to the interlacing of roots given in Lemma \ref{lem:rootcoeff}, the next
  sign change in a subsequence \(+,+,\cdots,+\) (or \(-,-,\cdots,-\)) of contiguous coefficients with same sign must happen at right
  end of the subsequence. When the subsequence \(+,+,\cdots,+\) (or \(-,-,\cdots,-\)) is at the rightmost end of the complete sign sequence
  \((\sgn a_N^{(N)}(y),\sgn a_{N-1}^{(N)}(y),\ldots,\sgn a_0^{(N)}(y))\) there is an additional sign change in the complete sequence and
  the sign change occurs at roots of \(a_0(y)\), that is, when \( y = k(k+2\e)\) for \(k \in \{1,2,\ldots,N-1\} \). In any other case there
  is no additional sign change. This completes the proof. 
\end{proof}

The combination of Lemmas \ref{lem:LBJPA} and \ref{lem:Upperbound} immediately gives Theorem \ref{LB}.

\begin{rem}
  It would be interesting to obtain a result where we switch the roles of $g$ and $\Delta$ in Theorem \ref{LB}
  from both mathematical (e.g. orthogonal polynomials of two variables) and physics (experimental and/or applications) points of view.
\end{rem}

\subsection{Negative \texorpdfstring{$\e$}{epsilon} case}

In this subsection we give some remarks on the estimation of positive roots for \(\e < 0\). First, we present the generalization
of Lemma \ref{lem:LBJPA}. Here, \(\floor{x}\) denotes the integer part of \(x\), that is, the unique integer \(\floor{x}\) such
that \( \floor{x} \leq x < \floor{x} + 1 \), and \(\fract{x} = x - \floor{x}\) is the fractional part of \(x\).

\begin{lem}
  \label{lem:LBJPAneg}
  Let \( \e < 0\) and set \(m = - \floor{2\epsilon} \). For each \(k\) \((0\leq k <N)\), there are at least $N-k$ positive roots
  (in the variable $x$) of the constraint polynomial \(\cp{N+m,\e}{N+m}(x,y)\) for $y$ in the range
  \begin{equation*}
     (m+k)(m+k+2\e) \leq y < (m+k+1)(m+k+1 +2\e).
   \end{equation*}
   Furthermore, let \(\alpha_1 \le \alpha_2 \leq \cdots \leq \alpha_m \) be the elements of the multiset \(\{ j(j+2 \e) :  1 \leq j \leq m \}\).
   Then, for \(k\) \((1 \leq k < m)\) such that \(\alpha_k \not= \alpha_{k+1}\), and \(y\) in the range
   \begin{equation*}
     \alpha_k \leq y < \alpha_{k+1},
   \end{equation*}
   the constraint polynomial \(\cp{N+m,\e}{N+m}(x,y)\) has at least
   \(N+m - k\) positive roots. If \( y < \alpha_1\), then the constraint
   polynomial \(\cp{N+m,\e}{N+m}(x,y)\) has exactly \(N+m\) roots.
\end{lem}

\begin{proof}
  The proof is done in the same manner as in Lemma \ref{lem:LBJPA} by
  replacing \(k\) with \(m+k\), and by noticing that for \(i=1,2,\ldots, m-1\),
  it holds that  \(\alpha_i <0 \) for any \(y \geq 0\). 
\end{proof}

\begin{rem}
  The numbers \((m+k)(m+k+2\e) \), for \(k\) \((0\leq k <N)\), in the proposition are also given
  by \((k - \floor{2\e})(k + \fract{2\e})\).
\end{rem}

The case of negative half-integer \(\e\) of Theorem \ref{LB} follows directly from Theorem \ref{thm:pos}.

\begin{cor}
  \label{cor:negepshalfinteger}
  Let \(N \in \Z_{\geq 0} \) and \(\ell \in \Z_{>0}\) satisfying \( N - \ell \geq 0\).
  For each \(k (0\leq k <N)\), there are exactly $N-k$ positive roots
  (in the variable $x$) of the constraint polynomial \(\cp{N+\ell,-\ell/2}{N+\ell}(x,y)\)
  for $y$ in the range
  \begin{equation*}
     k (\ell+k) \leq y < (k+1)(\ell+k+1).
   \end{equation*}
   Furthermore, when \( y \geq N (N+\ell) \), the polynomial
  \(\cp{N+\ell,-\ell/2}{N+\ell}(x,y)\) has no positive roots with respect to \(x\).
\end{cor}

\begin{rem}
  Recall that for the case \( N - \ell < 0\) the constraint polynomials \(\cp{N+\ell,-\ell/2}{N+\ell}(x,y)\)
  have no positive roots (cf. Proposition~\ref{prop:pos2}).
\end{rem}

\begin{proof}
  Since \(\cp{N+\ell,-\ell/2}{N+\ell}(x,y) = A_N^{\ell}(x,y) \cp{N,\ell/2}{N}(x,y)\), the result
  follows immediately from Theorem \ref{LB} and the fact that \(A_N^{\ell}(x,y)\) has
  no positive roots for \(x,y>0\). 
\end{proof}

In the case of non-half integral \(\e< 0\), there is no analog of the polynomial
\(A_N^{\ell}(x,y)\). Nevertheless, we expect that the following conjecture holds.

\begin{conject}
  Suppose \(\e < 0 \) and set \(m = \max(0,- \floor{2\epsilon}) \).
  For each \(k \, (0\leq k <N)\), there are exactly $N-k$ positive roots
  (in the variable $x$) of the constraint polynomial \(\cp{N+m,\e}{N+m}(x,y)\)
  for $y$ in the range
  \begin{equation*}
     (m+k)(m+k+2\e) \leq y < (m+k+1)(m+k+1 +2\e).
   \end{equation*}
   Furthermore, when \( y \geq (N+m)(N+m +2\e) \), the polynomial
  \(\cp{N+m,\e}{N+m}(x,y)\) has no positive roots with respect to \(x\).
\end{conject}

\section{Further discussion on the spectrum of AQRM}\label{sec:excepteigen}

In this section, we give additional results related to the spectrum of the AQRM. In particular, we define the
constraint function for non-Juddian exceptional eigenvalues, the constraint \(T\)-function \(T_\e^{(N)}(g,\Delta)\).
In addition, by studying the $G$-function and its poles, we define a new $G$-function $\mathcal{G}_\e(x;g,\Delta)$ that
captures the complete spectrum of AQRM.

\subsection{Non-Juddian exceptional solutions} \label{sec:NonJuddian}\label{Non-Juddian}

In this subsection, we study the constraint relation for non-Juddian exceptional eigenvalues.
Recall from \S \ref{sec:Gfunct} that the zeros of the \(G\)-function \(G_{\e}(x;g,\Delta)\)
corresponds to points of the regular spectrum \(\lambda = x - g^2 \).
Similarly, zeros of the constraint polynomial \(\cp{N,\e}{N}(x,y)\)
correspond to exceptional eigenvalues \(\lambda = N + \e  - g^2\) with Juddian solutions.

For non-Juddian exceptional eigenvalues, we define a constraint \(T\)-function \(T_\e^{(N)}(g,\Delta)\)
that vanishes for parameters \(g\) and \(\Delta\) for which \(\HRabi{\e}\) has the exceptional
eigenvalue \(\lambda = N + \e  - g^2\) with non-Juddian solution (see \cite{B2013AP} for the case
of the quantum Rabi model).

In order to define the function \(T_\e^{(N)}(g,\Delta)\), we first describe the local Frobenius solutions of system
of differential equations \eqref{eq:system1} and \eqref{eq:system2p} at the regular singular points $y=0,1$
(cf. \S \ref{sec:LargestExponent}).

Define the functions \(\phi_{1,\pm}(y;\e)\) as follows:
\begin{align}
  \phi_{1,+}(y;\e)&=\frac{(N+1)}{\Delta} y^N - \Delta \sum_{n=N+1}^{\infty} \frac{\bar{K}^-_n(N+\e;g,\Delta,\e)}{n-N} y^n, \label{eq:solphi1+} \\
  \phi_{1,-}(y;\e)&=\sum_{n=N+1}^{\infty} \bar{K}^-_n(N+\e;g,\Delta,\e) y^n, \label{eq:solphi1-}
\end{align}
with initial conditions \(\bar{K}^-_n(N+\e;g,\Delta,\e) = 0 \, (n \leq N)\), \(\bar{K}^-_{N+1}(N+\e;g,\Delta,\e) = 1\) and
\begin{align*}
  &(n+1) \bar{K}^-_{n+1}(N+\e;g,\Delta,\e) \\
  &\, =  \left(n - N + (2g)^2 - 2\e + \frac{\Delta^2}{N-n}   \right) \bar{K}^-_{n}(N+\e;g,\Delta,\e) - (2g)^2 \bar{K}^-_{n-1}(N+\e;g,\Delta,\e),
\end{align*}
for \(n \geq N+1\). Then,  ${}^t(\phi_{1,+}(y;\e),\phi_{1,-}(y;\e))$ is the local Frobenius solution corresponding to the
largest exponent of the system \eqref{eq:system1} at $y=0$.

Next, consider the solutions at $y=1$. For the case \(N + 2 \e \not\in \Z_{\geq 0} \) (i.e. \(\e \not\in \frac12 \Z\)
or \( \e = -\ell/2 \, (\ell \in \Z_{\geq 0})\) and \(N - \ell < 0 \) ) we define
\begin{align}
  \phi_{2,+}(\bar{y};-\e)&= \Delta\sum_{n = 0}^{\infty} \frac{\bar{K}_n^+(N+\e;g,\Delta,\e) }{N+2\e-n}\bar{y}^n, \\
  \phi_{2,-}(\bar{y};-\e)& = \sum_{n= 0}^{\infty} \bar{K}_n^+(N+\e;g,\Delta,\e) \bar{y}^n,
\end{align}
with initial conditions \(\bar{K}^+_n(N+\e;g,\Delta,\e) = 0\,(n<0)\), \(\bar{K}^+_0(N+\e;g,\Delta,\e) = 1\), while for the case  \(N + 2 \e \in \Z_{\geq 0} \) (i.e.
\(\e = \ell/2 \, (\ell \in \Z_{\geq 0})\) or \(\e = -\ell/2 \, (\ell \in \Z_{\geq 0})\) and \(N - \ell \geq 0 \)) we define
\begin{align}
  \phi_{2,+}(\bar{y};-\ell/2)&= \frac{(N+\ell+1)}{\Delta} \bar{y}^{N+\ell}  - \Delta\sum_{n=N+\ell + 1}^{\infty}
                        \frac{\bar{K}_n^+(N+\ell/2;g,\Delta,\ell/2)}{n-N-\ell} \bar{y}^n, \\
  \phi_{2,-}(\bar{y};-\ell/2)& = \sum_{n=N+\ell + 1}^{\infty} \bar{K}^+_n(N+\ell/2;g,\Delta,\ell/2) \bar{y}^n,
\end{align}
with initial conditions \(\bar{K}^+_{n}(N+\ell/2;g,\Delta,\ell/2) = 0 \, (n \leq N + \ell) \), \(\bar{K}^+_{N+\ell+1}(N+\ell/2;g,\Delta,\ell/2) = 1\)
and in both cases the coefficients satisfy
\begin{align*}
  \left(n - N + (2g)^2 + \frac{\Delta^2}{N+2\e-n}   \right)& \bar{K}^+_{n}(N+\e;g,\Delta,\e) - (2g)^2 \bar{K}^+_{n-1}(N+\e;g,\Delta,\e) \\
   &= (n+1) \bar{K}^+_{n+1}(N+\e;g,\Delta,\e) 
\end{align*}
Then ${}^t(\phi_{2,+}(\bar{y};-\e),\phi_{2,-}(\bar{y};-\e))$ is the local Frobenius solution of the system \eqref{eq:system2p}
at $\bar{y}=0$, where $\bar{y}=1-y$.

Note also that the radius of convergence of each series appearing above equals $1$. Moreover, the solutions can be expressed in
terms of the confluent Heun functions (see e.g. \cite{MPS2013,Le2016,ZXGBGL2014}).

A similar discussion to \cite{B2013AP} (see also \cite{B2011PRL-OnlineSupplement}) leads to the following set of equations to assure
the existence of the non-Juddian exceptional solutions. Actually, the eigenvalue equation for $\HRabi{\e}$, that is  \eqref{eq:system-1},
is equivalent via embedding to the system of differential equations given by
\begin{equation}\label{eq:embedding}
  \frac{d}{dz} \Psi(z)= A(z)\Psi(z),
\end{equation}
where
\begin{equation}
  A(z)=
  \begin{bmatrix}
    \frac{\lambda-\e-gz}{z+g} & 0 &   0 & \frac{-\Delta}{z+g} \\
    0 & \frac{\lambda+\e-gz}{z+g}  & \frac{-\Delta}{z+g} &0 \\
    0 &  \frac{-\Delta}{z-g} & \frac{\lambda-\e+gz}{z-g}  &0  \\
    \frac{-\Delta}{z-g}   & 0 & 0 & \frac{\lambda+\e+gz}{z-g}
  \end{bmatrix},
\end{equation}
for the vector valued function
\begin{align*}
  \Psi(z): ={}^t\Big(&e^{-gz}\phi_{1,+}\big(\frac{g+z}{2g};\e\big), e^{gz} \phi_{1,-}\big(\frac{g-z}{2g};\e\big),e^{gz} \phi_{1,+}\big(\frac{g-z}{2g};\e\big), e^{-gz} \phi_{1,-}\big(\frac{g+z}{2g};\e\big)\Big).
\end{align*}
  
It is not difficult to see that the function
\begin{align*}
  \Phi(z):={}^t\Big(&e^{gz} \phi_{2,-}\big(\frac{g-z}{2g};-\e\big), e^{-gz} \phi_{2,+}\big(\frac{g+z}{2g};-\e\big),e^{-gz}\phi_{2,-}\big(\frac{g+z}{2g};-\e\big), e^{gz} \phi_{2,+}\big(\frac{g-z}{2g};-\e\big)\Big)
\end{align*}
also satisfies \eqref{eq:embedding}. Hence, in order for a non-Juddian exceptional solution to exist it is necessary and sufficient that
for some $z_0\, (-g<z_0<g)$ (an ordinary point of the system), there exists a non-zero constant $c=c_N(g,\Delta, \e)$ and such that
\begin{equation}\label{eq:4equations}
  \left\{
    \begin{aligned}
      e^{-gz_0}\phi_{1,+}\big(\frac{g+z_0}{2g};\e\big) &= c\, e^{gz_0}\phi_{2,-}\big(\frac{g-z_0}{2g};-\e\big),\\
      e^{gz_0} \phi_{1,-}\big(\frac{g-z_0}{2g};\e\big) &= c\, e^{-gz_0}\phi_{2,+}\big(\frac{g+z_0}{2g};-\e\big),\\
      e^{gz_0} \phi_{1,+}\big(\frac{g-z_0}{2g};\e\big) &= c\, e^{-gz_0}\phi_{2,-}\big(\frac{g+z_0}{2g};-\e\big),\\
      e^{-gz_0} \phi_{1,-}\big(\frac{g+z_0}{2g};\e\big) &= c\, e^{gz_0}\phi_{2,+}\big(\frac{g-z_0}{2g};-\e\big).
    \end{aligned} \right.
\end{equation}
For $z_0=0$, it is obvious that the first and third, and the second and forth equations are equivalent respectively. Namely, the four
equations reduce to the following two equations when $y=\bar{y}=\frac12$.
\begin{equation}\label{eq:constraint}
  \left\{
  \begin{aligned}
    \phi_{1,-}(y;\e) & = c\, \phi_{2,+}(\bar{y};-\e) =  c\, \phi_{2,+}(1-y;-\e),\\
    \phi_{1,+}(y;\e) & = c\, \phi_{2,-}(\bar{y};-\e)=c\, \phi_{2,-}(1-y;-\e).
  \end{aligned} \right.
\end{equation}
for some non-zero constant \(c\) (as can be seen by applying the substitutions \(y \to \bar{y}=1-y\) and \(\e \to -\e \)
to the system \eqref{eq:systemN}). Therefore, by setting \(y = 1/2 \) (\(z = 0 \) in the original variable, an ordinary point of the system)
and eliminating the constant \(c\) in these linear relations gives the following  constraint \( T\)-function
\begin{equation} \label{eq:exceptGfuncDef}
  T_\e^{(N)}(g,\Delta) = \left(\bar{R}^{(N,+)} \cdot \bar{R}^{(N,-)} - R^{(N,+)} \cdot R^{(N,-)} \right)(g,\Delta;\e),
\end{equation}
where \(\cdot\) denotes the usual multiplication of functions and 
\begin{align}
  \label{eq:gfuncR_Nminus}
  \bar{R}^{(N,-)}(g,\Delta;\e) &= \phi_{1,+}\Bigl(\frac12;\e\Bigr),   &\bar{R}^{(N,+)}(g,\Delta;\e) &=\phi_{2,+}\Bigl(\frac12;-\e\Bigr),  \\
  R^{(N,-)}(g,\Delta;\e) &=  \phi_{1,-}\Bigl(\frac12;\e\Bigr),  &R^{(N,+)}(g,\Delta;\e)  &= \phi_{2,-}\Bigl(\frac12;-\e\Bigr).
\end{align}

Conversely, if there exists such $c=c_N(g,\Delta,\e)(\not=0)$, $\lambda=N+\e -g^2$ is a non-Juddian exceptional eigenvalue and the corresponding
functions \((\phi_{j,+}, \phi_{j,-}, \, j=1,2)\) satisfy \eqref{eq:constraint} and \eqref{eq:4equations}
(cf. \cite{Ince1956}).

\begin{rem}
  When $\e=0$ we observe that
  \[
    T_0^{(N)}(g,\Delta)= \left(\bar{R}^{(N,+)} -R^{(N,+)} \right) \cdot \left(\bar{R}^{(N,+)} +R^{(N,+)}\right)(g,\Delta,0)
  \]
  since $R^{(N,+)}(g,\Delta,0) = R^{(N,-)}(g,\Delta,0)$ and $\bar{R}^{(N,+)}(g,\Delta,0) = \bar{R}^{(N,-)}(g,\Delta,0))$.
\end{rem}

\begin{rem}
  By Corollary \ref{DegenerateStructure}, for any fixed $\Delta>0$, there are no common zeros between the constraint polynomial
  $P_N^{(N, \e)}((2g)^2, \Delta^2)$ and the $T$-function $T_{\e}^{(N)}(g,\Delta)$.
\end{rem}

In the same manner, we can define a  $T$-function \( \tilde{T}_\e^{(N)}(g,\Delta)\) that vanishes for values \(g,\Delta\) corresponding to
the non-Juddian exceptional eigenvalue \(\lambda = N - \e -g^2\). Clearly, we have $\tilde{T}_0^{(N)}(g,\Delta) = T_0^{(N)}(g,\Delta)$, and in general
it is straightforward to verify that the identity
\begin{equation} \label{eq:GNepsAndGNmeps}
  \tilde{T}_{\e}^{(N)}(g,\Delta) = T_{-\e}^{(N)}(g,\Delta)
\end{equation}
holds (up to a constant) as in the case of  $\tilde{P}_N^{(N,\e)}((2g)^2, \Delta^2)$ (see \cite{W2016JPA} and also \cite{LB2015JPA}).

We consider the particular case of $\e=\ell/2\, (\ell\in \Z_{\geq 0})$. Then, from \eqref{eq:constraint} we have
$\phi_{1,-}(y;\ell/2) =  c \phi_{2,+}(1-y;-\ell/2)$ and $\phi_{1,+}(y;\ell/2) =c \phi_{2,-}(1-y;-\ell/2)$. This shows that the non-Juddian
exceptional solution corresponding to $\lambda=(N+\ell)-\ell/2-g^2=N +\ell/2-g^2$ whose existence is guaranteed by the constraint
equation \(T_{\ell/2}^{(N)}(g,\Delta) = 0\) (resp. \(\tilde{T}_{\ell/2}^{(N+\ell)}(g,\Delta) = 0\)) are identical up to a scalar multiple.
Since the non-Juddian exceptional solution is non-degenerate, the compatibility of this fact, that is, that $T_{\ell/2}^{(N)}(g,\Delta)$
and $\tilde{T}_{\ell/2}^{(N+\ell)}(g,\Delta)$ have the same zero with respect to $g$ for a fixed $\Delta$, is confirmed by the lemma below.

\begin{lem}\label{lem:const-T-rel}
  For $\ell,N \in \Z_{\geq 0}$ we have
  \begin{align}\label{eq:const-T-rel}
    \tilde{T}_{\ell/2}^{(N+\ell)}(g,\Delta) =T_{\ell/2}^{(N)}(g,\Delta).
  \end{align}
\end{lem}
\begin{proof}
  From the definitions, we have \(\bar{K}_n^{\pm}(N-\ell/2 ;g,\Delta,-\ell/2) =  \bar{K}_n^{\mp}(N+\ell/2;g,\Delta,\ell/2)\), therefore
  \begin{align*}
    \bar{R}^{(N+\ell,\pm)}(g,\Delta,-\ell/2) &= \bar{R}^{(N,\mp)}(g,\Delta,\ell/2), \\
    R^{(N+\ell,\pm)}(g,\Delta,-\ell/2) &= R^{(N,\mp)}(g,\Delta,\ell/2).
  \end{align*}
  Hence, it follows that
  \begin{align*}
    \tilde{T}_{\ell/2}^{(N+\ell)}(g,\Delta) &= T_{-\ell/2}^{(N+\ell)}(g,\Delta) \\
    &= \left(\bar{R}^{(N+\ell,+)} \cdot \bar{R}^{(N+\ell,-)}- R^{(N+\ell,+)} \cdot R^{(N+\ell,-)} \right) (g,\Delta;-\ell/2) \\
    &= \left(\bar{R}^{(N,-)} \cdot \bar{R}^{(N,+)} - R^{(N,-)} \cdot R^{(N,+)}\right)(g,\Delta;\ell/2) \\
    &= T_{\ell/2}^{(N)}(g,\Delta). 
  \end{align*} 
\end{proof}

By the discussion above, the condition \(T_\e^{(N)}(g,\Delta) = 0\) (resp. \(\tilde{T}_\e^{(N)}(g,\Delta) = 0\)) can be indeed be regarded as
the constraint equation for the exceptional eigenvalues \(\lambda = N + \e - g^2 \) (resp. \(\lambda = N - \e - g^2\)) with non-Juddian
exceptional solutions.

We illustrate numerically the constraint relations \(\cp{N,\e}{N}((2g)^2,\Delta^2) = 0 \) (for Juddian eigenvalues) and  \(T_\e^{(N)}(g,\Delta)= 0\)
(for non-Juddian exceptional eigenvalues) in Figure \ref{fig:constraintgraph3} showing the curves in the \((g,\Delta)\)-plane defined by these
relations for \(\e = 0.45\) and \(N = 3\). Concretely, Figures \ref{fig:constraintgraph3}(a) and \ref{fig:constraintgraph3}(b) depict the
graph of the curve \(G_\e(x,g,\Delta)=0 \) for the values \(x = 3.2 \) and \(x = 3.4 \), while Figure \ref{fig:constraintgraph3}(c) shows the
graph of the curve \(T_\e^{(3)}(g,\Delta)= 0\) in continuous line and \(\cp{3,\e}{3}((2g)^2,\Delta^2) = 0 \) in dashed line. Notice that as \(x \to 3.45\)
adjacent closed curves near the origin in the graph of \(G_\e(x,g,\Delta)=0 \) approach each other. Some of these curves join to form the closed
curves (ovals) of \(\cp{N,\e}{N}((2g)^2,\Delta^2) = 0 \), corresponding to Juddian eigenvalues, while others form curves in the graph of
\(T_\e^{(N)}(g,\Delta)= 0\), corresponding to non-Juddian exceptional eigenvalues. Also observe that we have ovals (corresponding to non-Juddian
solutions) near the origin of the graph in Figure \ref{fig:constraintgraph3}(c), some of them very close to dashed ovals (corresponding to
Juddian eigenvalues).

\begin{figure}[htb]
  ~
  \begin{subfigure}[b]{0.3\textwidth}
    \centering
    \includegraphics[height=4.5cm]{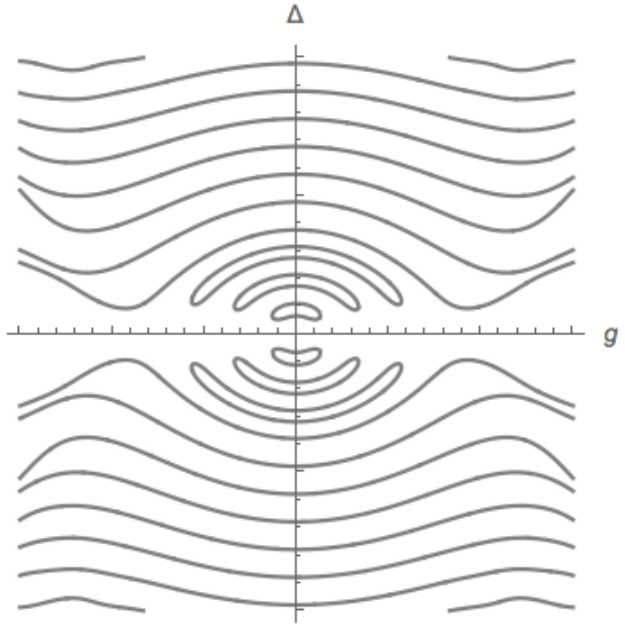}
    \caption{\(\e = 0.45,\, x= 3.2\)}
  \end{subfigure}
  ~
  \begin{subfigure}[b]{0.3\textwidth}
    \centering
    \includegraphics[height=4.5cm]{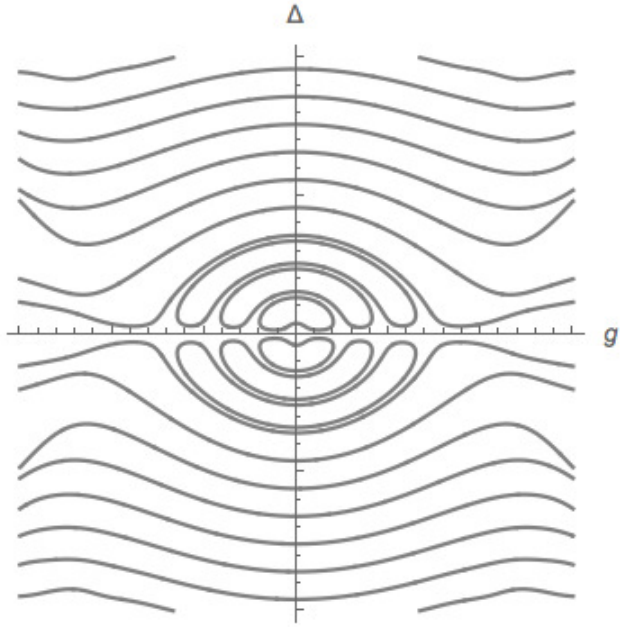}
    \caption{\(\e = 0.45,\, x = 3.4 \)}
  \end{subfigure}
  ~
  \begin{subfigure}[b]{0.32\textwidth}
    \centering
    \includegraphics[height=4.5cm]{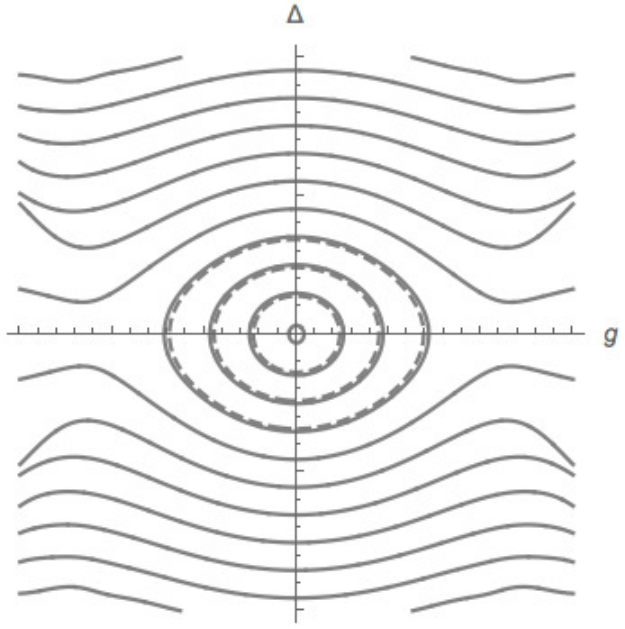}
    \caption{\(\e = 0.45, N = 3, x =N + \e\)}
  \end{subfigure}
  \caption{Curves of constraint relations for fixed regular eigenvalues ((a),(b)) and exceptional eigenvalues ((c)) for \( \e = 0.45\) for \(-3 \le g \le 3 \) and \(-10 \leq \Delta \le 10. \)}
  \label{fig:constraintgraph3}
\end{figure}

On the other hand, the case \(\e \in \frac12\Z_{\geq 0}\) is illustrated in Figure \ref{fig:constraintgraph2}.
As in the case above, Figures \ref{fig:constraintgraph2}(a) and \ref{fig:constraintgraph2}(b) depict the curves given by the
relation \(G_\e(x,g,\Delta)=0 \) for the values \(x = 3.2 \) and \(x = 3.4 \), while Figure \ref{fig:constraintgraph2}(c) shows the graph
of the curve \(T_\e^{(3)}(g,\Delta)= 0\) (\(N = 3\)) in continuous line and \(\cp{3,\e}{3}((2g)^2,\Delta^2) = 0 \) in dashed line.
Different from the case \(\e \not\in \frac12 \Z_{\geq 0}\) above, there are no continuous ovals (non-Juddian) near the origin in
Figure \ref{fig:constraintgraph2}(c). It is worth mentioning that Figure \ref{fig:constraintgraph3}(c) and Figure \ref{fig:constraintgraph2}(c)
support the conceptual graphs of Figure \ref{fig:excepteigen2}(a) and Figure \ref{fig:excepteigen2}(b) presented in the Introduction. Actually,
we can observe there are both dashed (Juddian) and continuous (non-Juddian) ovals when $\e=0.45$ in Figure \ref{fig:constraintgraph3}(c),
while the continuous ovals disappear when $\e=\frac12(\in \frac12\Z)$ in Figure \ref{fig:constraintgraph2}(c) (see Corollary \ref{NoNonJudd}
and its subsequent Remark \ref{no-non-Juddian}).

\begin{figure}[htb]
  ~
  \begin{subfigure}[b]{0.3\textwidth}
    \centering
    \includegraphics[height=4.5cm]{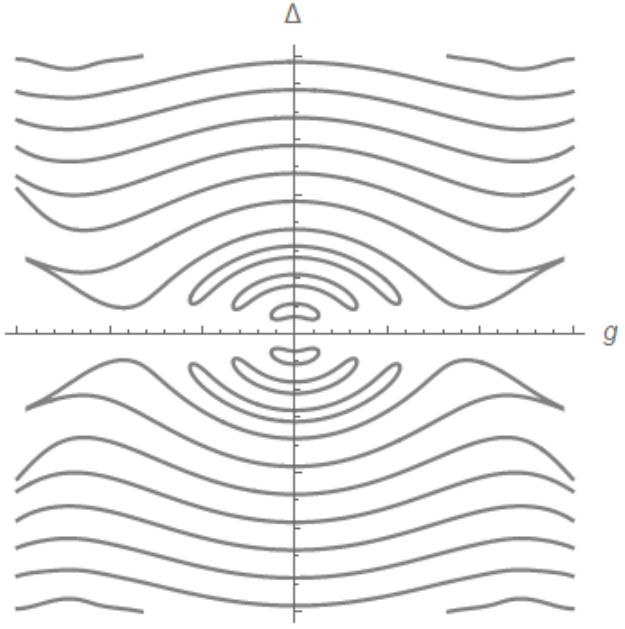}
    \caption{\(\e = \frac12, x= 3.2\)}
  \end{subfigure}
  ~
  \begin{subfigure}[b]{0.3\textwidth}
    \centering
    \includegraphics[height=4.5cm]{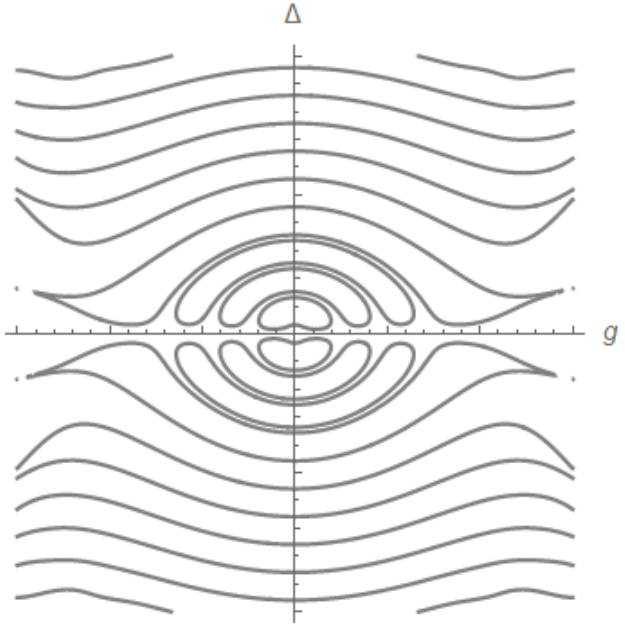}
    \caption{\(\e = \frac12,\, x = 3.4 \)}
  \end{subfigure}
  ~
  \begin{subfigure}[b]{0.32\textwidth}
    \centering
    \includegraphics[height=4.5cm]{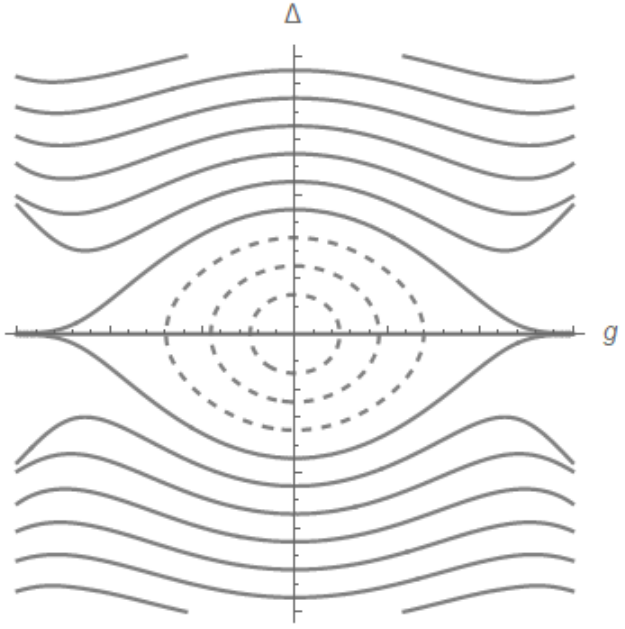}
    \caption{\(\e = \frac12, N = 3, x = N + \e\)}
  \end{subfigure}
  \caption{Curves of constraint relations for fixed regular eigenvalues ((a),(b)) and exceptional eigenvalues ((c)) for \( \e = \frac12\) for \(-3 \le g \le 3 \) and \(-10 \leq \Delta \le 10. \)}
  \label{fig:constraintgraph2}
\end{figure}

We next generalize Lemma \ref{G-function}, by including the exceptional eigenvalues. As a result, we see that the spectrum
of \(\HRabi{\e} \) does not depend on the sign of \(\e\).

\begin{prop}\label{prop:symmetry_spec}
  The spectrum of the Hamiltonian $\HRabi{\e}$ of AQRM depends only on $|\e|$. In other words, the spectrum of
  Hamiltonian $\HRabi{-\e}$ coincides with that of  $\HRabi{\e}$.
\end{prop}

\begin{proof}
  For the regular spectrum, since $G_{-\e}(x,g, \Delta)=G_\e(x,g, \Delta)$ by Lemma \ref{G-function} the result follows immediately.
  Moreover, since the constraint polynomials $P_N^{(N,\pm\e)}(x,y)$ of $\HRabi{\e}$ are also constraint polynomials $P_N^{(N, \mp(-\e))}(x,y)$ of
  $\HRabi{-\e}$, the result holds for Juddian eigenvalues as well.
  Finally, if $g$ is a positive zero of $T_\e^{(N)}(g, \Delta)$, that is, $\lambda=N+\e-g^2$ is a non-Juddian exceptional eigenvalue of $\HRabi{\e}$,
  then, $\lambda = N+\e-g^2 = N - (-\e) - g^2$ is also a non-Juddian exceptional eigenvalue of
  $\HRabi{-\e}$ since $g$ is actually a zero of $\tilde{T}^{(N)}_{-\e}(g,\Delta)(= T_\e^{(N)}(g, \Delta))$. Hence the assertion follows. 
\end{proof}

\,

\begin{rem}
The graphs in Figures \ref{fig:constraintgraph3} and \ref{fig:constraintgraph2} might explain the conical structure discussed in \cite{BLZ2015JPALetter}. It is therefore important to understand why ovals corresponding to the non-Juddian exceptional eigenvalues appear in Figure \ref{fig:constraintgraph3}(c).
\end{rem}
\,

\begin{rem}\label{rem:symmetry_spec}
  Proposition \ref{prop:symmetry_spec}  follows also from the comparison of the two systems of differential equations \eqref{eq:system1}
  and \eqref{eq:system2p}.
  It should be also noted that the proposition does not imply that  $\bar{\lambda} = N\mp\e -g^2$ is an eigenvalue of \(\HRabi{-\e} \) even if $\lambda = N\pm\e -g^2$ is  an eigenvalue of $\HRabi{\e}$.
\end{rem}

\begin{rem}
  For a fixed \(\Delta >0 \), define the involution \(\sigma : \R^2  \to \R^2 \)
  \begin{equation*}
    \sigma : (y, \e) \to (y+2\e,-\e).
  \end{equation*}
  Associating a tuple \((y,\e) \in \R^2\) to each eigenvalue \( \lambda = y + \e -g^2  \) of $\HRabi{\e}$, we easily see that
  the eigenvalues of \(\HRabi{\e} \) are invariant under \(\sigma\). For instance, if the eigenvalue \(\lambda\) is a regular, we have
  \(\lambda = x - g^2 \) for some \(x \in \R \, (x \not= N \pm  \e) \), thus \(y = x - \e \) and the image  \(\sigma(x-\e,\e) =  (x+\e,-\e)\)
  corresponds to same eigenvalue \( \lambda = (x+\e)-\e -g^2\) under this interpretation. The case of exceptional eigenvalues follows in a
  similar manner. It is an interesting problem to relate the involution \(\sigma\) with the identities \eqref{eq:conj},\eqref{eq:gfunctrel} and
  \eqref{eq:const-T-rel} when $\e$ is a half integer. Actually, it is widely believed among physicists (cf. \cite{AW_personal2017})
  that there must be a symmetry if there exist energy level crossings (i.e. spectral degeneration) like we have Juddian eigenvalues for
  a half-integral $\e$. See also e.g. \cite{GD2013JPA} for further discussion on the symmetry for the \(\e=0\) case.
\end{rem}

\subsection{Residues of the \texorpdfstring{$G$}{G}-function and spectral determinants of \texorpdfstring{$\HRabi{\e}$}{H}} \label{sec:SpectralDet}

In this subsection, we return to the discussion of the structure of poles of the $G$-function
started in \S \ref{sec:Gfunct}. This enables us to deepen the understanding of non-Juddian exceptional eigenvalues.
Also, we establish the relation between the $G$-function and the spectral determinant (i.e. the zeta regularized product  of the spectrum \cite{V1987CMP,QHS1993TAMS}) of the Hamiltonian $\HRabi{\e}$. This can be regarded as a mathematical refinement of the discussion partially made in \cite{LB2016JPA}.

Formally, to study the behavior of the $G$-function $G_\e(x;g,\Delta)$ at a point \(x = N \pm \e \, (N \in \Z_{\ge 0})\) we consider a sufficiently small
punctured disc centered at a fixed point \(x = N \pm \e \) and compute the residue of $G_\e(x;g,\Delta)$ as a function of the parameters \(g\) and
\(\Delta\). According to the value of the residue for the parameters \(g\) and  \(\Delta\) we classify the singularity as a removable singularity or a
pole. In the case of a removable singularity we consider the $G$-function $G_\e(x;g,\Delta)$ as a function defined at \(x = N \pm \e \) for the
particular parameters \(g\) and \(\Delta\). It is clear from the definition that the only singularities of $G_\e(x;g,\Delta)$ (as a function of \(x\))
appear at the points \(x = N \pm \e  \, (N \in \Z_{\geq 0})\) and that all singularities are either removable singularities or poles.
To simplify the notation, we say that a function has a pole of order \(\leq N\) when it has a removable singularity or a pole of order at most \(N\).

We consider the case \(\e \not\in \frac12 \Z \) and \(\e \in \frac12 \Z \) by separate. For the case of \(\e \not\in \frac12 \Z \), by the defining
recurrence formula \eqref{eq:RecursionKn}, we observe that the rational functions $K^{\mp}_n(x),$ for $ n\geq N+1$, have poles of order \(\leq 1\) at
$x=N\pm \e$. Hence, $G_\e(x;g,\Delta)$ has a pole of order \(\leq 1\) at $x=N\pm \e$. The residue of the $G$-function at a point \(x = N \pm \e \) is
given in the following result.

\begin{prop}\label{prop:polenhi1}
  Let \(\e \not\in \frac12 \Z \). Then any pole of the $G$-function $G_\e(x;g,\Delta)$ is simple.
  If \(N \in \Z_{\geq 0}\), the residue of $G_\e(x;g,\Delta)$ at the points \(x = N \pm \e \) is given by
  \begin{align*}
    \Res_{x = N \pm \e} G_\e(x;g,\Delta) &= C(N) \Delta^2 \cp{N,\pm \e}{N}((2g)^2,\Delta^2) T^{(N)}_{\pm \e} (g,\Delta),
  \end{align*}
  where \(C(N) = \frac{1}{N! (N+1)!} \).
\end{prop}

\begin{proof}
  We give the proof for the case \(x= N + \e\), for the case \(x = N-\e \) is completely analogous.
  From the definition of \(f_n^-(x,g,\Delta,\e)\), it is clear that
  \(  \Res_{x = N +\e} f_n^-(x,g,\Delta,\e) = \frac{\Delta^2}{2g} \delta_N (n) \), where \(\delta_x(y) \) is the Kronecker delta function.
  Likewise,  for \( 0 \leq n \leq N\) it is clear that \( \Res_{x = N +\e} K_n^-(x;g,\Delta,\e) = 0 \), and for \(n = N+1\) we have
  \begin{align*}
    \Res_{x= N+\e} K_{N+1}^-(x) &=  \lim_{x \to N+\e} (x - N -\e) \frac{1}{N+1} \left( f_N^{-}(x)K_N^-(x) - K^{-}_{N-1}(x) \right) \\
                                & = \frac{1}{N+1} K_N^{-}(N+\e) \Res_{x\to N+\e}f_N^-(x)  = \frac{\Delta^2}{2g(N+1)} K_N^{-}(N+\e).
  \end{align*}
  Setting $a_0=0$, $a_1 = 1$, and
  \[
     a_k = \frac{1}{N+k}\left( f_{N+k-1}^-(N+\e) a_{k-1} - a_{k-2} \right),
   \]
   for \(k \geq 2\), it is easy to see that \(\Res_{x= N+ \e} K_{N+k}^-(x) = \frac{\Delta^2}{2 g (N+1)} K_N^-(N+\e) a_k\).
   Furthermore, by the same method of the proof of Proposition \ref{Prop:reccurEquiv}, we observe that
   \[
     (2g)^{k-1} a_k  = \bar{K}_{N+k}^{-}(N+\e;g,\Delta,\e),
   \]
   for \(k \geq 1\), where \(\bar{K}_{N+k}^{-}(N+\e;g,\Delta,\e)\) are the coefficients of \(\phi_{1,-}(y;\epsilon)\) in \eqref{eq:solphi1-}.
   Then, from the definition, \(R^{(N,-)}(g,\Delta;\epsilon)\) is given by
   \begin{align*}
     \phi_{1,-}\Bigl(\frac12;\e\Bigr) &= \sum_{n = N+1}^{\infty} \bar{K}_n^{-}(N+\e;g,\Delta,\e) \Bigl(\frac12\Bigr)^n \\
     &= \Bigl(\frac12\Bigr)^{N+1} \sum_{n = N+1}^{\infty} \bar{K}_n^{-}(N+\e;g,\Delta,\e) \Bigl(\frac12\Bigr)^{n-N-1} \\
     &= \Bigl(\frac12\Bigr)^{N+1} \sum_{n = N+1}^{\infty} a_{n-N} g^{n-N-1},
   \end{align*}
   and, similarly, \(\bar{R}^{(N,-)}(g,\Delta;\e) \) is given by
   \begin{align*}
     \phi_{1,+}\Bigl(\frac12;\e\Bigr) &= \frac{(N+1)}{\Delta}\Bigl(\frac12\Bigr)^N - \Delta \sum_{n = N+1}^{\infty} \frac{\bar{K}_n^{-}(N+\e;g,\Delta,\e)}{n-N} \Bigl(\frac12\Bigr)^n \\
     &= \Bigl(\frac12\Bigr)^{N+1} \Bigl( \frac{2 (N+1)}{\Delta} - \Delta \sum_{n = N+1}^{\infty} \frac{\bar{K}_n^{-}(N+\e;g,\Delta,\e)}{n-N} \Bigl(\frac12\Bigr)^{n-N-1} \Bigr) \\
     &= \Bigl(\frac12\Bigr)^{N+1} \Bigl( \frac{2 (N+1)}{\Delta} - \Delta \sum_{n = N+1}^{\infty} \frac{a_{n-N}}{n-N} g^{n-N-1} \Bigr).
   \end{align*}
   Moreover, we recall that for \(\e \notin \frac12 \Z \), both functions \(R^+(x)\) and \(\bar{R}^{+}(x)\) are analytic at \(x = N +\e \) and
   \[
     R^+(N+\e) = R^{(N,+)}(g,\Delta;\e), \qquad \Delta \bar{R}^{+}(N+\e)  = \bar{R}^{(N,+)}(g,\Delta;\e).
   \]
   With these preparations, we compute \(\Res_{x=N+\e} R^+(x) R^-(x) \) as 
   \begin{align*}
     R^{+}(N+\e) &\Res_{x=N+\e} \sum_{n=0}^{\infty} K_N^-(x)  g^n \\
      &=  \frac{\Delta^2}{2 g (N+1)}K_N^-(N+\e) R^{+}(N+\e)  \sum_{n=N+1}^{\infty} a_{n-N}  g^n \\
                                 &= \frac{(2g)^{N} \Delta^2}{(N+1)}   K_N^-(N+\e) R^{(N,+)}(g,\Delta;\e) R^{(N,-)}(g,\Delta;\e).
   \end{align*}
   and, since \(\Res_{x=N+\e} \frac{K_N^{-}(x)}{x - N - \e} =  K_N^-(N+\e)\) holds trivially, we also obtain
   \begin{align*}
     &\Res_{x=N+\e} \Delta^2 \bar{R}^+(x) \bar{R}^-(x) = \Delta^2 \bar{R}^{+}(N+\e) \Res_{x=N+\e} \sum_{n=0}^{\infty} \frac{K_N^-(x)}{x-n-\e}  g^n \\
     &\,= \Delta^2 \bar{R}^{+}(N+\e) \Bigl( K_N^{-}(N+\e) g^N - \frac{\Delta^2}{2 g (N+1)}K_N^-(N+\e) \sum_{n=N+1}^{\infty} \frac{a_{n-N}}{n-N}  g^n  \Bigr) \\
     &\,= \frac{g^N \Delta^2}{2(N+1)} K_N^{-}(N+\e) (\Delta \bar{R}^{+}(N+\e))  \Bigl( \frac{2 (N+1)}{\Delta} -  \Delta \sum_{n=N+1}^{\infty} \frac{a_{n-N}}{n-N}  g^{n-N-1}  \Bigr),\\
     &\,= \frac{(2g)^N \Delta^2}{(N+1)} K_N^-(N+\e) \bar{R}^{(N,+)}(g,\Delta;\e) \bar{R}^{(N,-)}(g,\Delta,\e).
   \end{align*}
   Finally, using Lemma \ref{lem:GcoeffcPoly} we have
   \begin{align*}
     \Res_{x=N+\e} G_\e(x;g,\Delta)
     &= \frac{(2g)^N \Delta^2}{(N+1)} K_N^-(N+\e) T^{(N)}_\e(g,\Delta) \\
     &= \frac{\Delta^2}{ N! (N+1)!} \cp{N,\e}{N}((2g)^2,\Delta^2) T^{(N)}_\e(g,\Delta),
   \end{align*}
   which is the desired result. 
\end{proof}

\begin{rem}
We make a remark on the relation \((2g)^{k-1} a_k  = \bar{K}_{N+k}^{-}(N+\e;g,\Delta,\epsilon) \) appearing in the proof of the
proposition above. The coefficients \(K_n^{-}(x;g,\Delta,\e) \) (and thus the numbers \(a_k\)) in the definition of the
$G$-function \(G_\e(x;g,\Delta) \) arise from the solution of the system of differential equations \eqref{eq:system-1} by using the change of
variable \(y = g + z \) (instead of \(y = \frac{g+z}{2 g} \)). The use of the change of variable \(y = \frac{g+z}{2 g} \) results on the
system \eqref{eq:system1} compatible with the representation theoretical description of Proposition~\ref{prop:RedEigenProblem}, and
therefore we use the solutions arising from this system for the definition of the $T$-function \(T^{(N)}_\e(g,\Delta)\) (see \S~\ref{sec:LargestExponent}).
We also note that it is possible to equivalently redefine the $G$-function \(G_\e(x;g,\Delta) \) using the solutions of the system \eqref{eq:system1}
(i.e. with the change of variable \(y = \frac{g+z}{2 g} \)), however, we use the definition given in \S\ref{sec:Gfunct} since it is standard
in the literature, including e.g., \cite{B2011PRL,LB2015JPA,Le2016}.
\end{rem}

The proposition above completely characterizes the poles of the $G$-function for the case \(\e \not\in \frac12 \Z \)
in terms of the exceptional spectrum of \(\HRabi{\e}\). In particular, it shows that the function \(G_\e(x;g,\Delta) \) is finite at
points \(x = N \pm \e \, (N \in \Z_{\geq 0})\) corresponding to non-Juddian exceptional eigenvalues \(\lambda = N \pm \e - g^2\) (i.e. the parameters
\(g\) and \(\Delta\) are positive zeros of \(T_\e^{(N)}(g,\Delta)\)). This situation is illustrated in Figure \ref{fig:gfunct_graph3}(a) for the
parameters \(\e = 0.3, g \approx 0.8695, \Delta = 1/2\), showing the finite value of \(G_{0.3}(x;g,1/2) \) at \(x = 1.3 \). Here, \(g \approx 0.8695\) is a
root (computed numerically) of \(T_{0.3}^{(1)}(g,1/2)\). By Corollary \ref{DegenerateStructure}, this value of \(g\) must be different to the
value \(g' \approx 0.5809 \) in the Juddian case, shown in Figure \ref{fig:gfunct_graph}(a), which also has a finite value of \(G_{0.3}(x;g',1/2)\)
at \(x = 1.3 \).

The following corollary justifies the claim that the exceptional eigenvalues \(\lambda = N \pm \e - g^2\) vanish (or ``kill'') the poles
of the $G$-function.

\begin{cor}
  Suppose \(\e \not\in \frac12 \Z \), \(N \in \Z_{\geq 0} \) and \(\Delta>0 \). Then, \(\HRabi{\e}\) has the exceptional
  eigenvalue \( \lambda = N \pm \e - g^2\) if and only if the $G$-function \(G_{\e}(x,g,\Delta) \) does not have a pole at \(x = N \pm \e\). \QEDhere
\end{cor}

\begin{figure}[htb]
  ~
  \begin{subfigure}[b]{0.45\textwidth}
    \centering
    \includegraphics[height=3.25cm]{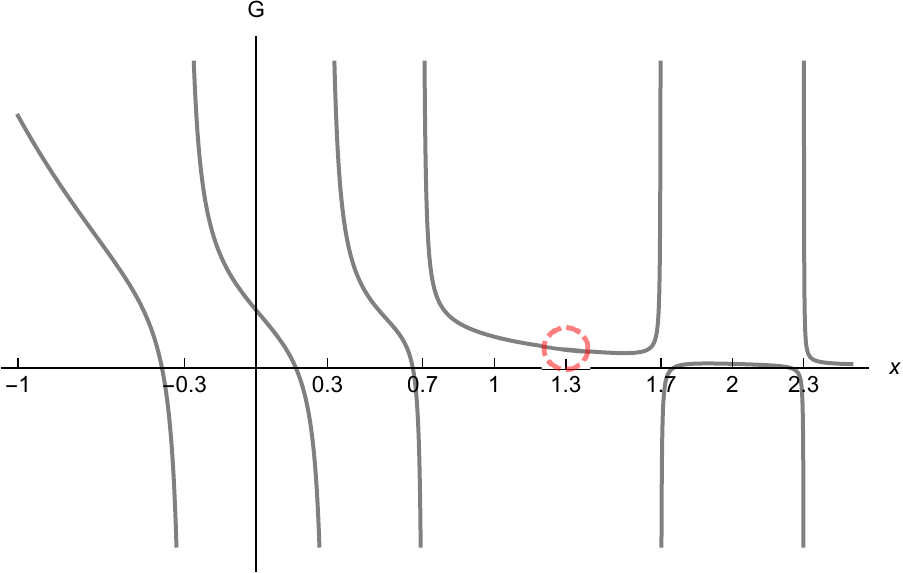}
    \caption{\(\e = 0.3,g \approx 0.8695,\Delta=0.5\)}
  \end{subfigure}
  ~
  \begin{subfigure}[b]{0.45\textwidth}
    \centering
    \includegraphics[height=3.25cm]{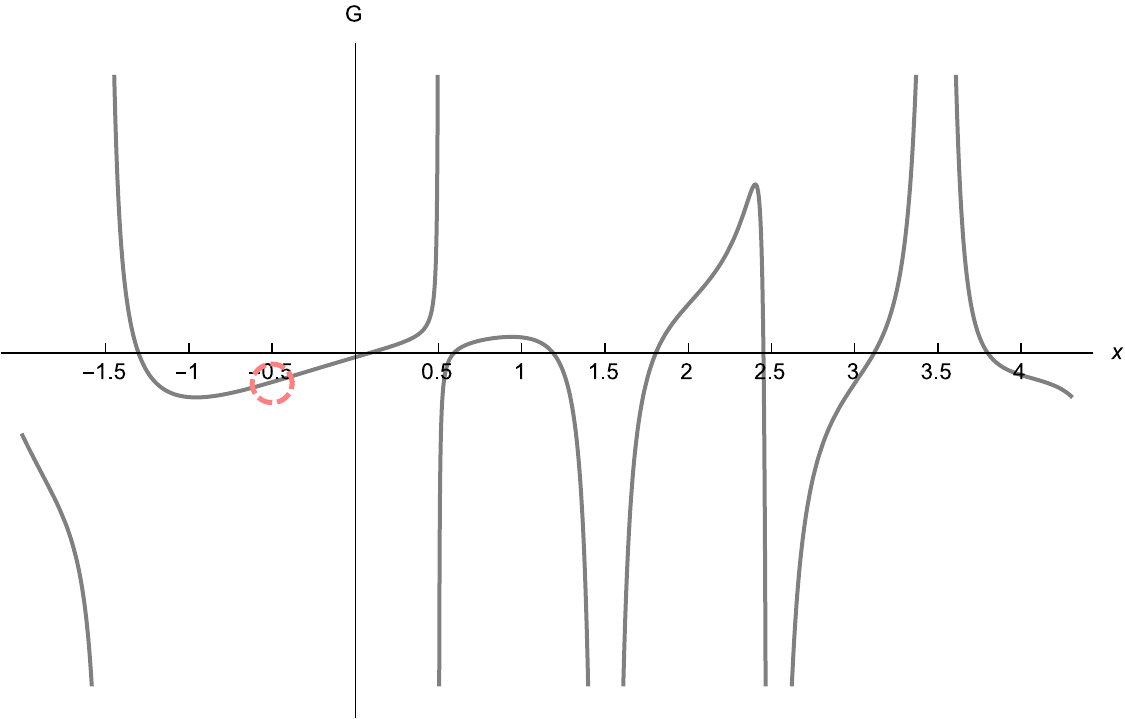}
    \caption{\(\e = 1.5, g \approx 1.6318, \Delta=3\)}
  \end{subfigure}
  \\
  \centering
  \begin{subfigure}[b]{0.45\textwidth}
    \centering
    \includegraphics[height=3.25cm]{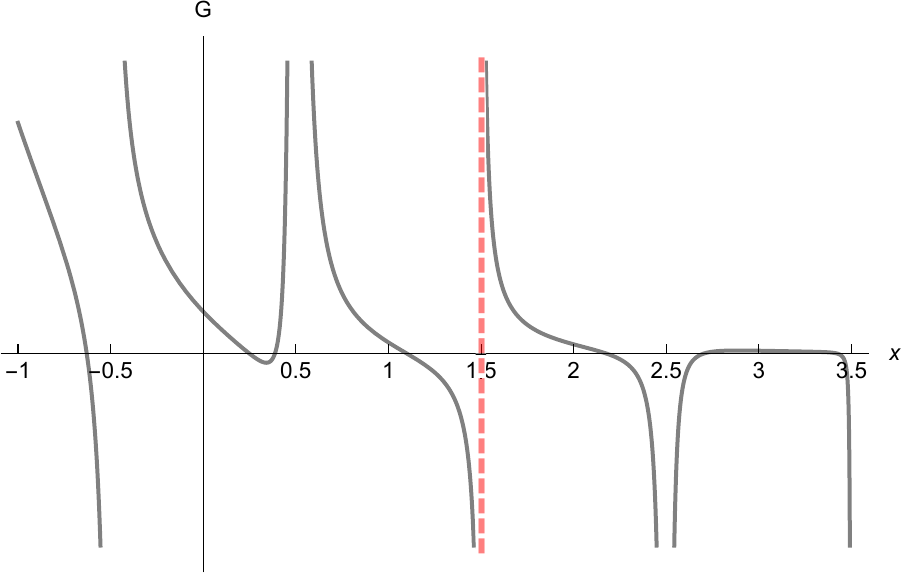}
    \caption{\(\e = 0.5, g \approx 1.3903, \Delta=1\)}
  \end{subfigure}
  \caption{Plot of \(G_\e(x;g,\Delta)\) for parameters \(g \) and \(\Delta\) corresponding to a non-Juddian eigenvalue \(\lambda = 1 \pm \e - g^2\) (i.e. $T^{(1)}_{\pm \e}(g, \Delta)=0$). A finite value of \(G_\e(x;g,\Delta)\) at \(x = 1 \pm \e \) is indicated by a dashed circle, while a simple pole at \(x = 1 + \e \)  is indicated with a dashed vertical line.}
  \label{fig:gfunct_graph3}
\end{figure}

  Next, we consider the case \(\e = \ell/2 \, (\ell \in \Z_{\geq 0})\). On the one hand, the functions \( R^{+}(x;g,\Delta,\e) \) and
  \( R^{-}(x;g,\Delta,\e) \) (resp. \( \bar{R}^{+}(x;g,\Delta,\e) \) and \( \bar{R}^{-}(x;g,\Delta,\e) \)) have poles of order \(\leq 1\) at points
  \(x = N + \ell/2 \, ( N \in \Z_{\geq 0})\). On the other hand, at points \( x = N - \ell/2 \) with \( 0 \leq N \leq \ell-1 \) only the
  functions \( \bar{R}^{+}(x;g,\Delta,\e) \) and \(\bar{R}^{+}(x;g,\Delta,\e) \) have poles of order \(\leq 1\). Consequently, the $G$-function \(G_{\ell/2}(x;g,\Delta) \)
  has poles of order \(\leq 1\) at the points \( x = N - \ell/2 \, (0 \leq N \leq \ell-1)\) and poles of order \(\leq 2\) at points \(x = N + \ell/2 \, (N \in \Z_{\geq 0})\).
  Note that all possible poles of the $G$-function are accounted since \(N = \ell + i \, (i \in \Z_{\geq 0}) \) yields
  \(x = N - \ell/2 = \ell + i - \ell/2  = i + \ell/2 \, (i \in Z_{\geq 0} )\). The residue at the poles of order \(\leq 1\) are given in the following
  proposition. The proof is identical to Proposition \ref{prop:polenhi1} and is therefore omitted.

\begin{prop}\label{prop:polehi1simple}
  Suppose \(\ell > 1 \) and let \(1 \leq N < \ell\). Then any pole of the $G$-function $G_{\ell/2}(x;g,\Delta)$ at a point \(x = N - \ell/2\) is simple.
  The residues of $G_{\ell/2}(x;g,\Delta)$ at the point \(x = N - \ell/2\) is given by
  \[
    \Res_{x = N - \ell/2} G_{\ell/2}(x;g,\Delta) = C(N) \Delta^2 \cp{N,-\ell/2}{N}((2g)^2,\Delta^2) \tilde{T}^{(N)}_{\ell/2} (g,\Delta)
  \]
  with \(C(N) = \frac{1}{N! (N+1)!} \).
\end{prop}

Similar to the non half-integer case, the residues of \(G_{\ell/2}(x;g,\Delta)\) at the points \(x = N - \ell/2  \) with \(1 \leq N < \ell\)  depend on the
constraint polynomial \(\cp{N,-\ell/2}{N}((2g)^2,\Delta^2)\) and $T$-function \(\tilde{T}^{(N)}_{\ell/2} (g,\Delta)\) for \(1 \leq N < \ell\). However, by
Proposition \ref{prop:pos2} \(\cp{N,-\ell/2}{N}((2g)^2,\Delta^2)\) is positive for \(g,\Delta>0\), in other words, the pole vanishes (i.e. it is a removable
singularity) if and only if \(\tilde{T}^{(N)}_{\ell/2} (g,\Delta)= 0\), as illustrated in Figure~\ref{fig:gfunct_graph3}(b).

In the following proposition we consider the remaining poles of \(G_{\ell/2}(x;g,\Delta)\).

\begin{prop} \label{prop:polehi}
  Suppose \(\e = \ell/2\, (\ell \in \Z_{\geq 0})\) and let \(N \in \Z_{\geq 0}\). Let
  \[
    G_{\ell/2}(x;g,\Delta) = \frac{A}{(x - N -\ell/2)^2} + \frac{B}{x - N -\ell/2} + H_{\ell/2}(x;\Delta,g)
  \]
  for a function \( H_{\ell/2}(x;\Delta,g) \) analytic at \(x = N + \ell/2 \).
  We have
  \[
    A = C(N)C(N+\ell) \Delta^4 \cp{N,\ell/2}{N}((2g)^2,\Delta^2) \cp{N+\ell,-\ell/2}{N+\ell}((2g)^2,\Delta^2) T^{(N)}_{\ell/2}(g,\Delta)^2,
  \]
  where \(C(N)\) is defined as in Proposition \ref{prop:polenhi1}, and

\begin{align*}
   B & =  \Res_{x = N + \ell/2}  G_{\ell/2}(x;g,\Delta) =  C(N) C(N+\ell) \Delta^2 \cp{N,\ell/2}{N}((2g)^2,\Delta^2)\\
  \times & \Big[ \frac{1}{C(N+\ell)}\Big( \bar{R}^{(N,-)}(g,\Delta,\frac{\ell}{2}) (\Delta \bar{Q}^+(N+\frac{\ell}{2};g,\Delta)) \\
  & \qquad \qquad \qquad \qquad - R^{(N,-)}(g,\Delta,\frac{\ell}{2})  Q^+(N+\frac{\ell}{2};g,\Delta) \Big) \\
     & +   A_N^{\ell}((2g)^2,\Delta^2) \frac{1}{C(N)} \Big( \bar{R}^{(N,+)}(g,\Delta,\frac{\ell}{2}) (\Delta\bar{Q}^-(N+\frac{\ell}{2};g,\Delta)) \\
     &\qquad \qquad \qquad \qquad \qquad \qquad -  R^{(N,+)}(g,\Delta,\frac{\ell}{2})  Q^-(N+\frac{\ell}{2};g,\Delta) \Big) \Big],
  \end{align*}
  where \(Q^{+}(x;g,\Delta) \) is defined by \(Q^{+}(x;g,\Delta) = R^{+}(x;g,\Delta) - \frac{\Res_{x=N+\ell/2} R^{+}(x;g,\Delta)}{x- N - \ell/2}\). The functions
  \(\bar{Q}^{+}(x;g,\Delta)\), \(Q^{-}(x;g,\Delta)\) and \(\bar{Q}^{-}(x;g,\Delta)  \) are defined similarly.
\end{prop}

\begin{proof}
  To compute the term \(A\) we notice that since each of the factors \(R^+(x)\) and \(R^-(x) \) (resp. \(\bar{R}^+(x)\) and \(\bar{R}^-(x) \))
  can have a pole of order exactly one (simple pole) we have
  \[
     \lim_{x \to N+\ell/2} (x- N - \ell/2)^2 R^+(x) R^{-}(x)  = \Res_{x= N+\ell/2} R^+(x) \Res_{x= N+\ell/2} R^-(x),
  \]
  and then the proof follows as in Proposition \ref{prop:polenhi1}. The second claim follows from the basic identity
  \[
    \Res_{x=a} \left(\frac{R_1}{x-a} + A_1(x)\right) \left(\frac{R_2}{x-a} + A_2(x)\right) = R_1 A_2(a) + R_2 A_1(a),
  \]
  valid for functions $A_1(x)$, $A_2(x)$ analytic at \(x=a\) and \(R_1,R_2 \in \C \). 
\end{proof}

By comparing the recurrence relations of $R^{\pm}(x;g,\Delta)$ and $\bar{R}^{\pm}(x;g,\Delta)$ with the residues (cf. the proof of Proposition \ref{prop:polenhi1}) the functions \(Q^{-}(x;g,\Delta) \), \(Q^{+}(x;g,\Delta) \), \(\bar{Q}^{-}(x;g,\Delta) \) and \(\bar{Q}^{+}(x;g,\Delta) \) can also
be expressed by recurrence relations. Namely, if we set \(s_{-1}(x;g,\Delta)= 0\), \(s_{0}(x;g,\Delta)=1 \),
\[
  s_{N+1}(x;g,\Delta) = \frac{1}{N+1}\left(\left(2g + \frac{N-x-\ell/2}{2g} \right) s_{N}(x;g,\Delta) - s_{N-1}(x;g,\Delta) \right),
\]
and
\[
  s_k(x;g,\Delta) = \frac{1}{k}\left( f_{k-1}^-(x,g,\Delta,\e) s_{k-1}(x;g,\Delta) - s_{k-2} (x;g,\Delta) \right),
\]
for positive integer \(k \neq N+1 \), we have
\begin{equation}
  \label{eq:qfunctdef-}
  Q^-(x;g,\Delta) = \sum_{n=0}^{\infty} s_n(x;g,\Delta) g^n, \qquad
  \bar{Q}^-(x;g,\Delta) = \sum_{\substack{n\ge 0 \\ n\ne N}} \frac{s_n(x;g,\Delta)}{x-n-\ell/2} g^n.
\end{equation}
Similarly, setting \(r_{-1}(x;g,\Delta)= 0\), \(r_{0}(x;g,\Delta)=1 \), and
\begin{align*}
  r_{N+\ell+1}(x;g,\Delta) = \frac{1}{N+\ell+1} \left(\left(2g + \frac{N-x+\ell/2}{2g} \right) r_{N+\ell}(x;g,\Delta) - r_{N+\ell-1}(x;g,\Delta) \right), 
\end{align*}
and
\[
  r_k(x;g,\Delta) = \frac{1}{k}\left( f_{k-1}^+(x,g,\Delta,\e) r_{k-1}(x;g,\Delta) - r_{k-2} (x;g,\Delta) \right),
\]
for positive integer \(k \ne N+\ell+ 1\), we have
\begin{equation}
  \label{eq:qfunctdef+}
  Q^+(x;g,\Delta) = \sum_{n=0}^{\infty} r_n(x;g,\Delta) g^n, \qquad
  \bar{Q}^+(x;g,\Delta) = \sum_{\substack{n\ge 0 \\ n\ne N+\ell}} \frac{r_n(x;g,\Delta)}{x-n+\ell/2} g^n.
\end{equation}
Note that by Theorem \ref{thm:Main}, when \(\lambda = N + \ell/2 - g^2 \) is a Juddian eigenvalue, the coefficients of the
poles of \(G_{\ell/2}(x;g,\Delta)\) vanish and the function \(G_{\ell/2}(x;g,\Delta)\) has a finite value at \(x = N+\ell/2 \). However, it is possible
to find numerically examples of parameters such that \(G_{\ell/2}(x;g,\Delta)\) has a pole at \(x = N+ \ell/2 \) yet \( \lambda= N + \ell/2 - g^2 \) is a
non-Juddian exceptional eigenvalue. One such example is shown in Figure \ref{fig:gfunct_graph3}(c) for the parameters
\(\e = 1/2, g \approx 1.3903, \Delta = 1 \). In this case, there is a pole of \(G_{1/2}(x;g,1)\) even though the parameters \(g\) and \(\Delta\) correspond
(numerically) to a zero of \(T_{1/2}^{(1)}(g,1)\)  at \(x = 1.5 \). We remark that the pole \(x = 1.5 \) must be simple.
Indeed, in the notation of Proposition \ref{prop:polehi}, since \(T_{1/2}^{(1)}(g,1)=0\) the second order term \(A\) vanishes while
the residue term \(B\) is non-vanishing. This is also apparent in the graph of \(G_{1/2}(x;g,1)\) in Figure \ref{fig:gfunct_graph3}(c), since the
lateral limits at \(x=1.5\) have different signs the pole must be simple and the term \(B\) must be non-zero in a neighborhood
of \(x=1.5\).

The situation for the poles of the \( G_{\ell/2}(x;g,\Delta)\) is summarized in the following result.

\begin{cor}

  Suppose \( \ell \in \Z_{\geq 0}\) and  \(\Delta>0 \). The $G$-function \( G_{\ell/2}(x;g,\Delta)\) has \(\ell\) poles of order \(\leq 1\)
  at \(x = N - \ell/2 \) for \(0 \leq N < \ell\) and poles of order \(\leq 2\) at \(x = N + \ell/2 \) for \(N \in \Z_{\geq 0}\). Moreover, for
  \(N \in \Z_{\geq 0} \), we have:
  \begin{itemize}
  \item If \(\lambda = N \pm \ell/2 -g^2\) is a Juddian eigenvalue of \(\HRabi{\ell/2}\), then \(x = N \pm \ell/2 \) is not a pole of \(G_{\ell/2}(x;g,\Delta) \).
  \item For \(0 \leq N < \ell \), the function \(G_{\ell/2}(x;g,\Delta) \) does not have a pole at \(x = N - \ell/2 \)  if and only if \(\lambda = N - \ell/2 - g^2 \) is a
    non-Juddian exceptional eigenvalue of \(\HRabi{\ell/2} \).
  \item If \(G_{\ell/2}(x;g,\Delta) \) has a simple pole at \(x = N + \ell/2 \), then \(\lambda = N + \ell/2 - g^2 \) is a non-Juddian exceptional eigenvalue
    of \(\HRabi{\ell/2} \).
  \item If \(G_{\ell/2}(x;g,\Delta) \) has a double pole at \(x = N \pm \ell/2 \), then there is no exceptional eigenvalue \(\lambda = N \pm \ell/2 - g^2\)
    of \(\HRabi{\ell/2} \).
  \end{itemize}
\end{cor}

\begin{rem}
  In the case of the QRM (i.e. \(\e=0 \)), all the singularities of the $G$-function \(G_{0}(x;g,\Delta) \) are of the type described in Proposition \ref{prop:polehi}
  (i.e. poles of order \(\leq 2\)).
\end{rem}

Note that is possible that non-Juddian exceptional eigenvalues corresponding to finite values of $G_{\ell/2}(x;g,\Delta)$ at points \(x = N \pm \e \)
are present in the spectrum. If such eigenvalues were to exist then the structure of the poles of the $G$-function alone would not be sufficient
to completely discriminate the structure of the exceptional spectrum.

To get a better understanding of the vanishing of the residues \(\Res_{x = N + \ell/2}  G_{\ell/2}(x;g,\Delta)\), we define the function
\begin{align*}
  B^N_{\ell}(g,\Delta) &:= \frac{1}{C(N)}\Big(\bar{R}^{(N,+)}(g,\Delta,\frac{\ell}{2}) (\Delta\bar{Q}^-(N+\frac{\ell}{2};g,\Delta)) \\
               & \qquad \qquad \qquad -   R^{(N,+)}(g,\Delta,\frac{\ell}{2})  Q^-(N+\frac{\ell}{2};g,\Delta)  \Big).
\end{align*}

This function may be thought of a ``regularized'' $T$-function. In the case \(T^{(N)}_{\ell/2}(g,\Delta)=0\) (thus
\(\cp{N,\ell/2}{N}((2g)^2,\Delta^2) \not= 0 \)), by Proposition \ref{prop:polehi} and the proof of Lemma \ref{lem:const-T-rel},
the vanishing of the residue \(\Res_{x = N + \ell/2}  G_{\ell/2}(x;g,\Delta)\) is equivalent to the equation
\begin{equation}\label{eq:divisibilityB}
  B^{N+\ell}_{-\ell}(g,\Delta) + A_N^{\ell}((2g)^2,\Delta^2) B^N_{\ell}(g,\Delta) = 0,
\end{equation}
resembling the divisibility problem of constraint polynomials. Actually, this equation distinguishes the cases where $x=N+\ell/2$
is a pole or not and the polynomial $A_N^{\ell}((2g)^2,\Delta^2)$ again may play a particular role for its determination.
Also, in terms of the constraint polynomials, we notice that $B=0$, i.e. \eqref{eq:divisibilityB}, is equivalent to
\begin{align*}
 & \det
  \begin{bmatrix}  P_{N}^{(N, \ell/2)}((2g)^2,\Delta^2) & -B_\ell^N(g,\Delta)  \\
  P_{N+\ell}^{(N+\ell, -\ell/2)}((2g)^2,\Delta^2) & B_{-\ell}^{N+\ell}(g,\Delta)
  \end{bmatrix}
  \\
  &\qquad \qquad= P_{N}^{(N, \ell/2)}((2g)^2,\Delta^2)
  \det \begin{bmatrix}
  1 & -B_\ell^N(g,\Delta)  \\
  A_N^\ell((2g)^2,\Delta^2) & B_{-\ell}^{N+\ell}(g,\Delta)
  \end{bmatrix}
  =0.
\end{align*}

\begin{prob} \label{pr:non-Juddian}
  With the notation of Proposition \ref{prop:polehi},
  \begin{itemize}
  \item  Are there non-Juddian exceptional eigenvalues \(\lambda = N \pm \ell/2 -g^2\) corresponding to finite values of
    $G$-function $G_{\ell/2}(x;g,\Delta)$ at the point \(x = N \pm \ell/2 \)? If the answer is affirmative, what are the
    properties of these non-Juddian exceptional eigenvalues?
  \item More concretely, can we characterize the vanishing of $B$ (equivalently \eqref{eq:divisibilityB})
    in terms of the function $T^{(N)}_{\ell/2}(g,\Delta)$?
    It would be quite interesting if the vanishing of $B$ can be formulated
    as a sort of duality of the equation
    $P_{N+\ell}^{(N+\ell, -\ell/2)}= A_N^\ell\cdot P_N^{(N,\ell/2)}$ in Theorem \ref{thm:Main}.
    We actually notice that $B=0$ is equivalent to the fact that the vector
    ${}^{t\!}\begin{bmatrix}B_{-\ell}^{N+\ell} & B_\ell^N\end{bmatrix}$
    is perpendicular to the vector
    ${}^{t\!}\begin{bmatrix}P_{N}^{(N, \ell/2)} & P_{N+\ell}^{(N+\ell, -\ell/2)}\end{bmatrix}$.
  \end{itemize}
\end{prob}

As a first step for the understanding of this problem, we present the graphs in the \((g,\Delta)\)-plane of the curves defined by the residue
vanishing condition \eqref{eq:divisibilityB} and the constraint conditions for exceptional eigenvalues in Figure \ref{fig:gfunctBvanish}.
In the graphs, we show the curve described by \( T^{(N)}_{\ell/2}(g,\Delta) = 0\) in continuous gray lines, the curve given by
\( \cp{N,\ell/2}{N}((2g)^2,\Delta^2)= 0 \) in dashed gray lines and the residue vanishing condition \eqref{eq:divisibilityB} in black lines. Figure
\ref{fig:gfunctBvanish}(a) shows the case $N=1$ and \(\ell=2 \) while Figure \ref{fig:gfunctBvanish}(b) depicts the case $N=3$ and $\ell = 1$.
Notice that in both cases there appears to be intersections in the vanishing condition \eqref{eq:divisibilityB} and the constraint relation
\( T^{(N)}_{\ell/2}(g,\Delta) = 0\), in other words, there are non-Juddian eigenvalues which kill the corresponding (double) poles of the $G$-function
$G_{\ell/2}(x;g,\Delta)$. While further investigation including numerical experiments is needed, the observations made on the numerical graphs
shown in Figure \ref{fig:gfunctBvanish} provide actually an evidence for the affirmative answer of the problem above. In addition, from Figure \ref{fig:gfunctBvanish},
we notice there are apparently no intersections between the curves of the Juddian constraint conditions and the curves of the vanishing condition
\eqref{eq:divisibilityB}, which may be related to the perpendicularity described in the problem above. Actually, further numerical experimentations
we have done so far support that this observation can be true in general.

\begin{figure}[htb]
  ~
  \begin{subfigure}[b]{0.45\textwidth}
    \centering
    \includegraphics[height=5cm]{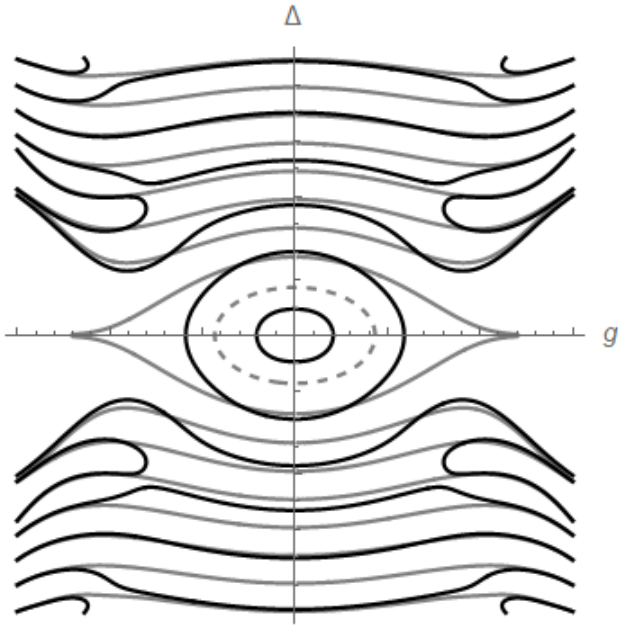}
    \caption{\(N=1,\, \ell = 2\)}
  \end{subfigure}
  ~
  \begin{subfigure}[b]{0.45\textwidth}
    \centering
    \includegraphics[height=5cm]{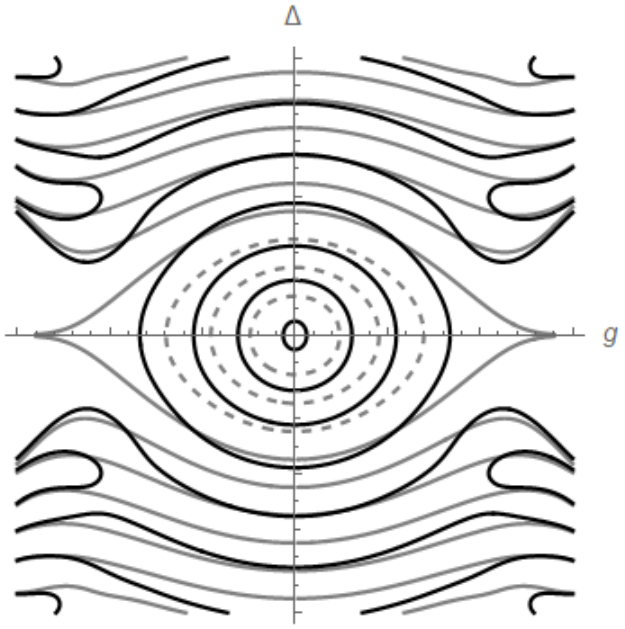}
    \caption{\(N = 3,\, \ell = 1\)}
  \end{subfigure}
  \caption{Plot of constraint relations $T^{(N)}_{\ell/2}(g,\Delta) = 0$ (gray), $\cp{N,\ell/2}{N}((2g)^2,\Delta^2)= 0$ (dashed gray) and
  residue vanishing condition \eqref{eq:divisibilityB} (black).}
  \label{fig:gfunctBvanish}
\end{figure}

\begin{rem}
  We make a small remark about Problem \ref{pr:non-Juddian}.
  It is obvious that
  \begin{equation}
    B^N_{\ell}(g,\Delta) = \frac{1}{C(N)} \det
    \begin{bmatrix}
      \bar{R}^{(N,+)}(g,\Delta,\ell/2) & Q^-(N+\ell/2;g,\Delta) \\
      R^{(N,+)}(g,\Delta, \ell/2) & \Delta \bar{Q}^-(N+\ell/2;g,\Delta)
    \end{bmatrix}.
  \end{equation}
  It follows also that
  \begin{align*}
    B^{N+\ell}_{-\ell}(g,\Delta) &= \frac{1}{C(N+\ell)}\det
                        \begin{bmatrix}
                          \bar{R}^{(N+\ell,+)}(g,\Delta,-\ell/2) & Q^+(N+\ell/2;g,\Delta) \\
                          R^{(N+\ell,+)}(g,\Delta, -\ell/2) & \Delta \bar{Q}^+(N+\ell/2;g,\Delta)
                        \end{bmatrix}\\
                      &= \frac{1}{C(N+\ell)}\det
                        \begin{bmatrix}
                          \bar{R}^{(N,-)}(g,\Delta,\ell/2) & Q^+(N+\ell/2;g,\Delta) \\
                          R^{(N,-)}(g,\Delta, \ell/2) & \Delta \bar{Q}^+(N+\ell/2;g,\Delta)
                        \end{bmatrix}.
  \end{align*}
  Since
  \begin{equation}
    T^{(N)}_{\ell/2}(g,\Delta) = \det
    \begin{bmatrix}
      \bar{R}^{(N,+)}(g,\Delta,\ell/2) & R^{(N,-)}(g,\Delta, \ell/2)  \\
      R^{(N,+)}(g,\Delta, \ell/2) & \bar{R}^{(N,-)}(g,\Delta,\ell/2)
    \end{bmatrix},
  \end{equation}
  if $T^{(N)}_{\ell/2}(g,\Delta)=0$, there exists a constant $c_N:=c_N(g,\Delta,\ell/2)(\not=0)$ (cf. the
  discussion in \S~\ref{sec:NonJuddian} for the definition of $T^{(N)}_{\e}(g,\Delta)$) such that
  \begin{equation}
    \begin{bmatrix}
      R^{(N,-)}(g,\Delta, \ell/2)  \\
      \bar{R}^{(N,-)}(g,\Delta,\ell/2)
    \end{bmatrix}
    =   c_N\,
    \begin{bmatrix}
      \bar{R}^{(N,+)}(g,\Delta,\ell/2)\\
      R^{(N,+)}(g,\Delta, \ell/2)
    \end{bmatrix}.
  \end{equation}
  It follows that
  \begin{equation}
    B^{N+\ell}_{-\ell}(g,\Delta) = -\frac{c_N}{C(N+\ell)}\det
    \begin{bmatrix}
      \bar{R}^{(N,+)}(g,\Delta,\ell/2) &  \Delta \bar{Q}^+(N+\ell/2;g,\Delta) \\
      R^{(N,+)}(g,\Delta, \ell/2) & Q^+(N+\ell/2;g,\Delta)
    \end{bmatrix}.
  \end{equation}
  Hence, if there exists a constant $\beta_N=\beta_N(g,\Delta, \ell/2)$ such that
  \begin{align*}
     \begin{bmatrix}
      \bar{R}^{(N,+)}(g,\Delta,\ell/2)\\
      R^{(N,+)}(g,\Delta, \ell/2)
    \end{bmatrix}
    = &\beta_N
        \Big\{
        \frac{c_N}{C(N+\ell)}
        \begin{bmatrix}
          \Delta \bar{Q}^+(N+\ell/2;g,\Delta) \\
          Q^+(N+\ell/2;g,\Delta)
        \end{bmatrix} \\
    &+ \frac{A_N^{\ell}((2g)^2,\Delta^2)}{C(N)}
    \begin{bmatrix}
      Q^-(N+\ell/2;g,\Delta) \\
      \Delta \bar{Q}^-(N+\ell/2;g,\Delta)
    \end{bmatrix}
    \Big\} ,
  \end{align*}
  then the equation \eqref{eq:divisibilityB} holds.
\end{rem}

\;

In the paper of Li and Batchelor \cite{LB2016JPA}(p. 4), the authors define a new $G$-function $\mathcal{G}_\e(x;g,\Delta)$
for numerical computation of the spectrum of the AQRM. The new definition uses a divergent product to make the function
$\mathcal{G}_\e(x;g,\Delta)$ vanish for all eigenvalues of AQRM, including the exceptional ones (i.e. at \(x = N \pm \e\)).
We note that, however, it is not well-defined theoretically due to the use of the divergent product. Nevertheless, according to the
following theorem, the numerical observation in \cite{LB2016JPA} by taking a certain truncation of the divergent product does seem to work
properly. To obtain a correct understanding, we use the gamma function $\Gamma(x)$ to alternatively define the new $G$-function $\mathcal{G}_\e(x;g,\Delta)$ as
\begin{equation}\label{eq:neqG}
  \mathcal{G}_\e(x;g,\Delta) := G_\e(x;g,\Delta)\Gamma(\e-x)^{-1}\Gamma(-\e-x)^{-1}.
\end{equation}

As a consequence of our discussion above on the poles of the $G$-function, we can establish the claim made in \cite{LB2016JPA}.

\begin{thm} \label{thm:LBcomp}
  For fixed \(g,\Delta>0\), \(x\) is a zero of \(\mathcal{G}_\e(x;g,\Delta)\) if and only if \(\lambda = x - g^2 \) is an eigenvalue of  \(\HRabi{\e}\).
\end{thm}

\begin{proof}
  The statement  for regular eigenvalues is clear since the factor $\Gamma(\e-x)^{-1}\Gamma(-\e-x)^{-1}$ does not contribute any further zeros in this case.
  Next, suppose \(\e \not\in \frac12 \Z\). Then, the point \(x = N+\e \) is a simple zero of
  \(\Gamma{(\e-x)}^{-1} \), therefore \(\mathcal{G}_\e(N+\e;g,\Delta) = C \Res_{x = N+\e} G_\e (x;g,\Delta) \) for
  a nonzero constant \(C \in \C \) and the result follows from Proposition \ref{prop:polenhi1}.
    In the case of \( \e = \ell/2 \, (\ell \in \Z_{\geq 0})\), the result for \(x = N - \ell/2\) with \(0 \leq N < \ell\) follows by
    Proposition~\ref{prop:polehi1simple} in the same way as the case \(\e \not\in \frac12 \Z \). Similarly,
    notice that the double zero of \(\Gamma(\ell/2-x)^{-1}\Gamma(-\ell/2-x)^{-1}\) at \(x = N +\ell/2 \, (N \in \Z_{\geq 0}) \) makes
    \(\mathcal{G}_\e(x;g,\Delta)\) equal (up to a nonzero constant) to the coefficient $A$ of \((x-N-\ell/2)^{-2}\) in the Laurent expansion
    of \(G_\e(x;g,\Delta)\) at \(x= N+ \ell/2 \) given in the Proposition \ref{prop:polehi}. Hence the theorem follows. 
\end{proof}

We now recall the so-called spectral determinant of the Hamiltonian $\HRabi{\e}$ of the AQRM.
Let $\lambda_0<\lambda_1\leq \lambda_2\leq \lambda_3\leq \lambda_4 \leq\ldots$ be the set of all eigenvalues of the Hamiltonian $\HRabi{\e}$ of the AQRM. Here note that the first
eigenvalue $\lambda_0$ is always simple (see the proof of Corollary \ref{DegenerateStructure}). Then the Hurwitz-type spectral zeta function of the
AQRM is defined by
\begin{equation}
  \zeta_{\HRabi{\e}}(s,\tau)= \sum_{i=0}^\infty (\tau- \lambda_i)^{-s}, \quad \RE(s)>1.
\end{equation}
Here we fix the log-branch by $-\pi\leq \arg(\tau- \lambda_i)<\pi$. We then define the zeta regularized product (cf. \cite{QHS1993TAMS}) over the spectrum of
the AQRM as
\begin{equation}\label{eq:ZRProduct}
  \regprod_{i=0}^\infty (\tau-\lambda_i):= \exp\big(-\frac{d}{ds}\zeta_{\HRabi{\e}}(s,\tau)\big|_{s=0}\big).
\end{equation}
We can prove that $\zeta_{\HRabi{\e}}(s,\tau)$ is holomorphic at $s=0$ by the same way as in the case of the QRM \cite{Sugi2016}. Actually, the meromorphy
of $\zeta_{\HRabi{\e}}(s,\tau)$ in the whole plane $\C$ follows in a similar way to the case $\HRabi{0}$. As the notation may indicate that this regularized
product is an entire function possessing its zeros exactly at the eigenvalues of $\HRabi{\e}$. Now define the spectral determinant of the
AQRM $\det \HRabi{\e}$ as
\begin{equation}
  \det (\tau-\HRabi{\e}):= \regprod_{i=0}^\infty (\tau-\lambda_i).
\end{equation}

The following result follows immediately from Theorem \ref{thm:LBcomp}.

\begin{cor}\label{cor:SDE}
  There exists an entire non-vanishing function $c_{\e}(\tau ; g,\Delta)$ such that
  \begin{equation}
    \det (\tau-g^2-\HRabi{\e}) =  c_{\epsilon}(\tau;g,\Delta) \mathcal{G}_{\epsilon}(\tau;g,\Delta).
    \QEDhere
  \end{equation}
\end{cor}

\begin{rem} \label{rem:gammafactor}
  The product of the two gamma functions in  \eqref{eq:neqG} may be interpreted as a sort of  ``gamma factor'' in a sense of the zeta function
  theory (see e.g. \cite{Gon1955PJA}) but it is hard to expect any functional equation satisfied by this spectral determinant. On the other hand,
  the study of the special values $\zeta_{\HRabi{\e}}(n,0)\, (n=2, 3,\ldots)$ of the spectral zeta function as in \cite{KW2007} is awaited in the future
  research (see \cite{RW2019,Sugi2016} for the special values at negative integers). This study may build a bridge between the arithmetics (particularly,
  modular forms and Ap\'ery-like numbers, e.g. \cite{Z2009}) and the spectrum of the AQRM. Actually, the value $\frac{d}{ds}\zeta_{\HRabi{\e}}(0,0)$ (see
  \eqref{eq:ZRProduct})
  may be thought to giving an analogue of the Euler constant (the constant term of the Riemann zeta function $\zeta(s)$ at $s=1$) and hence, considered also
  as the ``special value"
  (constant term) at $s=1$ which describes the Weyl law \cite{Sugi2016} (notice that since $s=1$ is a pole of the zeta functions,
  $\frac{d}{ds}\zeta(0)$ can be considered as the special value of $\zeta(s)$ at $s=1$ by the functional equation). This result is important because,
  among other reasons, the Weyl law  is relevant to the conjecture on the distribution of eigenvalues for $\HRabi{0}$ by Braak \cite{B2011PRL}.
   The conjecture by Braak on the distribution of eigenvalues of the QRM (i.e. $\e=0$ case) can be summarized \cite{BW_personal2017} as
   follows: The number of eigenvalues $\HRabi{0}$ in each interval $[n, n+1)$ $(n\in \Z)$ is restricted to $0, 1$ or $2$ for a given parity
   (see Remark \ref{rem:Parity}), and moreover, two intervals $[n,n+1)$ containing two (resp. no) eigenvalues (i.e. two (resp. no) roots of
   $G_{\pm}(x)$) cannot be contiguous.
\end{rem}
\,

\begin{rem}
  We may define also the following two  functions by zeta regularized products over Juddian and non-Juddian exceptional eigenvalues for
  given $g$ and $\Delta$:
  \begin{align}
    \Gamma_{\e, \text{Judd}} (\tau)^{-1} :&= \regprod_{N \in \Z_{\geq0}: \; P_N^{(N, \e)}((2g)^,\Delta^2)=0} (\tau-(N+\e)),\\
    \Gamma_{\e, \text{n-Judd}} (\tau)^{-1}:&= \regprod_{N \in \Z_{\geq0}: \; T_{\e}^N(g,\Delta)=0}(\tau-(N+\e)).
  \end{align}
  Note that if the number of Juddian (resp. non-Juddian exceptional) eigenvalues is finite, then the regularized product $\Gamma_{\e, \rm{Judd}} (\tau)^{-1} $
  (resp. $\Gamma_{\e, \rm{n-Judd}} (\tau)^{-1}$) is actually a standard product. In fact, as widely believed among physicists, it can be strongly
  conjectured that the product
  \[
    \prod_{a=\pm\e} \, \prod_{b=\text{Judd, \,n-Judd}} \Gamma_{a, b} (\tau)^{-1}
  \]
  is a polynomial, that is, almost all the eigenvalues of $\HRabi{\e}$ are regular for fixed parameters \(g,\Delta>0\). If it is true, it is also
  quite interesting to determine also the degree of this polynomial, more precisely the degrees of $\Gamma_{\e, \,\rm{Judd}}(y)^{-1}$ and
  $\Gamma_{\e,\,\rm{n-Judd}}(y)^{-1}$ in terms of the parameters $g$ and $\Delta$, and to study any symmetry appearing in them (see also Remark
  \ref{rem:symmetry_spec}). In order to study these functions and possible symmetry, it is necessary to analyze $\det (\tau-\HRabi{\e})$ more
  deeply than in \cite{Sugi2016}. Actually, it seems difficult to obtain a sufficient knowledge from the current method (the study of coefficients
  of the asymptotic expansion of the trace of the heat kernel of $\HRabi{\e}$) of meromorphic extension of $\zeta_{\HRabi{0}}(s,\tau)$ investigated
  in \cite{Sugi2016} (see also \cite{IW2005a,P2010S,Robert1978}). It might also be useful to consider the spectrum from the viewpoints of
  dynamics as in \cite{P2014Milan}.

  We conclude this remark by discussing the finiteness of the product in  \(\Gamma_{\e, \rm{Judd}} (\tau)^{-1}\) for the degenerate cases \(g =0 \) and
  \(\Delta = 0 \). First, suppose that \(g = 0\). Directly from the definition it is clear that \(\cp{N,\e}{N}(0,i(i+2\e)) = 0\) for any
  \(i \in \Z_{>0}\) and \(N \geq i\), thus for \(\Delta = \sqrt{i(i+2\e)}\) the gamma factor \(\Gamma_{\e, \rm{Judd}} (\tau)^{-1}\) is an infinite product, essentially
  given by $\Gamma(\e-\tau)^{-1}\prod_{N=0}^{i-1}(\tau-N-\e)^{-1}$.

  On the other hand, suppose that \(\Delta = 0 \) and \(\e \geq  0\). In this case, by Theorem \ref{thm:Laguerre} the constraint polynomials
  \(\cp{k,\e}{k}(x,0) \, (k \in \Z_{\geq 0})\) are given by \( (-1)^k (k!)^2 L^{(2\e)}_k (x)\), where  \(L^{(2\e)}_k (x)\) are the generalized Laguerre
  polynomials. It is well-known that orthogonal polynomials interlace zeros strictly (see e.g. \cite{AAR1999,C1978}) and therefore
  \(\Gamma_{{\e}, \rm{Judd}} (\tau)^{-1}\) is a polynomial of degree \(1\) or \(0\) depending on whether \( L^{(2\e)}_k (g^2) = 0\) for some
  \(k \in \Z_{\geq 0} \) or not. Furthermore, when \(\e = \ell/2 \, (\ell \in \Z_{\geq 0})\), we have shown in Proposition \ref{cor:noNegJudd} that
  for \(0 \leq k < \ell\) the constraint polynomials \(\cp{k,-\ell/2}{k}(x,0)\) have no positive roots. Moreover, by the divisibility
  of Theorem~\ref{thm:Main}, we have \(\cp{N+\ell,-\ell/2}{N+\ell}(x,0) = A_N^{\ell}(x,0)\cp{N,\ell/2}{N}(x,0)\) for \( N \in \Z_{\geq 0} \) . It follows that
  \(  \Gamma_{{\ell/2}, \rm{Judd}} (\tau)^{-1} \, \Gamma_{{-\ell/2}, \rm{Judd}} (\tau)^{-1}\) is a polynomial of degree \( 2\) or \(0\), depending on whether
  \( L^{(\ell)}_k (g^2)=0\) for some \(k \in \Z_{\geq 0} \) or not.
\end{rem}

\section{Representation theoretical picture of the spectrum} \label{sec:repr-theor-pict}

In this section, we illustrate how the spectrum of the AQRM can be captured by irreducible representations of the Lie algebra $\mathfrak{sl}_2$. In fact, we find that Juddian eigenstates are identified with vectors in the finite dimensional irreducible representation \cite{W2016JPA}, non-Juddian exceptional eigenstates with vectors in irreducible lowest weight representations, and regular eigenstates with vectors in (non-unitary) irreducible principal series representations.

\subsection{A brief review for \texorpdfstring{$\mathfrak{sl}_2$}{sl2}-representations}

In this section we introduce the necessary background from representation theory
of the Lie algebra \(\mathfrak{sl}_2(\R)\) and/or \(\mathfrak{sl}_2(\C)\). The reader is directed to \cite{W2016JPA} for an extended discussion and \cite{L,HT1992} for the
general theory of \(\mathfrak{sl}_2\)-representations.

The standard generators $H, E$ and $F$ of $\mathfrak{sl}_2(\R)$ are given by
\begin{align*}
  H= \begin{bmatrix}
    1 & 0  \\
    0 &  -1
  \end{bmatrix},\quad
  E= \begin{bmatrix}
    0 & 1  \\
    0 &  0
  \end{bmatrix},\quad
  F= \begin{bmatrix}
    0 & 0  \\
    1 &  0
  \end{bmatrix}.
\end{align*}
These generators satisfy the commutation relations
\[
  [H,\, E] = 2E,\quad [H,\, F] = -2F,\quad  [E,\, F] = H.
\]
For \(a \in \C\) define the algebraic action \(\varpi_a\) of $\mathfrak{sl}_2$ on
the vector spaces $\rV_{1}:= y^{-\frac14} \C[y, y^{-1}]$ and $\rV_{2}:=y^{\frac14} \C[y, y^{-1}]$  given by
\begin{gather*}
  \varpi_a(H) :=2y\partial_y+\frac12,\quad
  \varpi_a(E) :=y^2\partial_y+\frac12(a+\frac12)y,\\
  \varpi_a(F) := -\partial_y+\frac12(a-\frac12)y^{-1}
\end{gather*}
with \(\partial_y := \frac{d}{d y} \). It is not difficult to verify that these operators indeed act on the
space $\rV_j (j=1,2)$, and define infinite dimensional representations of $\mathfrak{sl}_2$. 
We call $\rV_1$ the spherical (resp. $\rV_2$ the non-spherical) representation.
Write  $\varpi_{j,a} :=\varpi_a|_{\rV_{j}}$ and put $e_{1,n}:=y^{n-\frac14}$ and $ e_{2,n}:=y^{n+\frac14}$.
Then we have
\begin{equation*}
  \left\{
  \begin{aligned}
    \varpi_{1,a}(H)e_{1,n} &= 2ne_{1,n},\\
    \varpi_{1,a}(E)e_{1,n} &= \big(n+\frac{a}{2}\big)e_{1,n+1},\\
    \varpi_{1,a}(F)e_{1,n} &= \big(-n+\frac{a}{2}\big)e_{1,n-1},
  \end{aligned}
  \right.
  \,\,
  \left\{
  \begin{aligned}
    \varpi_{2,a}(H)e_{2,n} &= (2n+1)e_{2,n},\\
    \varpi_{2,a}(E)e_{2,n} &= \big(n+\frac{a+1}{2}\big)e_{2,n+1},\\
    \varpi_{2,a}(F)e_{2,n} &= \big(-n+\frac{a-1}{2}\big)e_{2,n-1}.
  \end{aligned}
  \right.
\end{equation*}
Note that $(\varpi_{1,a},\rV_{1})$ (resp. $(\varpi_{2,a},\rV_{2})$) is irreducible
when $a\not\in 2\Z$ (resp. $a\not\in 2\Z-1$) and that there is an equivalence between
$\varpi_{j,a}$ and $\varpi_{j,2-a}$ under the same condition. We call such irreducible representation a principal series.
Next, for a non-negative integer $m$, define subspaces $\rD^{\pm}_{2m},\rF_{2m-1}$ of
$\rV_{1,2m}(=\rV_{1})$, and  $\rD^{\pm}_{2m+1},\rF_{2m}$ of
$\rV_{2,2m+1}(=\rV_{2})$ respectively by
\begin{equation*}
  \rD^{\pm}_{2m}:=\bigoplus_{n\geq  m}\C\cdot e_{1,\pm n},\quad
  \rF_{2m-1}:=\bigoplus_{-m+1\leq n\leq m-1}\C\cdot e_{1,n},
\end{equation*}
\begin{equation*}
  \rD^{-}_{2m+1}:=\bigoplus_{n\geq m+1}\C\cdot e_{2,-n},\,\,
  \rD^{+}_{2m+1}:=\bigoplus_{n\geq m}\C\cdot e_{2,n},\,\,
  \rF_{2m}:=\bigoplus_{-m\leq n\leq m-1}\C\cdot e_{2,n}.
\end{equation*}
The spaces $\rD^{\pm}_{2m}$ (resp. $\rD^{\pm}_{2m+1}$) are invariant under the action $\varpi_{1,2m}(X)$,  (resp. $\varpi_{1,2m+1}(X)$), $(X\in \mathfrak{sl}_2)$, and define irreducible representations (having the lowest
and highest weight vector respectively) known to be equivalent to (holomorphic and anti-holomorphic) discrete series for $m>0$ of
$\mathfrak{sl}_2(\R)$. The irreducible representation $\rD^{\pm}_{1}$ are the (infinitesimal version of) limit of discrete series of
$\mathfrak{sl}_2(\R)$ (see e.g. \cite{HT1992,L}). Moreover, the finite dimensional space $\rF_{m}$ ($\dim_{\C} \rF_{m}=m$),  is invariant
and defines irreducible representation of $\mathfrak{sl}_2$ for $a=2-2m$ when $j=1$ and $a=1-2m$ when $j=2$, respectively.

The following result describes the irreducible decompositions of \((\varpi_a, \, \rV_{j,a})\), \((a=m\equiv j-1 \mod 2)\) for $m\in \Z_{\geq 0}$ and $j=1,2$.

\begin{lem} \label{lem:reducible}
Let $m\in \Z_{\geq 0}$.
\begin{enumerate}[$(1)$]
\item The subspaces $\rD^{\pm}_{2m}$ are irreducible submodules of $\rV_{1,2m}$ under the action
  $\varpi_{1, 2m}$ and $\rF_{2m-1}$ is an irreducible submodule of $\rV_{1,2-2m}$ under
  $\varpi_{1, 2-2m}$. In the former case, the finite dimensional irreducible representation $\rF_{2m-1}$
  can be obtained as the subquotient as $\rV_{1,2m}/\rD^{-}_{2m}\oplus \rD^{+}_{2m}\cong   \rF_{2m-1} $.
  In the latter case, the discrete series $\rD^{\pm}_{2m}$ can be realized as the irreducible components
  of the subquotient representation as $\rV_{1,2-2m}/\rF_{2m-1} \cong   \rD^{-}_{2m}\oplus \rD^{+}_{2m}$.
\item  The subspaces $\rD^{\pm}_{2m+1}$ are irreducible submodule of $\rV_{2,2m+1}$ under the action
  $\varpi_{2, 2m+1}$ and $\rF_{2m}$ is an irreducible submodule of $\rV_{2,1-2m}$ under $\varpi_{2, 1-2m}$.
  In the former case, the finite dimensional irreducible representation $\rF_{2m}$ can be obtained as the
  subquotient as $\rV_{2,2m+1}/\rD^{-}_{2m+1}\oplus \rD^{+}_{2m+1}\cong   \rF_{2m}$. In the latter case, the discrete
  series $\rD^{\pm}_{2m+1}$ can be realized as  the irreducible components of the subquotient representation
  as $\rV_{2,1-2m}/\rF_{2m} \cong   \rD^{-}_{2m+1}\oplus \rD^{+}_{2m+1}$.
\item The space $\rV_{2,1}$ is decomposed as the irreducible sum:
  $\rV_{2,1}=\rD^{-}_{1}\oplus \rD^{+}_{1}$.
\end{enumerate}
Moreover, the spaces of irreducible submodules
$\rD^{\pm}_{m} (\subset \rV_{j,m})$, $\rF_{m} (\subset \rV_{j,1-m})$ and
the direct sum $\rD^{+}_{m}\oplus \rD^{-}_{m} (\subset\rV_{j,m})$ above are the only
non-trivial invariant subspaces of $\rV_{j,m}$ for $j=1$ (resp. $j=2$) when $m$ is even (resp. odd) under
the action of $\mathcal{U}(\mathfrak{sl}_2)$, the universal enveloping algebra of $\mathfrak{sl}_2$. \QEDhere
\end{lem}

\begin{rem}
  By Lemma \ref{lem:reducible}, when \(m \in \Z_{> 0}\) the space \( \rF_{2m+1} \) is realized as an invariant subspace of \( \rV (= \rV_{1,2-2m})\)
  while the space \(\rD^-_{2 m}\oplus \rD^+_{2 m}\) is just obtained as a subquotient of \( \rF_{2m+1} \), that is, \(\rD^\pm_{2 m}\) do not constitute
  invariant spaces of $\rV_{1,2-2m}$. In other words, \(\rV_{1,2-2 m} \) cannot have the direct sum decomposition
  \(\rV_{1,2-2 m} = \rD^-_{2 m} \oplus \rF_{2m+1}  \oplus \rD^+_{2 m}\). It would be interesting to give an explanation in the framework of \(\mathfrak{sl}_2 \)
  representations for the absence of non-Juddian exceptional solutions when a Juddian solution exists with the same eigenvalue
  for $\e\in \frac12 \Z_{\geq0}$ (Corollary \ref{NoNonJudd}).
\end{rem}
  
Next, by using the a particular element \(\mathbb{K} \) of the universal enveloping algebra \(\mathcal{U}(\mathfrak{sl}_2)\) of \(\mathfrak{sl}_2\) we capture, via the representation \(\varpi_a\), the confluent Heun operators \(\mathcal{H}_1^{\e}(\lambda) \) and \(\mathcal{H}_2^{\e}(\lambda) \), corresponding to the eigenvalue problem of AQRM in the Bargmann space (see \S \ref{sec:confpict} for the derivation and the explicit form of the operators). Let  $(\alpha, \beta, \gamma, C) \in \R^4$.  Define a second order element ${\mathbb{K}}=\mathbb{K}(\alpha, \beta, \gamma; C) \in \mathcal{U}(\mathfrak{sl}_{2})$ and a constant $\lambda_a=\lambda_a(\alpha, \beta, \gamma)$ depending on the representation $\varpi_a$ as follows:
\begin{align*}
  \mathbb{K}(\alpha, \beta, \gamma; C):= & \left[\frac{1}{2}H-E+\alpha \right]\left(F+\beta\right)
                                           + \gamma\left[H-\frac{1}{2}\right]+C,\\
  \lambda_a(\alpha, \beta, \gamma):= & \beta\left(\frac{1}{2}a +\alpha\right)+\gamma\left(a-\frac{1}{2}\right).
\end{align*}
Noticing $y^{-\frac12(a-\frac12)}\,y\partial_y \,y^{\frac12(a-\frac12)}= y\partial_y+ \frac12(a-\frac12)$, we
obtain the following lemma.
\begin{lem}[\cite{W2016JPA,WY2014JPA}]\label{lem:K-element}
  We have the following expression.
  \begin{align*}
    &\frac{y^{-\frac12(a-\frac12)}\varpi_a(\mathbb{K}(\alpha, \beta, \gamma; C))y^{\frac12(a-\frac12)}}{y(y-1)}\\
    = &\frac{d^2}{dy^2} +\Big\{-\beta + \frac{\frac12a+\alpha}{y} + \frac{\frac12a+2\gamma-\alpha}{y-1} \Big\}\frac{d}{d y}
        +  \frac{-a\beta y+\lambda_a(\alpha, \beta, \gamma)+C}{y(y-1)}. \QEDhere
  \end{align*}
\end{lem}

Now, by choosing suitable parameters $(\alpha, \beta, \gamma; C)$ we define
from ${\mathbb{K}}=\mathbb{K}(\alpha, \beta, \gamma; C)$ two second order elements $\mathcal{K}$ and
$\tilde{\mathcal{K}}$ $\in \mathcal{U}(\mathfrak{sl}_2)$ that capture the Hamiltonian \(\HRabi{\e} \)
of the AQRM. In the following proposition, \(\mathcal{H}_1^{\e}(\lambda)\) (resp. \( \mathcal{H}_2^{\e}(\lambda)\)) is the second order differential operator (confluent Heun ODE) of \eqref{eq:H1eps} (resp. \eqref{eq:H2eps}) corresponding to the solution \( \phi_{1,+}\) (resp. \( \phi_{2,+}\) ) in the system \eqref{eq:system1} (resp. \eqref{eq:system2p}).

\begin{prop} \label{prop:RedEigenProblem}
  Let $\lambda$ be an eigenvalue of $\HRabi{\e}$.
  Set $a=-(\lambda+g^2-\e)$, $a'=a-2\e+1$ and \(\mu = (\lambda + g^2)^2 -4g^2 (\lambda + g^2) - \Delta^2 \).
  \begin{enumerate}[$(1)$]
  \item Define
\begin{align*}
  {\mathcal{K}} &:= \mathbb{K}\Big(1+\frac a2, 4g^2, \frac{a'}2\,;\, \mu+4\e g^2 -\e^2\Big) \in \mathcal{U}(\mathfrak{sl}_2), \\
  \Lambda_a &:= \lambda_a\Big(1 + \frac a2, 4g^2, \frac{a'}2\Big).
\end{align*}
Then
\begin{equation}
  y(y-1)\mathcal{H}_1^{\e}(\lambda)= y^{-\frac12(a-\frac12)}(\varpi_a(\mathcal{K})-\Lambda_a)y^{\frac12(a-\frac12)}.
\end{equation}
\item 
Define
\begin{align*}
  \tilde{\mathcal{K}} &:= \mathbb{K}\Big(-1+\frac{a'}2, 4g^2, \frac a2\,;\, \mu-4\e g^2 -\e^2\Big) \in \mathcal{U}(\mathfrak{sl}_2), \\
  \tilde{\Lambda}_{a'} &:= \lambda_{a'}\Big(-1+\frac{a'}2, 4g^2, \frac a2\Big).
\end{align*}
Then
\begin{equation}
  y(y-1)\mathcal{H}_2^{\e}(\lambda)= y^{-\frac12(a'-\frac12)}(\varpi_{a'}(\tilde{\mathcal{K}})-\tilde{\Lambda}_{a'})y^{\frac12(a'-\frac12)}. \QEDhere
\end{equation}
\end{enumerate}
\end{prop}

\subsection{Juddian solutions and finite dimensional representations}
\label{sec:judd-solut-finite}

  In this subsection,  we recall briefly from \cite{W2016JPA} the fact that Juddian solutions are obtained in the finite dimensional irreducible representation of $\mathfrak{sl}_2$. In other words, we explain that the constraint polynomials are derived from the continuants of matrices describing the coefficients of certain element in the finite dimensional representation \(\mathbf{F}_{m}\).
 
For instance, let us consider the (spherical) case $a=-(\lambda+g^2-\epsilon)=-2m\; (m\in \Z_{>0})$. Namely, for $\mathcal{K}$ and $\Lambda_a$ in Proposition \ref{prop:RedEigenProblem} we take 
\begin{equation*}
\begin{cases}
1 + \frac{a}2= 1 -  \frac{\lambda+g^2-\epsilon}2=1-m,\\
\frac{a'}2=\frac12- \frac{\lambda+g^2+\epsilon}2=\frac12-m-\epsilon.
\end{cases} 
\end{equation*}
The following fact can be directly checked (Lemma 5.1 in \cite{W2016JPA}). 
\smallskip

\noindent
{\bf Fact}:  Write  $v^+=v^+_{2m+1}:= \sum_{n=-m}^{m} a_n e_{1,n} (\in \mathbf{F}_{2m+1}$). 
Then the eigenvalue equation 
$$\varpi_{1,-2m}(\mathcal{K})v^+=\Lambda_{-2m}v^+$$ 
is equivalent to the following
$$
\beta_n^+ a_{n+1} +\alpha_n^+ a_n + \gamma_n^+ a_{n-1}=0,
$$
where
\begin{equation*}
\begin{cases}
\alpha_n^+= (m-n)^2-4g^2(m-n)-\Delta^2+2\epsilon(m-n),\\
\beta_n^+=(m+n+1)(m-n-1),\\
\gamma_n^+=4g^2(m-n+1).
\end{cases}
\end{equation*}

We define a matrix $M^{(2m,\epsilon)}_k=M_k^{(2m,\epsilon)}((2g)^2), \Delta)$ $(k=0,1,2,\cdots,2m)$ by 
\begin{align*}
M^{(2m,\epsilon)}_k=
\begin{bmatrix}
\alpha_{m}^+ & \gamma_{m}^+ &0&0& \cdots& \cdots & 0 &0 \\
\beta_{m-1}^+& \alpha_{m-1}^+ & \gamma_{m-1}^+ &0& \cdots&\cdots & \cdot & \cdot \\
0& \beta_{m-2}^+ & \alpha_{m-2}^+ & \gamma_{m-2}^+ &0& \cdots &\cdot & \cdot \\
\cdot& \cdots&\ddots& \ddots & \ddots& \cdots & 0& \cdot\\
\cdot& \cdots&\cdots& \ddots & \ddots& \gamma_{m-k+3}^+& 0& \cdot\\
\cdot& \cdots&\cdots& 0& \beta_{m-k+2}^+ & \alpha_{m-k+2}^+ & \gamma_{m-k+2}^+ & 0\\
0 & \cdots&\cdots& \cdots & 0& \beta_{m-k+1}^+ & \alpha_{m-k+1}^+ & \gamma_{m-k+1}^+ \\
0 & \cdots &\cdots& \cdots & 0 & 0 & \beta_{m-k}^+ & \alpha_{m-k}^+
\end{bmatrix}.
\end{align*}
The continuant $\{\det M^{(2m,\epsilon)}_k\}_{0\leq k\leq 2m}$ 
is seen to satisfy the recurrence relation
\begin{align}\label{continuant}
\det M^{(2m,\epsilon)}_k 
=  \alpha_{m-k}^+ \det M^{(2m,\epsilon)}_{k-1}-\gamma_{m-k+1}^+\beta_{m-k}^+\det M^{(2m,\epsilon)}_{k-2}
\end{align}
with initial values 
\begin{equation*}
\begin{cases}
\det M^{(2m,\epsilon)}_{0}=\alpha_m^+=-\Delta^2, \\
\det M^{(2m,\epsilon)}_{1}=\alpha_m^+\alpha_{m-1}^+=-\Delta^2(1-4g^2-\Delta^2+2\epsilon).
\end{cases}
\end{equation*}
Since 
\begin{equation*}
\begin{cases}
\alpha_{m-k}^+=k^2-4g^2k-\Delta^2+2\epsilon k,\\
\beta_{m-k}^+=(2m-k+1)(k-1) 
\gamma_{m-k+1}^+=4kg^2,
\end{cases}
\end{equation*}
we find that the following relation holds by comparing the recurrence equation of the continuant \eqref{continuant} having the above initial values with that of the constraint polynomials in Definition \ref{def:cp}.
$$
P^{(2m,\epsilon)}_k((2g)^2, \Delta^2))= (-1)^k \det M_{k}^{(2m,\epsilon)}((2g)^2), \Delta)/(-\Delta^2).
$$
It follows particularly that 
\begin{align*}
P^{(2m,\epsilon)}_{2m}((2g)^2, \Delta^2))=0 
& \Longleftrightarrow  
\det M_{2m}^{(2m,\epsilon)}((2g)^2)=0\\
& \Longleftrightarrow  
(\varpi_{1,-2m}(\mathcal{K})-\Lambda_{-2m})v^+_{2m+1}=0 
\, (\exists v^+_{2m+1} \not=0 \in \mathbf{F}_{2m+1}).
\end{align*}
Hence, defining \(\phi_{1,+}(y)\) by the equation \(y^{-m-\frac14}\phi_{1,+}(y)=:v^+_{2m+1}\), we obtain a Juddian solution as follows. 
 \[
    \phi_{1,+}(y) = \frac{4 g^2 K^{(N,\e)}_{2m-1}}{\Delta} y^{2m}   - \Delta \sum_{n=0}^{2m-1} \frac{K^{(N,\e)}_n}{n-2m} y^n.
  \]
This shows that the Juddian solutions of \S \ref{sec:smallexp} can be constructed by 
  the corresponding eigenvectors \(v^+_{2m+1}\) captured in the finite dimensional irreducible submodules
  \(\rF_{2m+1}\) of \(\varpi_{1,-2m}\). The remaining three cases, non-spherical $\varpi_{2,1-2m}(\mathcal{K})$-eigenproblem in $\mathbf{F}_{2m}$ $(\lambda=2m-1-g^2+\epsilon)$, non-spherical $\varpi_{2,1-2m}(\tilde{\mathcal{K}})$-eigenproblem in $\mathbf{F}_{2m}$ $(\lambda=2m-g^2-\epsilon)$, and spherical $\varpi_{1,-2m}(\tilde{\mathcal{K}})$-eigenproblem in $\mathbf{F}_{2m+1}$ $(\lambda=2m+1-g^2-\epsilon)$ are similarly obtained and we direct the reader to \cite{W2016JPA} for the full derivation.

\subsection{Non-Juddian solutions and the lowest weight representations of $\mathfrak{sl}_2$} 
\label{sec:exceptandrepn}

In this subsection, we describe the way that the exceptional solutions of the AQRM can be captured in
the representation theoretical picture of the AQRM (cf. \S \ref{sec:confpict}). In fact, we show that the non-Juddian
exceptional eigenstates are captured by a pair of irreducible lowest weight representations of $\mathfrak{sl}_2$.

We begin with the solution corresponding to the larger exponent from \S \ref{sec:LargestExponent}.
For \( N = 2 m \), define
\begin{align}\label{eq:defvecD}
  v^+_{\phi_1} &:= y^{\frac{1}{2}(a - \frac{1}{2})} \phi_{1,+}(y) = y^{-m - \frac{1}{4}} \phi_{1,+}(y) = \frac{2 m+1}{\Delta}\bar{K}^{(N,\e)}_{2 m+1} e_{1,m}  - \Delta \sum_{n=2 m+1}^\infty \frac{\bar{K}^{(N,\e)}_n}{n - 2 m}  e_{1,n-m},
\end{align}
where \(\phi_{1,+}\) is the solution \eqref{eq:largexpsol+}.

Next, we continue the discussion following Proposition 7.2 of \cite{W2016JPA}.
Namely, we claim that the vector \(v^+_{\phi_1}\) is a non-zero eigenvector corresponding to the eigenvalue problem
\begin{equation}
  \label{eq:Eigen$-2m$}
  \varpi_{1,-2m}(\mathcal{K})v^+_{\phi_1} = \Lambda_{-2m} v^+_{\phi_1}.
\end{equation}
To see this, it is enough to compute the recurrence relation satisfied by the solution of the eigenvalue problem (the computation
follows like in Section 5.1 of \cite{W2016JPA}). Concretely, let \( v = \sum_{n \in \Z} a_n e_{1,n} \) be a solution of
\(\left(\varpi_{1,-2m}(\mathcal{K}) - \Lambda_{-2m}\right) v=0\), then from the definition of the representation \(\varpi_{1,-2m}\) the
coefficients \(\{a_n\}_{n \in \Z} \) must satisfy
\begin{align*}
  (m+n+1) (m-n-1) a_{n+1} &+ \left( (m-n)^2 - 4g^2 (m-n) + 2\e (m-n) - \Delta^2 \right)a_n \\
  &\qquad \qquad \qquad \qquad \qquad + 4g^2 (m-n+1) a_{n-1} = 0.
\end{align*}
By shifting the index \(n\) by \(m\), and relabeling the equation becomes
\begin{align}
  \label{eq:recurr1}
  (2m+n+1)(n+1) a_{n+1} +\left( -n^2 - 4g^2 n + 2\e n + \Delta^2 \right) a_n - 4g^2 (n-1) a_{n-1} = 0.
\end{align}
Note that the coefficient of \( v^+_{\phi_1} \) corresponding to the basis vector \( e_{1,m+n} \) \((n \in \Z_{\geq 0}) \) is
\( -\Delta (\bar{K}^{(N,\e)}_{2m + n})/n \), therefore by plugging these coefficients into \eqref{eq:recurr1} we get
\begin{equation*}
  (2m+n+1) \bar{K}^{(N,\e)}_{2m+n+1} +\left( -n - 4g^2 + 2\e + \frac{\Delta^2}{n} \right)\bar{K}^{(N,\e)}_{2m+n}
   - 4g^2  \bar{K}^{(N,\e)}_{2m+n-1} = 0,
\end{equation*}
which is equivalent to recurrence \eqref{eq:recurrKn}, thus proving the claim.

Recall that there is an intertwining operator \(A_a\) between the representations \((\varpi_{1,a},\rV_{1,a})\), and \((\varpi_{1,2-a},\rV_{1,2-a}) \)
for \( a \not\in 2 \Z\) (see \cite{W2016JPA}). The isomorphism  \(A_a : \rV_{1,a} \to \rV_{1,2-a}\) $(\rV_{1,a}=\rV_{1,2-a}=\rV_1)$
is explicitly given with respect to the basis $\{e_{1,n}\}_{n\in \Z}$) by the diagonal matrix
\[
  A_a= \Diag(\cdots,c_{-n},\cdots,c_0,\cdots,c_n,\cdots),
\]
with \(c_0 \neq 0\) and
\[
  c_n =  \big(A_a\big)_n= c_0 \prod_{k=1}^{|n|} \frac{k-\frac{a}{2}}{k-1+\frac{a}{2}}.
\]

Recall from Lemma \ref{lem:reducible}, for \(a = -2m \) \((m \in \Z_{>0}) \), there is an isomorphism
\begin{equation}
  \label{eq:intertwine1}
  \rV_{1,-2 m}/\rF_{2m+1} \simeq \rD^-_{2(m+1)} \oplus \rD^+_{2(m+1)} \subset
   \rV_{1, 2(m+1)}.
\end{equation}
In fact, from the expression of the intertwiner \(A_a\) \( (a \not\in 2 \Z)\), we can construct the linear isomorphism
\(\tilde{A}_{-2m}\) of \eqref{eq:intertwine1} by defining
\[
  \tilde{A}_{-2m} := \frac{1}{4\pi}\lim_{a \to -2m} \sin(2\pi a)  A_a,
\]
multiplication being elementwise. Then, as we may take $c_0=1$, we have
\begin{equation}
  \big(\tilde{A}_{-2m}\big)_n=
  \begin{cases}
    (2m+1)\prod_{k=1, \, k\not=m+1}^{|n|} \frac{k+m}{k-m-1} &  \text{ if } |n|>m\\
    0  &  \text{ if } |n|\leq m.
\end{cases}
\end{equation}
Since $e_{1,m} \in \Ker\tilde{A}_{-2m}$, we have $\tilde{A}_{-2m} v^+_{\phi_1}  \in \rD^-_{2(m+1)} \oplus \rD^+_{2(m+1)}$.
Hence, it follows from the formula
\begin{align*}
 \tilde{A}_{-2m} v^+_{\phi_1}
 & =  - \Delta \sum_{n=2 m+1}^\infty \frac{\bar{K}^{(N,\e)}_n}{n - 2 m}   \tilde{A}_{-2m}e_{1,n-m}\\
 & = - \Delta    (2m+1)\sum_{n=m+1}^\infty \frac{\bar{K}^{(N,\e)}_{n+m}}{n -  m}  \prod_{k=1, \, k\not=m+1}^{n} \frac{k+m}{k-m-1}e_{1,n}
\end{align*}
that $\tilde{A}_{-2m} v^+_{\phi_1}  \in \rD^+_{2(m+1)}$.

By definition, if \(v \in \rV_{1,-2 m}\) is a solution of \((\varpi_{1,-2m}(\mathcal{K})- \Lambda_{-2m}) v = 0\), then \(\tilde{A}_{-2m}v \in \rV_{1,2(m+1)}\)
satisfies \((\varpi_{1,2(m+1)}(\mathcal{K})- \Lambda_{-2m}) \tilde{A}_{-2m}v = 0\).

The discussion above is summarized in the following theorem.

\begin{thm} \label{thm:eigenproblem}
  Let \(N \in \Z_{\geq 0}\), \(\Delta >0\) and \( T_\e^{(N)}(g,\Delta)\) the constraint $T$-function defined in \S \ref{sec:NonJuddian}.
  If \(g\) is a positive zero of \( T_\e^{(N)}(g,\Delta)\), we have a non-degenerate non-Juddian exceptional eigenvalue \( \lambda = N + \e - g^2 \).
  Furthermore:
  \begin{enumerate}[$(1)$]
  \item If \(N=2m\), let \(v^+_{\phi_1} \in \rV_{1,-2m}\) be as in \eqref{eq:defvecD}.
    Then $w:= \tilde{A}_{-2m} v^+_{\phi_1}$ is a solution to the eigenproblem
    \((\varpi_{1,2(m+1)}(\mathcal{K})- \Lambda_{-2m}) w = 0\)
    and $w \in \rD^+_{2(m+1)}$.
  \item If \(N=2m-1\), let \(v^+_{\phi_1} \in \rV_{2,1-2m}\) be as described above.
    Then $w:= \tilde{A}_{1-2m} v^+_{\phi_1}$ is a solution of the eigenproblem
    \((\varpi_{2,2 m+1}(\mathcal{K})- \Lambda_{1-2m}) w = 0\)
    and $w \in \rD^+_{2 m+1}$.
  \end{enumerate}
  Let \(N \in \Z_{\geq 0}\), \(\Delta >0\) and \( \tilde{T}_\e^{(N)}(g,\Delta)\) the constraint $T$-function defined in \S \ref{sec:NonJuddian}.
  If \(g\) is a zero of \( \tilde{T}_\e^{(N)}(g,\Delta)\), we have a non-degenerate non-Juddian exceptional eigenvalue \( \lambda = N - \e - g^2 \).
  Furthermore:
  \begin{enumerate}[$(1)$]
  \item If \(N=2m\), let \(v^+_{\phi_2} \in \rV_{2,1-2m}\) be as described above.
    Then $w:= \tilde{A}_{1-2m} v^+_{\phi_2}$ is a solution to the eigenproblem
    \( (\varpi_{2,2m+1}(\tilde{\mathcal{K}})- \tilde{\Lambda}_{1-2m}) w = 0 \)
    and $w  \in \rD^+_{2m + 1}$.
  \item If \(N=2m+1\), let \(v^+_{\phi_2} \in \rV_{1,-2 m}\) be as described above.
    Then $w:= \tilde{A}_{-2m} v^+_{\phi_2}$ is a solution of the eigenproblem
    \( (\varpi_{1,2(m+1)}(\tilde{\mathcal{K}})- \tilde{\Lambda}_{-2m}) w = 0 \)
    and \(w \in \rD^+_{2(m+1)}\).
  \end{enumerate}
 \end{thm}

 \begin{proof}
   In the foregoing discussion we proved the case for \( \lambda = N + \e - g^2 \) and \(N= 2m \). The remaining cases are proved in a similar manner,
   so we leave the proof of those cases to the reader (for the computations, see Section 5 of \cite{W2016JPA}). 
 \end{proof}

\begin{rem}
  When $\e=\ell/2\,( \ell\in \Z_{\geq0})$, the relation $\tilde{T}_{\ell/2}^{(N+\ell)}(g,\Delta) =T_{\ell/2}^{(N)}(g,\Delta)$ in Lemma \ref{lem:const-T-rel} guarantees
  the compatibility of the former and latter assertions in Theorem \ref{thm:eigenproblem}. For instance, let us observe the case when $N=2m$ and
  $\ell=2\ell'\,(\ell'\in \Z)$, the remaining cases are verified in the same manner. If there exists a non-Juddian exceptional solution
  $(e^{-gz}\phi_{1,+}(\frac{g+z}{2g}; \ell/2), e^{-gz}\phi_{1,-}(\frac{g+z}{2g}; \ell/2))$ for $\lambda=N +\ell/2-g^2$  then the theorem asserts $T_{\ell/2}^{(N)}(g,\Delta)=0$
  must hold and $w:= \tilde{A}_{-2m} v^+_{\phi_1}\in \rD^+_{2(m+1)}$. Then, since $\tilde{T}_{\ell/2}^{(N+\ell)}(g,\Delta)=0$, the latter assertion of the theorem
  shows that there exists a non-Juddian exceptional solution $(e^{gz}\phi_{2,-}(\frac{g-z}{2g}; -\ell/2), e^{gz}\phi_{2,+}(\frac{g-z}{2g}; -\ell/2))$
  (corresponding to eigenvalue $\lambda=(N+\ell) -\ell/2 -g^2$) and $w':= \tilde{A}_{1-2(m+\ell')} v^+_{\phi_2}\in  \rD^+_{2(m+\ell')+ 1}$, we can verify that $w$
  corresponds to the first and $w'$ (up to a constant) to the second component of the solution
  $(e^{-gz}\phi_{1,+}(\frac{g+z}{2g}; \ell/2), e^{-gz}\phi_{1,-}(\frac{g+z}{2g}; \ell/2))$, and thus, there is no contradiction with the
  non-degeneracy of the non-Juddian exceptional solution even for the case $\e=\ell/2$.
\end{rem}

\begin{rem}
   The eigenvector corresponding to the non-Juddian exceptional solutions corresponding to $N=0$
   in the proof of Corollary \ref{DegenerateStructure} is captured in the limit of discrete series $\rD_1^+$.
\end{rem}

\begin{rem}
  For non-Juddian exceptional eigenvalues \(  \lambda = N \pm \e -g^2\), Theorem \ref{thm:eigenproblem} describes how  the eigenvectors
  \(w \in (\varpi_{j,a}, \rV_{j}\)) (\(j \in \{1,2\}\)) in the corresponding eigenvalue problem (depending on the sign in front of \(\e\)
  and the parity of \(N\)) can be understood in the representation theoretical scheme (i.e. the setting in Lemma \ref{lem:reducible} and
  Proposition \ref{prop:RedEigenProblem}).

  For the case of the Juddian solutions, we can capture any corresponding eigenvector in a  (finite dimensional) irreducible subspace of
  $\rV_{j}$ and such irreducible representation is uniquely determined by the eigenvector \cite{W2016JPA}. Unlike the case of Juddian solutions,
  any non-Juddian exceptional solution can not be captured in an invariant subspace  $\rV_{j}$ in an appropriate manner but it is required to
  consider  the subquotient of $\rV_{j}$.  Theorem \ref{thm:eigenproblem} shows that through the isomorphism obtained by the operator
  \(\tilde{A}_{-m}\) (the intertwiner between the subquotient $\rV_{j}/\rF_{m} \cong   \rD^{-}_{m+1}\oplus \rD^{+}_{m+1}$), the eigenvector corresponding
  to a non-Juddian exceptional solution determines uniquely an (infinite dimensional) irreducible representation (i.e. a piece corresponding
  to the lowest weight representation $\rD^{+}_{m}$) of the subquotient space of $\rV_{j}$ which contains the eigenvector.
  Since each of the following short exact sequences
  \begin{align*}\label{non-split_SES}
    & 0 \longrightarrow  \rD^{+}_{2m}\oplus \rD^{-}_{2m}  \longrightarrow \rV_{1,2m}  \longrightarrow \rF_{2m-1}  \longrightarrow 0,\\
    & 0 \longrightarrow  \rD^{+}_{2m+1}\oplus \rD^{-}_{2m+1}  \longrightarrow \rV_{2,2m+1}  \longrightarrow \rF_{2m}  \longrightarrow 0
  \end{align*}
  of $\mathfrak{sl}_2$-modules for $m>0$ are not split, we might give a representation theoretic explanation of Corollary \ref{DegenerateStructure} for $\e\in\frac12\Z_{\geq0}$ (i.e. Juddian and non-Juddian exceptional solutions  corresponding to the same eigenvalue cannot appear simultaneously in the spectrum).

\subsection{Regular eigenvalues and classification of representations’ type} \label{sec:regular}

In this subsection, we remark that a regular eigenvalue arising from a solution of the equation $G_\epsilon(x;g, \Delta)=0$ may be considered to be an element of the (so-called non-unitary) principal series $\mathbf{V}_{j,a}\, (j=1,2,\, a\not \in \Z)$. However, different from the cases of exceptional eigenvalues, the remark here is just a sort of formal observation. Moreover, in the last we summarize the relation between the spectrum of the AQRM and irreducible representations of $\mathfrak{sl}_2$. 

Let $\lambda$ be a regular eigenvalue, that is, an eigenvalue $\lambda= x-g^2$ with  $x \pm \epsilon \notin \mathbb{Z}$. 
Then, $x=\lambda+g^2$ is a solution of $G_\epsilon(x;g, \Delta)=0$. From Table \ref{tab:exp} in \S \ref{sec:confpict}, we see that the difference of the characteristic exponent of $\mathcal{H}_1^{\e}(\lambda)\phi_1=0$\, (resp. $\mathcal{H}_2^{\e}(\lambda)\phi_2=0)$ at each regular singular point $y=0, 1$ is respectively $\lambda+g^2-\epsilon, \lambda+g^2+\epsilon+1$ (resp. $\lambda+g^2+\epsilon+1, \lambda+g^2-\epsilon$) and is not an integer. 
We know that, for instance, the confluent Heun equation $\mathcal{H}_1^{\e}(\lambda)\phi_1=0$ 
has a local Frobenius solution around $y=0$, which is known \cite{Heun2008} as the confluent Heun function 
\[
  \phi_1(y)=HC(-4g^2,-(\lambda+g^2-\epsilon), -(\lambda+g^2-\epsilon+1), \cdot\,,\cdot\,; y)= \sum_{n=0}^\infty h_n y^n,
\]
where the coefficients $h_n$ are defined by a certain three-term recurrence relation determined by the constants (i.e. coefficients) in the operator $\mathcal{H}_1^{\e}(\lambda)$, with $h_0=1$ and $h_{-1}=0$. Another linearly independent solution is given in the form
\[
  \tilde\phi_1(y)=y^{\lambda+g^2-\epsilon}HC(-4g^2, \lambda+g^2-\epsilon, -(\lambda+g^2-\epsilon+1), \cdot\,,\cdot\,; y).
\]
We remark that it is equivalently given in \cite{Heun2008} by an asymptotic series in  $y^{n}$ for $n\in \Z$. If the value $\lambda+g^2-\epsilon$ is not a nonnegative integer, the solution $\tilde\phi_1(y)$ can not be analytic (i.e. a non-physical solution). In other words, a possible solution is $\phi_1(y)$. Hence the corresponding eigenvector $v_1(\not=0)$ is obtained from $\phi_1(y)$ as of the form 
$$
v_1=y^{\frac12(-(\lambda+g^2-\epsilon))-\frac12)}\sum_{n=0}^\infty h_n y^n = 
y^{-\frac12(\lambda+g^2-\epsilon)}\sum_{n=0}^\infty h_ne_{1,n}.
$$
Now, for $\alpha=-\frac12(\lambda+g^2-\epsilon)\not\in \mathbb{Z}$, formally we have 
$$
y^\alpha= \sum_{n\in \mathbb{Z}} a_n y^n \quad\text{with}\quad  
a_n:= \frac{\sin(\pi i (\alpha-n))}{\pi(\alpha-n)} \not\equiv 0.  
$$
Therefore, we observe the eigenvector $v_1$ can be obtained in the following form.  
$$
v_1= \sum_{m\in \mathbb{Z}} \Big\{\sum_{k+n=N} a_kh_n\Big\} e_{1,m} \in \mathbf{V}_{1,-(\lambda+g^2-\epsilon)}.
$$
We note that the principal series $(\varpi_{1,-(\lambda+g^2-\epsilon)}, \mathbf{V}_{1,-(\lambda+g^2-\epsilon)})$ is irreducible
when $\lambda+g^2-\epsilon\not\in 2\mathbb{Z}$. (In particular, if the $\phi_1$ above gives a solution then corresponding $v_1$ is spherical.) It follows that the eigenvector corresponding to the regular spectrum determines the irreducible principal series representation and vice-versa. 

\begin{rem}
The sphericality (or non-sphericality) of the eigenvector is depends on which of $\phi_1$ and $\phi_2$ gives a solution (is determined by $x=\lambda+g^2$).
 
\end{rem}


  \ctable[
  caption = {Correspondence of spectrum, irreducible representations and constraint relations},
  label   = tab:spect-constrel,
  width   = \hsize,
  pos     = ht,
  left
  ]{>{\raggedright}X c c >{\raggedright}X   l}{
    \tnote{Determination of $m$ in \(\rF_{m}\). Case $N+\e$: $m=N+1$. Case $N-\e$: $m=N$.}
    \tnote[b]{Determination of $m$ in \(\rD_{m}^+\) for the first component of the solution (see Theorem \ref{thm:eigenproblem}). Case $N+\e$: $m = N+2$. Case $N-\e: m= N+1$.}
    \tnote[c]{Case \(\e \in \frac12\Z\): Non-Juddian exceptional solutions are non-degenerate. Juddian solutions are always degenerate (e.g. corresponding to
      the space \(\rF_m \oplus \rF_{m+1}\) when \(\e=0\). See \cite{W2016JPA} for general $\e=\ell/2$).}
    \tnote[d]{Constants: \(N \in \Z_{\geq 0}, \quad x \not\in \Z_{\geq 0}\).}
  }{                                                         \FL
    \textbf{Type }   & \textbf{Eigenvalue}\tmark[d]  &   & \textbf{Rep. of \(\mathfrak{sl}_2\)} & \textbf{Constraint relation}  \NN
    \cmidrule(r){1-4}\cmidrule(l){5-5}
    Juddian\tmark[c]& \( N\pm\e-g^2\)  & \(\Leftrightarrow\)  & \(\rF_{m}\): finite dim. irred. rep.\tmark & \(P_N^{(N, \pm\e)}((2g)^2, \Delta^2)=0\) \NN
    Non-Juddian exceptional & \( N\pm\e-g^2\) & \(\Leftrightarrow\)  & \(\rD_{m}^+\): irred. lowest weight rep.\tmark[b] & \(T_{\pm \e}^{(N)}(g,\Delta) = 0\)  \NN
    \rule{0pt}{3ex}Regular & \( x\pm\e-g^2\)  & \(\Leftrightarrow\)  & \(\varpi_{j,a}\): irred. principal series &  \(G_\e(x;g,\Delta)=0\)\LL
  }

\end{rem}

\begin{rem}
  We may ask if there is a geometric (group theoretic) interpretation of the spectrum of the AQRM. In fact, since the Lie algebra of any
  covering group of $SL_2(\R)$ is $\mathfrak{sl}_2(\R)$ (and $\mathfrak{sl}_2(\R)_\C=\mathfrak{sl}_2(\C)$), relating to the spectral zeta function
  the following question comes up naturally (cf. Table \ref{tab:spect-constrel}): Is there any covering group $G:=G(g,\Delta)$ of $SL_2(\R)$ (or
  $G:=SL_2(\C)$) with a discrete subgroup $\Gamma:=\Gamma(g,\Delta)$ of $G$ such that, e.g. the regular spectrum of $\HRabi{\e}$ can be captured in
  $L^2(\Gamma\backslash G)$?  (see Problem  6.1 in \cite{W2016JPA} for a relevant question).
\end{rem}

    We summarize the relation of eigenvalue type, constraint relation and related irreducible representations in Table \ref{tab:spect-constrel}.

\section{Generating functions of the constraint polynomials}\label{sec:further}

In this section, we study the generating functions of constraint polynomials $P_N^{(N,\e)}((2g)^2,\Delta^2)$ with their defining
sequence $P_k^{(N,\e)}((2g)^2,\Delta^2)$ from the viewpoints of confluent Heun equations. As we have already observed (cf. Proposition
\ref{Prop:reccurEquiv}), $P_N^{(N,\e)}((2g)^2,\Delta^2)$ is essentially identified with the coefficient of the $\log$-term of the local Frobenius
solution of the smallest exponent $ \rho^{\pm}= 0$ for the confluent Heun equation corresponding to the eigenvalue problem of the AQRM. As a byproduct of
this discussion, we obtain an alternative proof of the divisibility part of Theorem \ref{thm:Main}. 

\subsection{Constraint polynomials and confluent Heun differential equations} \label{sec:div}

In this subsection we study certain confluent Heun differential equations satisfied
by the generating function of the polynomials \(\cp{N,\e}{k}(x,y)\) and obtain
combinatorial relations between the constraint polynomials.

For convenience, define the normalized polynomials \(\ncp{N,\e} k(x,y)\) by
\begin{equation*}
  \ncp{N,\e} k(x,y):=\frac{\cp{N,\e}k(x,y)}{k!(k+1)!}
\end{equation*}
and their generating function
\begin{equation*}
  \Gncp{N,\e}(x,y,t):=\sum_{k=0}^\infty \ncp{N,\e}k(x,y)t^k.
\end{equation*}

Clearly, the normalized polynomials \(\ncp{N,\e}{k}(x,y)\) satisfy the recurrence relation
\begin{equation}\label{eq:recurrence of ncp}
  \begin{split}
    \ncp{N,\e}0(x,y)&=1,\qquad
    \ncp{N,\e}1(x,y)=\frac{x+y-1-2\e}2, \\
    \ncp{N,\e}k(x,y)&=\frac{kx+y-k^2-2k\e}{k(k+1)}\ncp{N,\e}{k-1}(x,y)
    -\frac{(N-k+1)x}{k(k+1)}\ncp{N,\e}{k-2}(x,y),.
  \end{split}
\end{equation}
for \(k\geq 2 \). With these preparations, we study the differential equation satisfied by
the generating function \(\Gncp{N,\e}(x,y,t)\).

\begin{prop}\label{lem:diff eq of z}
  As a function in the variable $t$, the generating function $z=\Gncp{N,\e}(x,y,t)$ satisfies the differential equation
  \begin{equation} \label{eq:cheundiff}
    \biggl\{t(1+t)\frac{\partial^2}{\partial t^2}
    +(2-(x-3-2\e)t-xt^2)\frac{\partial}{\partial t}
    -(x+y-1-2\e-(N-1)xt)\biggr\}z=0.
  \end{equation}
\end{prop}

\begin{proof}
  We observe that
  \small
  \begin{align*}
    &\frac{\partial^2}{\partial t^2}(=\sum_{k=1}^\infty k(k+1)\ncp{N,\e}k(x,y)t^{k-1} \\
    &\, =2\ncp{N,\e}1(x,y)+\sum_{k=2}^\infty
      \biggl(
      (kx+y-k^2-2k\e)\ncp{N,\e}{k-1}(x,y)
      -(N-k+1)x\ncp{N,\e}{k-2}(x,y)
      \biggr)t^{k-1} \\
    &\, =2\ncp{N,\e}1(x,y)
      +y\sum_{k=2}^\infty
      \ncp{N,\e}{k-1}(x,y)t^{k-1}
      +(x+1-2\e)\sum_{k=2}^\infty k\ncp{N,\e}{k-1}(x,y)t^{k-1} \\
    &\, \quad{}-\sum_{k=2}^\infty k(k+1)\ncp{N,\e}{k-1}(x,y)t^{k-1}
      -(N+1)x\sum_{k=2}^\infty\ncp{N,\e}{k-2}(x,y)t^{k-1}
      +x\sum_{k=2}^\infty k\ncp{N,\e}{k-2}(x,y)t^{k-1} \\
    &\,=2\ncp{N,\e}1(x,y)+y(z-1)+(x+1-2\e)\frac{\partial}{\partial t}(tz-t)
      -\frac{\partial^2}{\partial t^2}(t^2z-t^2)-(N+1)xtz+x\frac{\partial}{\partial t}(t^2z) \\
    &\,=yz+(x+1-2\e)\Bigl(z+t\frac{\partial z}{\partial t}\Bigr)
      -\Bigl(2z+4t\frac{\partial z}{\partial t}+t^2\frac{\partial^2 z}{\partial t^2}\Bigr)
      -(N+1)xtz+x\Bigl(2tz+t^2\frac{\partial z}{\partial t}\Bigr).
  \end{align*}
  \normalsize
  Hence, it follows that
  \begin{align*}
    t\frac{\partial^2z}{\partial t^2}+2\frac{\partial z}{\partial t} =& yz+(x+1-2\e)\Bigl(z+t\frac{\partial z}{\partial t}\Bigr)
      -\Bigl(2z+4t\frac{\partial z}{\partial t}+t^2\frac{\partial^2 z}{\partial t^2}\Bigr) \\                                       
    &\, -(N+1)xtz+x\Bigl(2tz+t^2\frac{\partial z}{\partial t}\Bigr).
  \end{align*}
  Since $\frac{\partial^2}{\partial t^2}(tz)=t\frac{\partial^2z}{\partial t^2}+2\frac{\partial z}{\partial t}$,
  we have the desired conclusion immediately. 
\end{proof}

\begin{rem}
  Note that equation \eqref{eq:cheundiff} is a confluent Heun differential equation.
  By \eqref{eq:coeffGfunct}, we know that the constraint polynomials \(\cp{N,\e}{n}(x,y)\) are
  the constant multiples of the coefficients \(K_n^{\pm}(N \pm \e;g,\Delta,\e)\) of solutions of the
  confluent Heun picture of the spectral problem of AQRM.
  Therefore, it is not surprising that \(\Gncp{N,\e}(x,y,t)\) satisfies the confluent Heun differential
  equation \eqref{eq:cheundiff}.
\end{rem}

Our main application is the following combinatorial relation between the polynomials \(\ncp{N,\e}{k}(x,y)\).

\begin{thm}
  \label{thm:mainConj}
  For $N, \ell, k \in \Z_{\geq0}$, the following equation holds.
  \begin{equation}\label{eq:mainConj}
    \ncp{N+\ell,-\ell/2}k(x,y)=\sum_{i=0}^k \binom{\ell}{k-i} \ncp{N,\ell/2}i(x,y).
  \end{equation}
\end{thm}

We illustrate Theorem \ref{thm:mainConj} with an example, which is used later in the proof of the theorem.

\begin{ex}\label{ex:k=0,1}
  When $k=0$, both sides of \eqref{eq:mainConj} are equal to $1$.
  When $k=1$, the left-hand side of \eqref{eq:mainConj} is $\ncp{N+\ell,-\ell/2}1(x,y)=\frac{x+y-1+\ell}2$,
  and the right-hand side of \eqref{eq:mainConj} is
\begin{equation*}
  \binom \ell 1\ncp{N,\ell/2}0(x,y)+\binom \ell 0\ncp{N,\ell/2}1(x,y) =\ell+\frac{x+y-1-\ell}2 =\frac{x+y-1+\ell}2.
\end{equation*}
Hence \eqref{eq:mainConj} is valid when $k=0,1$.
\end{ex}

Instead of proving Theorem \ref{thm:mainConj} directly, we use
the following equivalent formulation in terms of generating functions.

\begin{lem}\label{lem:reduction to GF}
  The equation \eqref{eq:mainConj} is equivalent to the equation
  \begin{equation}
    \Gncp{N+\ell,-\ell/2}(x,y,t)=(1+t)^\ell\Gncp{N,\ell/2}(x,y,t). \label{eq:reduction to GF}
  \end{equation}
\end{lem}

\begin{proof}
  We see that
  \begin{align*}
    (1+t)^{\ell}  \Gncp{N,\ell/2}(x,y,t) &=\sum_{j=0}^\infty\binom{\ell}jt^j\sum_{i=0}^\infty \ncp{N,\ell/2}i(x,y)t^i =\sum_{k=0}^\infty\sum_{i+j=k} \binom{\ell}j\ncp{N,\ell/2}i(x,y)t^k \\
      &=\sum_{k=0}^\infty \biggl\{\sum_{i=0}^k \binom{\ell}{k-i}\ncp{N,\ell/2}i(x,y)\biggr\}t^k,
  \end{align*}
  from which the lemma follows. 
\end{proof}

The proof of the equivalent statement is done using the differential equation \eqref{eq:cheundiff} satisfied
by the generating function \(\Gncp{N,\e}(x,y,t)\). We need the followings lemmas.

\begin{lem}\label{lem:diff eq of w}
  Let $p(t)$, $q(t)$ be polynomials in $t$ and $\alpha\in\R$.
  Suppose that a function $z=z(t)$ satisfies the differential equation
  \begin{equation*}
    t(1+t)\frac{\partial^2z}{\partial t^2}+p(t)\frac{\partial z}{\partial t}+q(t)z=0.
  \end{equation*}
  Then $w=(1+t)^\alpha z$ satisfies
\begin{equation*}
  t(1+t)\frac{\partial^2w}{\partial t^2}+(p(t)-2\alpha t)\frac{\partial w}{\partial t}+(q(t)-\alpha u(t))w=0,
\end{equation*}
where $u(t)=\frac{p(t)-(\alpha+1)t}{1+t}$.
\end{lem}

\begin{proof}
Since
\begin{align*}
  \frac{\partial w}{\partial t} &= \alpha(1+t)^{\alpha-1}z+(1+t)^{\alpha}\frac{\partial z}{\partial t}, \\
  \frac{\partial^2 w}{\partial t^2} &= \alpha(\alpha-1)(1+t)^{\alpha-2}z+2\alpha(1+t)^{\alpha-1}\frac{\partial z}{\partial t}+(1+t)^{\alpha}\frac{\partial^2 z}{\partial t^2},
\end{align*}
we have
\begin{align*}
  t(1+t)\frac{\partial^2w}{\partial t^2}+p(t)\frac{\partial w}{\partial t}  &=(1+t)^\alpha t(1+t)\frac{\partial^2z}{\partial t^2} +2\alpha(1+t)^\alpha t\frac{\partial z}{\partial t} \\
  &\quad +\alpha(\alpha-1)(1+t)^{\alpha-1}tz
    +(1+t)^{\alpha}p(t)\frac{\partial z}{\partial t} +\alpha(1+t)^{\alpha-1}p(t)z \\
  &=-q(t)w +2\alpha t\Bigl(\frac{\partial w}{\partial t}-\alpha(1+t)^{\alpha-1}z\Bigr) +\alpha(\alpha-1)(1+t)^{\alpha-1}tz \\
  &\quad +\alpha(1+t)^{\alpha-1}p(t)z  \\
  &=-q(t)w +2\alpha t\frac{\partial w}{\partial t} +\alpha u(t)w,
\end{align*}
from which the conclusion follows. 
\end{proof}

\begin{lem}\label{lem:z and w}
  The function $w=(1+t)^{2\e}z$, where $z=\Gncp{N,\e}(x,y,t)$, satisfies
  \begin{equation*}
    \biggl\{t(1+t)\frac{\partial^2}{\partial t^2}
    +(2+(x-3+2\e)t-xt^2)\frac{\partial}{\partial t}
    -(x+y-1+2\e-(N+2\e-1)xt)\biggr\}
    w=0.
  \end{equation*}
\end{lem}

\begin{proof}
  Let
  \begin{equation*}
    p(t)=2-(x-3-2\e)t-xt^2, \quad
    q(t)=-(x+y-1-2\e)+(N-1)xt, \quad
    \alpha=2\e.
  \end{equation*}
  Then $z$ satisfies the equation
  \begin{equation*}
    t(1+t)\frac{\partial^2z}{\partial t^2}+p(t)\frac{\partial z}{\partial t}+q(t)z=0.
  \end{equation*}
  We see that
  \begin{equation*}
    u(t)=\frac{p(t)-(\alpha+1)t}{1+t}=\frac{2-(x-3-2\e)t-xt^2-(2\e+1)t}{t+1}=2-xt,
  \end{equation*}
  and
  \begin{equation*}
    p(t)-2\alpha t=2-(x-3+2\e)t-xt^2,\quad
    q(t)-\alpha u(t)=-(x+y-1+2\e)+(N+2\e-1)xt.
\end{equation*}
Thus, by Lemma \ref{lem:diff eq of w}, $w$ satisfies the equation
\begin{equation*}
  t(1+t)\frac{\partial^2w}{\partial t^2}+(2-(x-3+2\e)t-xt^2)\frac{\partial w}{\partial t}+(-(x+y-1+2\e)+(N+2\e-1)xt)w=0
\end{equation*}
as desired. 
\end{proof}

\begin{proof}[Proof of Theorem \ref{thm:mainConj}]
  By Proposition \ref{lem:diff eq of z} and Lemma \ref{lem:z and w}, the functions $\Gncp{N+\ell,-\ell/2}(x,y,t)$ and
  $(1+t)^\ell\Gncp{N,\ell/2}(x,y,t)$ satisfy the same second order linear differential equation
  \begin{equation*}
    \biggl\{t(1+t)\frac{\partial^2}{\partial t^2}
    +(2+(x-3+\ell)t-xt^2)\frac{\partial}{\partial t}
    -(x+y-1+\ell-(N+\ell-1)xt)\biggr\}
    z=0.
  \end{equation*}
  By Example \ref{ex:k=0,1}, the constant and linear terms of the power series expansion of these two
  functions at $t=0$ are equal, and hence they are equal. By Lemma \ref{lem:reduction to GF}, this implies \eqref{eq:mainConj}. 
\end{proof}

\subsection{Revisiting divisibility of constraint polynomials} \label{sec:revdiv}

Using the identity of Theorem \ref{thm:mainConj}, we give another proof of the divisibility property of
constraint polynomials stated in Theorem \ref{thm:div1}.

\begin{proof}[Proof of Theorem \ref{thm:div1}]
  From the defining recurrence relations \eqref{eq:recurrence of ncp} of the
  polynomials \(\ncp{N,\e}{k}(x,y)\) we notice that
  \begin{equation*}
    \ncp{N,\e}{N+1}(x,y)=\frac{(N+1)x+y-(N+1)^2-2(N+1)\e}{(N+1)(N+2)}\ncp{N,\e}N(x,y),
  \end{equation*}
  that is, $\ncp{N,\e}{N+1}(x,y)$ is divisible by $\ncp{N,\e}N(x,y)$.
  Hence, using the recurrence \eqref{eq:recurrence of ncp} again,
  we see by induction on $k$ that $\ncp{N,\e}k(x,y)$ is divisible by $\ncp{N,\e}N(x,y)$ if $k\ge N$.

  Next, putting $k=N+\ell$ and $\e=\frac{\ell}{2}$ in \eqref{eq:mainConj}, we have
  \begin{equation}\label{eq:relation between ncps}
    \ncp{N+\ell,-\ell/2}{N+\ell}(x,y)
    =\sum_{i=N}^{N+\ell}\binom{\ell}{N+\ell-i} \ncp{N,\ell/2}i(x,y)
    =\sum_{j=0}^{\ell}\binom{\ell}{j} \ncp{N,\ell/2}{N+j}(x,y).
  \end{equation}
  Since each of the terms in the right-hand side are divisible by $\ncp{N,\ell/2}N(x,y)$,
  this completes the proof. 
\end{proof}

\begin{rem}
Define the polynomials by
$\mathcal{Q}_k^{(N,\e)}(x,y) = \frac{P_{N+k}^{(N,\e)}(x,y)}{P_{N}^{(N,\e)}(x,y)}$.
Then we observe that \eqref{eq:relation between ncps} is equivalent to
\begin{equation}\label{eq:polyAnpolysum}
A_N^\ell(x,y) = \sum_{j=0}^\ell \binom{\ell}{j} \frac{(N+\ell)!(N+\ell+1)!}{(N+j)!(N+j+1)!} \mathcal{Q}_{j}^{(N,\ell/2)}(x,y).
\end{equation}
We notice that for $x<0$ the family of polynomials $\{\mathcal{Q}_k^{(N,\ell/2)}(x,y)\}_{k \geq 0}$ actually forms
an orthogonal system of polynomials (see \cite{C1978}) and
might have a representation theoretic interpretation.
It is therefore desirable to obtain a straightforward  proof of the positivity for the polynomial
\[
  A_N^{\ell}(x,y)= P_{N+\ell}^{(N+\ell,-\ell/2)}(x,y)/P_N^{(N,\ell/2)}(x,y)
\]
for $x, y>0$
in a reasonable framework
(e.g. \cite{FK1994} for the harmonic analysis on symmetric cones) of orthogonal polynomials in two variables.
\end{rem}

\section*{Acknowledgements}
This work was supported by JST CREST Grant Number JPMJCR14D6, Japan, and
by Grand-in-Aid for Scientific Research (C) JP16K05063 of JSPS, Japan.
CRB was supported during the duration of the research by the Japanese
Government (MONBUKAGAKUSHO: MEXT) scholarship.


\begin{flushleft}
\bigskip

Kazufumi Kimoto \par
Department of Mathematical Sciences, \par
University of the Ryukyus \par
1 Senbaru, Nishihara, Okinawa 903-0213 JAPAN \par
\texttt{kimoto@math.u-ryukyu.ac.jp} \par

\bigskip

Cid Reyes-Bustos \par
Institute of Mathematics for Industry,\par
Kyushu University \par
744 Motooka, Nishi-ku, Fukuoka 819-0395 JAPAN \par
\texttt{c-reyes@imi.kyushu-u.ac.jp} 

\medskip

Current address: \\
Department of Mathematical and Computing Science, School of Computing, \par
Tokyo Institute of Technology \par
2 Chome-12-1 Ookayama, Meguro, Tokyo 152-8552 JAPAN \par\par
\texttt{reyes@c.titech.ac.jp}

\bigskip

Masato Wakayama \par
Institute of Mathematics for Industry,\par
Kyushu University \par
744 Motooka, Nishi-ku, Fukuoka 819-0395 JAPAN \par
\texttt{wakayama@imi.kyushu-u.ac.jp}
\end{flushleft}

\end{document}